\documentclass[%
 reprint,
 superscriptaddress,
 amsmath, 
 amssymb,
 amsthm,
 aps,
 pra,
]{revtex4-2}
\usepackage{graphicx}
\usepackage{dcolumn}
\usepackage{bm}

\usepackage{comment}

\usepackage[english]{babel}
\makeatletter
\let\ORIbbl@fixname\bbl@fixname
\def\bbl@fixname#1{%
  \@ifundefined{languagealias@\expandafter\string#1}
    {\ORIbbl@fixname#1}
    {\edef\languagename{\@nameuse{languagealias@#1}}}%
}
\newcommand{\definelanguagealias}[2]{%
  \@namedef{languagealias@#1}{#2}%
}
\makeatother
\definelanguagealias{en}{english}
\definelanguagealias{eng}{english}

\usepackage{enumitem} 
\usepackage{braket} 
\usepackage{mathtools} 
\usepackage[dvipsnames]{xcolor}
\usepackage{comment} 
\usepackage{physics} 

\usepackage[font=small, justification=Justified]{caption}
\usepackage{subcaption}

\newcommand{\mbeq}{\overset{!}{=}}

\usepackage{amsthm}
\newtheorem{theorem}{Theorem}[section]

\newtheorem{proposition}[theorem]{Proposition}
\newtheorem*{remark}{Remark}
\theoremstyle{definition}
\newtheorem{definition}{Definition}[section]

\definecolor{S_Blue}{RGB}{0,135,252}
\definecolor{S_Red}{RGB}{214,13,63}
\definecolor{Blue}{RGB}{47,89,151}
\definecolor{S_Grey}{RGB}{150,150,158}
\definecolor{S_Yel}{RGB}{255,204,0}
\definecolor{S_Green}{RGB}{102,204,0}
\definecolor{S_Brown}{RGB}{154,41,41}

\usepackage{mathrsfs} 

\usepackage{stmaryrd} 

\usepackage{pgffor}
\usepackage{ifthen}
\usepackage{placeins}
\usepackage[unicode=true]{hyperref} 

\begin{document}

\preprint{APS/123-QED}

\title{
Non-classicality at equilibrium and efficient predictions under non-commuting charges 
}

\author{Lodovico Scarpa}
 \email{lodovico.scarpa@physics.ox.ac.uk}
 \affiliation{%
 Clarendon Laboratory, University of Oxford, Parks Road, Oxford OX1 3PU, UK
 }%
 \affiliation{%
 Quantum and Condensed Matter Physics Group (T-4), Theoretical Division, Los Alamos National Laboratory, Los Alamos, New Mexico 87545, USA
 }

 \author{Nishan Ranabhat}
 \affiliation{
 Department of Physics, University of Maryland, Baltimore County, Baltimore, Maryland 21250, USA
 }
  \affiliation{Department of Physics and Astronomy, University of Waterloo, Ontario, N2L 3G1, Canada}

\author{Amit Te'eni}
\affiliation{Faculty of Engineering and the Institute of Nanotechnology and Advanced Materials, Bar-Ilan University, Ramat Gan 5290002, Israel}

 \author{Abdulla Alhajri}
  \affiliation{%
 Clarendon Laboratory, University of Oxford, Parks Road, Oxford OX1 3PU, UK
 }%
 \affiliation{%
 Quantum Research Centre, Technology Innovation Institute, 9639 Abu Dhabi, United Arab Emirates
 }%

 \author{Vlatko Vedral}%
 \affiliation{%
 Clarendon Laboratory, University of Oxford, Parks Road, Oxford OX1 3PU, UK
 }

  \author{Fabio Anza}
 \email{fanza@umbc.edu}
  \affiliation{
 Department of Physics, Cybersecurity Institute and Quantum Science Institute, University of Maryland, Baltimore County, Baltimore, Maryland 21250, USA
 }

 \author{Luis Pedro Garc\'ia-Pintos}%
 \email{lpgp@lanl.gov}
 \affiliation{Quantum and Condensed Matter Physics Group (T-4), Theoretical Division, Los Alamos National Laboratory, Los Alamos, New Mexico 87545, USA
 }

\date{\today}

\begin{abstract}
\noindent A quantum thermodynamic system can conserve non-commuting observables, but the consequences of this phenomenon on relaxation are still not fully understood. We investigate this problem by leveraging an observable-dependent approach to equilibration and thermalization in isolated quantum systems. We extend such approach to scenarios with non-commuting charges, and show that it can accurately estimate the equilibrium distribution of coarse observables without access to the energy eigenvalues and eigenvectors. Our predictions do not require weak coupling and are not restricted to local observables, thus providing an advantage over the non-Abelian thermal state. Within this approach,  weak values and quasiprobability distributions emerge naturally and play a crucial role in characterizing the equilibrium distributions of observables. We show and numerically confirm that, due to charges' non-commutativity, these weak values can be anomalous even at equilibrium, which has been proven to be a proxy for non-classicality. Our work thus uncovers a novel connection between the relaxation of observables under non-commuting charges, weak values, and Kirkwood-Dirac quasiprobability distributions.
\end{abstract}

\maketitle

Conserved quantities, known as charges, can drastically influence the equilibrium properties of a classical system, by leading to charge-dependent generalized Gibbs states~\cite{huang_statistical_1987}. In quantum systems, the influence of conserved quantities is even richer in situations where the charges do not commute with each other~\cite{majidy_noncommuting_2023, hinds_mingo_quantum_2018, yunger_halpern_microcanonical_2016, halpern_beyond_2018, yunger_halpern_noncommuting_2020, kranzl_experimental_2023, noh_eigenstate_2023}. 
A problem of active research is understanding the implications that such non-commuting charges, corresponding to the presence of non-Abelian symmetries, have on equilibration and thermalization processes of quantum systems~\cite{majidy_noncommuting_2023, hinds_mingo_quantum_2018, yunger_halpern_microcanonical_2016, halpern_beyond_2018, yunger_halpern_noncommuting_2020, kranzl_experimental_2023, noh_eigenstate_2023, popescu_quantum_2018, popescu_reference_2020, khanian_resource_2023, lostaglio_thermodynamic_2017, guryanova_thermodynamics_2016}. 

On the one hand, certain results suggest that the presence of non-commuting charges may hinder equilibration and thermalization. The charges' non-commutation implies the Hamiltonian is degenerate~\cite{majidy_noncommuting_2023, yunger_halpern_noncommuting_2020}, which could influence the extent to which isolated quantum systems equilibrate~\cite{reimann_foundation_2008}. Moreover, microcanonical subspaces might not exist~\cite{yunger_halpern_microcanonical_2016} and the Eigenstate Thermalization Hypothesis (ETH) needs modifications~\cite{murthy_non-abelian_2023, lasek_numerical_2024, patil_eigenstate_2025, noh_eigenstate_2023}. Finally, the thermodynamic entropy production rate decreases due to non-commuting charges~\cite{manzano_non-abelian_2022, shahidani_thermodynamic_2022, upadhyaya_non-abelian_2024}.
On the other hand, other results suggest non-commuting charges may enable thermalization in certain scenarios; for instance, by increasing average entanglement entropy~\cite{majidy_non-abelian_2023}, destabilizing many-body localization~\cite{potter_symmetry_2016, protopopov_effect_2017}, or removing dynamical symmetries~\cite{majidy_noncommuting_2024}. We aim to shed light on the role of non-commuting charges in the equilibration and thermalization of isolated quantum systems.

We explore the role of non-commuting charges via the lens of Observable Statistical Mechanics, an observable-specific approach to equilibration and thermalization~\cite{scarpa_observable_2025, anza_information-theoretic_2017, anza_eigenstate_2018}. 
This approach has given successful predictions in non-integrable~\cite{scarpa_observable_2025} and integrable~\cite{scarpa_observable_2025} systems, as well as systems that exhibit many-body localization~\cite{anza_pure_2018}. 
The core idea behind this approach is to recognize that observables play a significant role in equilibration and thermalization. Instead of considering a density matrix, it focuses on the probability distribution of outcomes of an observable and makes use of a maximum entropy principle à la Jaynes~\cite{jaynes_information_1957, jaynes_information_1957-1, jaynes_brandeis_1989}. 

While traditional approaches focused on the equilibration of density matrices of subsystems have been successful in predicting the behavior of local observables in many cases~\cite{gogolin_equilibration_2016}, they ignore non-local observables and fail to describe potential variability in behaviour of different observables. For instance, in many-body localized systems, the local density matrix is not thermal, yet some local observables thermalize while others do not~\cite{anza_pure_2018}. Further, even if two observables equilibrate, their relaxation time scales can be enormously different~\cite{goldstein_time_2013, malabarba_quantum_2014}. 
 Moreover, in state-based approaches, it can be hard to make predictions when the relevant ensembles are not the canonical (Gibbs) or the grand canonical. In fact, while these two latter ensembles only depend on subsystem information, namely, the local Hamiltonian and particle-number operator, the diagonal ensemble~\cite{scarpa_observable_2025, anza_new_2018}, microcanonical ensemble and generalized Gibbs ensemble of integrable systems~\cite{essler_quench_2016} (and others) require microscopic information about the energy to be constructed. In particular, the generalized Gibbs ensembles describing the equilibrium behavior of integrable systems are global states that require knowledge of an extensive number of quasilocal conserved quantities, which include the global Hamiltonian and are algebraically dependent on it~\cite{essler_quench_2016, sutherland_beautiful_2004}.

Ultimately, one of the main goals of the field is to predict physical properties at equilibrium. For instance, fully characterizing the physical conditions that lead to thermal-like behaviour is a major open question. The answer will likely involve the interplay of three quantities: a Hamiltonian, an initial state, and an observable~\cite{reimann_foundation_2008, garcia-pintos_equilibration_2017}. Hence, here we shift the focus from density matrices to observables' probability distributions.

As we will show, 
Observable Statistical Mechanics allows predicting the equilibrium distributions of large classes of observables without knowledge of the energy eigenvalues and eigenstates, even in the presence of non-commuting conserved charges. Importantly, the use of such a framework is justified for large families of observables, namely, observables with few outcomes (`coarse observables')~\cite{scarpa_observable_2025} and observables whose eigenbasis is unbiased with respect to the energy basis (`Hamiltonian Unbiased Observables')~\cite{anza_information-theoretic_2017}. 

Here, we generalize Observable Statistical Mechanics to accommodate for non-commuting conserved quantities. In doing so, we obtain three main results. Firstly, we prove that a Hamiltonian is degenerate if and only if it possesses a non-Abelian symmetry. In other words, it is equivalent to consider degenerate Hamiltonians and Hamiltonians with non-commuting charges. 
Secondly, we show that non-classical behavior emerges \textit{even at equilibrium} via quasiprobability distributions and weak values that characterize observable's equilibrium distributions. Crucially, such non-classicality does not emerge in a system where charges commute. Finally, we show that the framework allows estimating equilibrium distributions of observables more accurately than the non-Abelian thermal state (NATS), for the model and cases considered. Importantly, our predictions do not require detailed information about the energy or the charges, and they do not rely on the assumptions of weak coupling, large system size and locality of observables, in contrast to the NATS. \\

The paper is organized as follows. In Section~\ref{sec:setup}, we introduce the physical setup, the notation, and our first result, which shows that a Hamiltonian is degenerate if and only if its symmetry group is non-Abelian. 
In Section~\ref{sec:generalizing_obs_stat_mech}, we generalize the framework of Observable Statistical Mechanics to accommodate for the presence of non-commuting charges. We also justify the use of a maximum entropy principle for coarse observables in this context, and discuss the limits of validity of this approach. In Section~\ref{sec:interpretation_KDQ_WV}, we uncover a connection to weak values and  Kirkwood-Dirac quasiprobability distributions. This connection allows us to show that in the presence of non-commuting charges, non-classical behavior emerges even at equilibrium. Section~\ref{sec:numerical_results} contains numerical results testing the theory, which include witnessing non-classical effects. We discuss our findings and future work in Section~\ref{sec:discussion}.

\section{Setup and first result} \label{sec:setup}

\subsection{Evolution under a degenerate Hamiltonian}
Consider a normalised pure initial state $\ket{\psi(0)}$ living in a Hilbert space $\mathcal{H}$ of dimension $D := \operatorname{dim}\mathcal{H} < \infty$ and evolving under a degenerate Hamiltonian
\begin{equation}
    H := \sum_{n=1}^{n_E} E_n \Pi_n,
\end{equation}
where $n_E < D$ is the number of distinct eigenenergies and $\sum_n \Pi_n = \mathbb{I}$. The setup identifies a decomposition of the Hilbert space as $\mathcal{H} = \bigoplus_{n=1}^{n_E} \mathcal{H}_n$, where $\mathcal{H}_n$ is the degenerate eigensubspace corresponding to eigenvalue $E_n$. Each of these subspaces has dimension $d_n := \operatorname{dim}\mathcal{H}_n = \Tr(\Pi_n)$. Since the Hamiltonian is degenerate, it does not have a unique complete eigenbasis, and we have the freedom to choose any basis within each of its eigensubspaces. 

The probability of measuring the energy and obtaining outcome $E_n$ is given by $p_n \coloneqq p(E_n) := \langle \psi(0) | \Pi_n | \psi(0) \rangle$. We can write the initial state in the energy eigenbasis as $\ket{\psi(0)} = \sum_n \sqrt{p_n} \ket{\lambda_n}$, where $\{\ket{\lambda_n} := \frac{1}{\sqrt{p_n}} \Pi_n \ket{\psi(0)}\}$ are the normalised projections of the initial state onto the subspaces $\{\mathcal{H}_n\}$. 
Note that $H \ket{\lambda_n} = E_n \ket{\lambda_n}$, but $\{ \ket{\lambda_n} \}_{n=1}^{n_E}$ is not a complete basis. 

The system's density matrix at time $t$ is
\begin{equation}
    \rho(t) := |\psi(t) \rangle\langle \psi(t)| = \sum_{n,m=1}^{n_E}  e^{-i(E_m - E_n)t} \sqrt{p_n p_m} \ketbra{\lambda_m}{\lambda_n}.
\end{equation}
 Recalling that the indices $n$ and $m$ label \textit{distinct} energy eigenvalues, the infinite-time-averaged state is
\begin{equation}
    \rho_\infty \coloneqq \overline{\rho(t)}^\infty := \lim_{T\to\infty} \frac{1}{T} \int_0^T dt \, \rho(t) = \sum_{n=1}^{n_E} p_n \ketbra{\lambda_n}.
\end{equation}
Analogously to the non-degenerate case, this state coincides with the diagonal ensemble or totally dephased state, i.e., the state in which the coherences between different energy subspaces vanish. Given the coarse energy eigenprojectors $\{\Pi_n\}$, $\rho_{DE} := \sum_n \Pi_n \rho(t) \Pi_n = \rho_\infty$.
In the case of degenerate Hamiltonians, the equilibrium state $\rho_\infty = \rho_{DE}$ takes a block-diagonal form, as coherences within degenerate subspaces are preserved. Moreover, while in the non-degenerate case the information about the initial state is contained only in the coefficients $p_n$, in this case it is also contained in the states $\{ \ket{\lambda_n} \}_{n=1}^{n_E}$.

\subsection{Non-commuting charges and Hamiltonian degeneracies}\label{sec:NonCommQ}
We consider Hamiltonians with non-commuting conserved quantities (`charges') $\{ Q^a \}_{a}$, implying the presence of a non-Abelian symmetry.
 Such charges force the Hamiltonian to be degenerate~\cite{majidy_noncommuting_2023, yunger_halpern_noncommuting_2020}. This is because, for two non-commuting operators to share an eigenbasis with the Hamiltonian,  the energy eigenbasis must not be unique. However, one might wonder whether the reverse holds, i.e., does a degenerate Hamiltonian always possess non-commuting charges? Our first result answers this question in the affirmative, showing that the presence of non-commuting charges and Hamiltonian degeneracies are inextricably linked, both in finite and infinite-dimensional (separable) Hilbert spaces. More precisely:
\begin{theorem} \label{thm:non-Abelian_degenerate_Hamiltonian}
    Consider a complex finite-dimensional Hilbert space $\mathcal{H}$ and a Hamiltonian $H$ with symmetry group $\mathcal{G}(H) := \{ U\in \mathrm{U}(\mathcal{H}) : [U, H] = 0\}$, where $\mathrm{U}(\mathcal{H})$ is the group of unitary operators on $\mathcal{H}$. Then, $H$ is non-degenerate if and only if $\mathcal{G}(H)$ is Abelian. Moreover, this holds even if $\mathcal{H}$ is infinite-dimensional separable, with adapted definitions of commutation and degeneracy.
\end{theorem}
\begin{proof}[Proof of the finite-dimensional case]
     A unitary commutes with $H$ if and only if it is block-diagonal with respect to the decomposition of the Hilbert space in terms of the energy eigensubspaces $\mathcal{H} = \bigoplus_{n=1}^{n_E} \mathcal{H}_n$. By definition, if $H$ is non-degenerate then every block $\mathcal{H}_n$ is one-dimensional; so the block-diagonal unitaries are actually a direct sum of scalars, hence they all commute. However, if $\exists \, d_n > 1$ then the set of block-diagonal unitaries contains a copy of the unitary group $\mathrm{U} \left( d_n \right)$, which is non-Abelian.
\end{proof}
\noindent A similar argument holds for the infinite-dimensional case, with modified versions of the decomposition $\mathcal{H} = \bigoplus_{n=1}^{n_E} \mathcal{H}_n$ and of (block-)diagonal operators; the details appear in Appendix~\ref{sec:app_deg_H_non-Abelian_symmetry}. This suggests that if a Hamiltonian becomes degenerate in the thermodynamic limit, then it has a corresponding non-Abelian symmetry emerging in the same limit. \\

Given a degenerate Hamiltonian, Theorem~\ref{thm:non-Abelian_degenerate_Hamiltonian} guarantees a continuous non-Abelian symmetry group and hence the presence of non-commuting charges. The latter are the elements of the Lie algebra that generate the Lie group $\mathcal{G}(H)$. Throughout the rest of the paper, when referring to ``non-commuting charges", we will specifically mean a set $\{Q^a\}_a$ of algebraically independent elements of the Lie algebra. For instance, consider the following Hamiltonian acting on a one-dimensional lattice of $N$ spins: $H = \sum_{a=x,y,z} \left\{ \sum_{i=1}^{N-1} J_i \sigma^a_i \sigma^a_{i+1} + \sum_{i=1}^{N-2} J_i \sigma^a_i \sigma^a_{i+2} \right\}$, where $i$ labels the lattice sites, $J_i$ are the interaction strength coefficients, and $\sigma^a_i$ is a Pauli operator acting non-trivially only on site $i$. This has a global SU(2) symmetry, in the sense that it conserves the mean total magnetizations along $x$, $y$ and $z$, namely, $Q^a := \frac{1}{N} \sum_{i=1}^N \sigma^a_i$ with $a = x, y, z$.  
We will be studying this model numerically in Sec.~\ref{sec:numerical_results}.

Given that $[H, Q^a] = 0 \; \forall a$, the $\{Q^a\}_a$ cannot mix different energy eigenspaces. This can be seen by considering a generic state $\ket{\psi_n} \in \mathcal{H}_n$ and noting that $H (Q^a \ket{\psi_n}) = Q^a (H \ket{\psi_n}) = E_n (Q^a \ket{\psi_n})$, so $Q^a \ket{\psi_n}$ lives in the same subspace $\mathcal{H}_n$. Indeed, we can write $Q^a = \sum_n \Pi_n Q^a \Pi_n =: \sum_n Q^a_n$ and thus $Q^a \ket{\psi_n} = Q^a_n \ket{\psi_n}$. Nevertheless, the charges act non-trivially within each subspace. 
In particular, the action of $Q^a$ on $\ket{\lambda_n}$ is generally non-trivial and hence $[Q^a, \ketbra{\lambda_n}] = \frac{1}{p_n} \Pi_n [Q^a, \rho(0)] \Pi_n \neq 0$. 

We are especially interested in cases where these charges play a thermodynamic role. This means they have the form $Q^a = Q^a_\mathcal{S} \otimes \mathbb{I}_\mathcal{E} + \mathbb{I}_\mathcal{S} \otimes Q^a_\mathcal{E}$, where $Q^a_\mathcal{S}$ and $Q^a_\mathcal{E}$ act respectively on a subsystem $\mathcal{S}$ and on the environment $\mathcal{E}$. $Q^a$ is globally conserved, and the local charges are exchanged between a subsystem and the environment~\cite{majidy_noncommuting_2023, yunger_halpern_how_2022}. This is indeed the case for the $\mathrm{SU}(2)$-symmetric model described above, which we will study numerically.

\subsection{Observables and initial states} \label{subsec:observables}
We focus on so-called coarse or highly degenerate observables, as these are a large and experimentally accessible class of observables. Thus, we consider observables $A$ with $n_A \ll D$ distinct outcomes $\{a_j\}_{j=1}^{n_A}$, each having degeneracy $d_j$. That is, 
\begin{equation}
    A := \sum_{j=1}^{n_A} a_j A_j := \sum_{j=1}^{n_A} a_j \sum_{s=1}^{d_j} \ketbra{j, s}, 
\end{equation}
where $A_j$ is the eigenprojector onto the subspace corresponding to eigenvalue $a_j$, and $\{ \ket{j,s} \}_{js}$ is a complete basis with the index $s$ labelling the degeneracies for each eigensubspace. We denote the probability distribution of each eigenstate by $\{p_{js} := |\braket{\psi(t)}{j,s}|^2\}_{j,s}$. 
Our main focus will be the probability distribution of the observable's eigenvalues, given by
\begin{gather}
    p_j(t) := \expval{A_j}{\psi(t)} \quad j = 1, \ldots, n_A.
\end{gather}
These distributions of course depend on the observable and so should have an index A; however, we omit this throughout the text for simplicity. 

We consider experimentally accessible initial states supported in an energy window which is macroscopically small, yet large enough to contain many energy eigenstates. That is, the energy of the system is $E := \expval{H}{\psi}$ up to some small uncertainty $\Delta E := \sqrt{\expval{H^2}{\psi} - E^2}$ such that $\frac{\Delta E}{E} \sim \frac{1}{\sqrt{N}} \xrightarrow{N \to \infty} 0$. Thus, we have a small energy window $\mathcal{I}_{mc} := \left[ E - \frac{\Delta E}{2}, E + \frac{\Delta E}{2} \right]$. 
We denote the charges' expectation values by $\{q^a := \expval{Q^a}{\psi}\}_a$ and their standard deviations by $\{\Delta q^a := \sqrt{\expval{(Q^a)^2}{\psi} - (q^a)^2}\}_a$. Following Ref.~\cite{yunger_halpern_noncommuting_2020}, we assume that the standard deviations $\Delta q^a$ scale at most as $\sqrt{N}$, which corresponds to assuming the existence of an approximate microcanonical subspace \cite{yunger_halpern_microcanonical_2016, yunger_halpern_noncommuting_2020}. This condition is for instance satisfied by all short-range correlated states \cite{yunger_halpern_noncommuting_2020}.

\subsection{Non-Abelian thermal state} \label{subsec:nats}
Consider the standard thermodynamic setup in which both energy and particle number are globally conserved while being exchanged locally. Under the usual assumptions of weak coupling, exponential density of states and large $N$, one can show the grand canonical ensemble is the subsystem's thermal state (see, e.g., Refs.~\cite{landau_statistical_2011, tolman_principles_1980}). An analogous thermal state has been derived for Hamiltonians which conserve non-commuting charges~\cite{jaynes_information_1957-1, jaynes_brandeis_1989, halpern_beyond_2018, yunger_halpern_microcanonical_2016, hinds_mingo_quantum_2018, lostaglio_thermodynamic_2017, guryanova_thermodynamics_2016, majidy_noncommuting_2023}. It is usually referred to as the non-Abelian thermal state (NATS) and it takes the form
\begin{equation} \label{eqn:local_NATS}
    \rho_\textrm{NATS} = \frac{e^{-\beta H_S - \sum_a \mu^a Q^a_S}}{Z},
\end{equation}
where $H_S$ and the $\{Q^a_S\}_a$ are respectively the subsystem's Hamiltonian and charges, $\beta$ is the inverse temperature, the $\{\mu^a\}_a$ are the charges' chemical potentials and $Z := \Tr(e^{-\beta H_S - \sum_a \mu^a Q^a_S})$ is the partition function.

For $\rho_\textrm{NATS}$ to be predictive for small subsystems, there must exist an approximate microcanonical subspace~\cite{yunger_halpern_microcanonical_2016, balian_equiprobability_1987}. Moreover, the aforementioned assumptions needed to derive the canonical and grand canonical ensembles must still hold~\cite{yunger_halpern_microcanonical_2016, yunger_halpern_noncommuting_2020}. Some numerical evidence has been provided in~\cite{yunger_halpern_noncommuting_2020} to support this form of the local thermal state of a subsystem in the presence of non-commuting charges. Moreover, the \textit{global} form of the NATS has been shown experimentally to give better predictions than the canonical and grand canonical ensembles \cite{kranzl_experimental_2023}. Hence, in Sec.~\ref{sec:numerical_results}, we compare our predictions to those obtained via the local NATS. We show that in some cases our estimates are more accurate. Moreover, our results do not assume weak coupling nor locality of observables, just their coarseness, thus allowing to make computationally easy predictions beyond the regime of applicability of the NATS.

\section{Observable Statistical Mechanics under non-commuting charges} \label{sec:generalizing_obs_stat_mech}
\subsection{The Maximum Observable Entropy Principle} \label{subsec:MOEP_constraints}

Observable Statistical Mechanics~\cite{anza_information-theoretic_2017, scarpa_observable_2025} is based on the Maximum Observable Entropy Principle~\cite{jaynes_information_1957, jaynes_information_1957-1, jaynes_brandeis_1989}. Such principle states that one's best guess for the probability distribution of the observable's eigenbasis $p_{js} := |\braket{\psi(t)}{j,s}|^2$ at equilibrium is given by the distribution that maximizes its Shannon entropy $S_A(\{p_{js}\}) := - \sum_{j,s} p_{js} \log p_{js}$ under a set of constraints.

In Ref.~\cite{scarpa_observable_2025}, the constraints were taken to be the state normalization $\Tr(\rho) = 1$ and the average energy $\Tr(\rho H) = E$. 
In principle, many more constraints (e.g., all moments of the Hamiltonian) should be included if one takes Jaynes' principle at face value~\cite{jaynes_information_1957, jaynes_information_1957-1, jaynes_brandeis_1989}.
However, in Ref.~\cite{scarpa_observable_2025} some of us argued that keeping only the first two moments is sufficient to make good predictions for coarse observables, which possess little information about the energy due to their high degeneracy. We go over such arguments next.

One can interpret measuring an observable at equilibrium as a communication channel where the energy eigenstates $\{|\lambda_n\rangle\}_n$ are the inputs, each sent with probability $p_n$, and the observable's eigenprojectors $\{A_j\}_j$ are the outputs~\cite{scarpa_observable_2025}. Consider the classical mutual information between the observable's and energy's distributions at equilibrium $I_{eq}(A, H) := S(H) - S_{eq}(H|A)$, where $S(H) \coloneqq S_H(\{p_n\}) := - \sum_n p_n \log p_n$ is the Shannon entropy of the energy distribution and $S(H|A) := - \sum_{n, j} p(n, j) \log \frac{p(n,j)}{p(j)}$ is the conditional entropy, with $p(n, j)$ and $p(j)$ respectively the joint distribution and the observable's one (see Sec.~\ref{sec_app:small_mutual_info_extensive_entropy} in the Appendices for the technical details, including how we define a joint distribution). The mutual information $I_{eq}(A, H)$ quantifies the information shared by the random variables $H$ and $A$. For a coarse observable $A$, $I_{eq}(A, H)$ is much smaller than the maximum value it can take, $S(H)$. For non-degenerate Hamiltonians, the latter is expected to be extensive for realistic initial states, i.e., $S(H) \sim k_E N$, yet still be compatible with a narrow energy distribution, that is, $k_E \ll 1$. Here we assume this remains true for Hamiltonians that are not excessively degenerate, i.e., with $n_E \sim D$ distinct energy levels. Hence, the measurement is not very informative and including just the first two moments as constraints is enough.

Moreover, in Ref.~\cite{scarpa_observable_2025} some of us argued that at equilibrium $S_A(\{p_{js}\}) \geq S_H(\{p_n\})$, thus $S_A(\{p_{js}\})$ is extensive if $S_H(\{p_n\})$ itself is. It was also shown that for highly degenerate observables one can find a bound that is stronger in some cases. In particular, for a local observable with support on $N_S$ lattice sites and constant degeneracy $\log d_j \sim N - N_S$, we have that at equilibrium $S_A(\{p_{js}\}) \gtrsim N - N_S$, independently of the initial state. These arguments back the use of a maximum entropy principle.

Since here we focus on coarse observables, the previous arguments hold, as detailed in Appendix~\ref{sec_app:small_mutual_info_extensive_entropy}. There, we also argue that for Hamiltonians with non-commuting charges that are sums of identical one-site operators, higher moments of the charges are not needed, just like for the energy. This is because also the classical mutual information between the observable and charges distributions at equilibrium is much smaller than its upper bound. Hence, analogously to Ref.~\cite{scarpa_observable_2025}, here we only include as constraints the expectation values of the energy and of the charges, besides state normalization.

\subsection{The generalized Equilibrium Equations and their solution} \label{subsec:generalized_EEs_solution}
Consider maximizing the Shannon entropy 
\begin{align}
S_A(\{p_{js}\}) := - \sum_{j,s} p_{js} \log p_{js}
\end{align}
of the observable's distribution under the constraints
\begin{equation}
\label{eq:constraints}
    \Tr(\rho)=1, \quad \Tr(\rho H)=E, \quad \{ \Tr(\rho Q^a) = q^a \}_a.
\end{equation}
We solve this optimization problem using Lagrange multipliers to obtain a set of equilibrium equations for the distribution $\{p_{js}\}$. In every sufficiently degenerate observable eigensubspace such that $d_j(d_j-1) \geq n_E + 1$ holds, we can use Theorem 1 in Ref.~\cite{anza_eigenstate_2018} (see also Appendix~\ref{sec:AGH_theorem}) to obtain a set of equilibrium equations for $\{p_j\}$.

The equilibrium equations take the form 
\begin{gather}
    \beta_A \Tr(\rho [A_j, H]) + \sum_a \mu^a_A \Tr(\rho [A_j, Q^a]) \mbeq 0, \label{eqn:EE1} \\
    - p_j \log \left( \frac{p_j}{d_j} \right) 
    \mbeq \lambda_N p_j + \beta_A R_j + \sum_a \mu^a_A R_j^{a}, \label{eqn:EE2}
\end{gather}
where $\lambda_N, \beta_A, \{\mu^a_A\}_a$ are the Lagrange multipliers corresponding to the constraints in Eq.~\eqref{eq:constraints}. 

We refer to Eqs.~\eqref{eqn:EE1} and~\eqref{eqn:EE2} as the First and Second Equilibrium Equations, respectively. While $\rho(t)$ and the corresponding $p_j(t)$ within can be time-dependent, the above conditions must hold at equilibrium for the entropy to be at its constrained maximum. 
In Eqs.~\eqref{eqn:EE1} and~\eqref{eqn:EE2}, we defined the inner products between operators
\begin{gather}
    R_j(t) :=  \Tr\left(\rho(t) \frac{\{A_j, H\}}{2}\right) = \operatorname{Cov}^s_\rho(A_j, H) + E \, p_j \\
    R_j^a(t) := \Tr\left(\rho(t) \frac{\{A_j, Q^a\}}{2} \right) = \operatorname{Cov}^s_\rho(A_j, Q^a) + q^a \, p_j,
\end{gather}
where $\{A, B\} := AB + BA$ is the anticommutator, and $\operatorname{Cov}^s_\rho(A, B) := \tfrac{1}{2}\Tr(\rho \{A, B\}) - \Tr(\rho A) \Tr(\rho B)$ is the symmetrized covariance between $A$ and $B$. $E = \expval{H}{\psi}$ and $q^a = \expval{Q^a}{\psi}$ are the average energy and charges of the system. The details of the derivation can be found in Appendix~\ref{sec_app:generalized_EEs}. 

Summing both sides of Eq.~\eqref{eqn:EE2} over $j$ gives
\begin{equation} \label{eqn:equilibrium_shannon_entropy}
    S_A \left( \left\{ \frac{p_j}{d_j} \right\} \right)= \lambda_N + \beta_A E + \sum_a \mu^a_A q^a.
\end{equation}
Thanks to the Gibbs' inequality~\cite{jaynes_brandeis_1989, cover_elements_2005}, the most general probability distribution for which Eq.~\eqref{eqn:equilibrium_shannon_entropy} holds is given by
\begin{equation} \label{eqn:general_solution_EEs}
    p_j^\textrm{est} = \frac{d_j \, e^{- \beta_A \varepsilon_j - \sum_a \mu_A^a q^a_j}}{\mathcal{Z}_A},
\end{equation} 
with $\mathcal{Z}_A \coloneqq e^{\lambda_N} \mbeq \sum_j d_j \, e^{- \beta_A \varepsilon_j - \sum_a \mu_A^a q^a_j}$. $\varepsilon_j$ and $\{q^a_j\}_a$ are quantities such that $\sum_j \varepsilon_j p_j^\textrm{est} = E$ and $\sum_j q^a_j p_j^\textrm{est} = Q^a \; \forall a$.
Substituting Eq.~\eqref{eqn:general_solution_EEs} into Eq.~\eqref{eqn:EE2}, we find that at equilibrium one must have 
\begin{equation} \label{eqn:conditions_linearity_Rj}
    R_j = \varepsilon_j p_j^\textrm{est} \qquad \text{and} \qquad R_j^a = q^a_j p_j^\textrm{est} \; \forall a
\end{equation}
which define $\varepsilon_j$ and $\{q^a_j\}_a$.
The Equilibrium Equations Eqs.~\eqref{eqn:EE1} and \eqref{eqn:EE2}, as well as the solution Eq.~\eqref{eqn:general_solution_EEs}, are generalizations of the results in Ref.~\cite{scarpa_observable_2025}, which can be retrieved by setting all chemical potentials $\{\mu^a_A\}_a$ to zero.

Equation~\eqref{eqn:general_solution_EEs} is one of the main results of this work. It suggests thermal-like behaviour for the probability distributions of observables, with $\beta_A$ and $\mu_A^a$ playing the role of effective observable-dependent inverse temperatures and chemical potentials, respectively. In the next section, we show that $\varepsilon_j$ and the $q_j^a$'s can be interpreted as the real parts of weak values of the Hamiltonian and the charges, respectively. These, in turn, are linked to quasiprobability distributions that cannot be interpreted classically.

\subsection{Interpreting the first equilibrium equation}
The first term in the First Equilibrium Equation, Eq.~\eqref{eqn:EE1}, is proportional to the time derivative of the probability distribution $\partial p_j / \partial t$. This simply reflects the fact that the probability distribution must be (approximately) constant at equilibrium. The other terms in Eq.~\eqref{eqn:EE1} can be interpreted by considering a symmetry-preserving rotation $e^{-iQ^a \theta}$. Since $[Q^a,H] = 0$, such a rotation commutes with the Hamiltonian and hence can only act non-trivially within each energy eigensubspace. By considering the rotation's action on the state $\rho(t)$ at any time $t$, i.e., $\rho(t, \theta) = e^{-iQ^a \theta} \rho(t) \, e^{iQ^a \theta}$, and differentiating with respect to $\theta$, we find a von Neumann-like equation $\Tr(A_j [Q^a, \rho]) = i \frac{\partial p_j}{\partial \theta}$. 
Thus, we can re-write Eq.~\eqref{eqn:EE1} as
\begin{equation}
   0 \mbeq \beta_A \frac{\partial p_j}{\partial t} + \sum_a \mu^a_A \frac{\partial p_j}{\partial \theta^a}. \label{eqn:EE1_params_derivatives}
\end{equation}
Each charge-related term can be understood as meaning that $p_j$ should not change much if we vary the support of the initial state (or equivalently the equilibrium state) within each energy eigensubspace. In other words, what truly affects the probability distribution should be the energy sectors on which it has support, not the specific details within each sector.

\subsection{Validity of the equilibrium prediction}
When is Eq.~\eqref{eqn:general_solution_EEs} a good prediction for the observable's distribution at equilibrium? 
We stress that, while the framework can also be applied to fine-grained observables, the derivation of Eq.~\eqref{eqn:general_solution_EEs} relies on the coarseness of the observables under study (i.e., that their eigenprojectors are of large rank). 
Given coarse observables, the following three conditions have to hold. 
First, as we are predicting the stationary value of the observable's distribution, the observable must equilibrate (on a time average sense~\cite{reimann_foundation_2008}). 
The First Equilibrium Equation, Eq.~\eqref{eqn:EE1}, linked to equilibration, must hold. 
Secondly, the observable's Shannon entropy must be close to its constrained maximum, thus justifying the use of a maximum entropy principle. Lastly, the (classical) mutual information between the observable and the energy, as well as between the observable and the charges, must be small. As we argued in Sec.~\ref{subsec:MOEP_constraints}, if an observable carries little information about the Hamiltonian and the charges, using only a few constraints in the entropy maximization is justified (see Appendix~\ref{sec_app:small_mutual_info_extensive_entropy} for a more detailed discussion).

\section{Non-classical behavior at equilibrium} \label{sec:interpretation_KDQ_WV}
Next, we show that the Equilibrium Equations and the distribution that maximizes the entropy can be interpreted in terms of quantum weak values~\cite{aharonov_how_1988, dressel_colloquium_2014} and their related Kirkwood-Dirac quasiprobability distributions (KDQs)~\cite{lostaglio_kirkwood-dirac_2023, arvidsson-shukur_properties_2024}. In this way, we provide novel operational meanings to weak values and KDQs in a  quantum thermodynamic context. 

Given a state $\rho$ (the preparation) and a state $\sigma$ (the postselection), the weak value of an observable $O$ is defined as
\begin{equation} \label{eqn:weakvalue}
    O_w \left(\rho, \sigma \right) \coloneqq  \frac{\Tr( \rho O \sigma)}{\Tr(\rho \sigma)}.
\end{equation}
Given an observation that leaves a system in the postselected state $\sigma$, $O_w \left(\rho, \sigma \right)$ is the best retrodiction to estimate what the value of the observable $O$ was \emph{before} the measurement~\cite{dressel_colloquium_2014}. Weak values can manifest in scenarios with postselection, or in conditioned systems which arise naturally in monitored quantum systems~\cite{gammelmark_past_2013}. For instance, the past output of a continuous quantum measurement is more accurately described by weak values than by estimates obtained from the prior states of a quantum system~\cite{garcia-pintos_past_2017}.

Unlike traditional expectation values, weak values can be complex and can lie outside of the spectrum of the observable $O$. When they do, they are said to be anomalous. Systems where anomalous weak values manifest are provably non-classical, in the sense that no non-contextual classical model can describe the system~\cite{pusey_anomalous_2014, kunjwal_anomalous_2019}.

Using Eq.~\eqref{eqn:weakvalue}, we can re-write the Equilibrium Equations in terms of the real and imaginary parts of the weak values of the Hamiltonian and of the charges, $E_w (\rho,A_j) \coloneqq \Tr(\rho H A_j)/\Tr(\rho A_j)$ and $Q_w^a (\rho,A_j) \coloneqq \Tr(\rho Q^a A_j)/\Tr(\rho A_j)$. These correspond to a postselection to state $A_j/d_j$.  Equations~\eqref{eqn:EE1} and~\eqref{eqn:EE2} become
\begin{align}
    &0 \mbeq \beta_A \Im[E_w(\rho,A_j)] + \sum_a \mu^a_A \Im[Q^a_w(\rho,A_j)], \label{eqn:EE1WV} \\
    &- \log \left( \frac{p_j}{d_j} \right) 
    \mbeq \lambda_N + \beta_A \Re[E_w(\rho,A_j)] \label{eqn:EE2WV} \\
    &\qquad \quad \quad \qquad \qquad  + \sum_a \mu^a_A \Re[Q^a_w(\rho,A_j)]. \nonumber 
\end{align}
This sets conditions which the real and imaginary parts of the weak values of the charges and of the Hamiltonian must satisfy for the Equilibrium Equations to hold.

Similarly, Eqs.~\eqref{eqn:conditions_linearity_Rj} and~\eqref{eqn:weakvalue} imply that
\begin{align}
  \varepsilon_j = \Re[E_w(\rho,A_j)],\qquad  q_j^a = \Re[Q_w^a(\rho,A_j)].
\end{align}
We leverage the interpretations of weak values to provide a physical meaning to the $\varepsilon_j$s and the $q_j^a$s: they are the best estimates of the system's energy and charges given the preparation $\rho$ and the postselection $A_j/d_j$~\cite{dressel_weak_2015}. These are refinements to the estimates obtained solely from the system's equilibrium state $\rho$.

Such quantities are central to the equilibrium distribution obtained from maximizing the constrained entropy. From Eq.~\eqref{eqn:general_solution_EEs}, we find it to be
\begin{equation} \label{eqn:general_solution_EEsWV}
    p_j^\textrm{est} = \frac{d_j \, e^{- \beta_A \Re[E_w(\rho,A_j)] - \sum_a \mu_A^a \Re[Q_w^a(\rho,A_j)]}}{\mathcal{Z}_A}.
\end{equation} 
This uncovers widespread scenarios where (real parts of) weak values naturally arise. We have found that the equilibrium probability distributions of observables are characterized by the real parts of the weak values of Hamiltonian and charges. Further, we will numerically show that, under non-commuting conserved quantities, the charges' weak values can be anomalous.

A possible solution to Eq.~\eqref{eqn:EE1WV}, which is independent of the specific values of the Lagrange multipliers and eigenvalues, corresponds to all imaginary parts of the weak values being null at equilibrium.
However, if the $\Im[Q_w^a(\rho,A_j)] \neq 0$, which we will find to be the case in the simulated examples, then Eq.~\eqref{eqn:EE1WV} sets a non-trivial constraint between the charges' weak values and their corresponding chemical potentials.
For instance, consider a system with two non-commuting charges $Q^a$ and $Q^b$, and their associated  effective chemical potentials $\mu_A^a$ and $\mu_A^b$. Since at equilibrium $[\rho,H] = 0$, if the charges' imaginary parts of the weak values are not zero, then $\mu_A^a/\mu_A^b = - \Im[Q_w^b]/\Im[Q_w^a]$ holds.

Weak values are closely related to Kirkwood-Dirac quasiprobability distributions (KDQs). KDQs represent (generalized) joint distributions for two operators but, just like other quantum distributions such as the Wigner function,  can violate the first Kolmogorov axiom. This is a consequence of the non-commutativity of the bases and the state. In particular, the KDQs can take negative and non-real values, in which case they are said to be non-classical~\cite{lostaglio_kirkwood-dirac_2023,arvidsson-shukur_properties_2024}.

Consider the spectral decomposition of the charges $Q^a = \sum_\mu q^a_\mu Q^a_\mu$, where $\{Q^a_\mu\}_\mu$ are a charge's coarse-grained eigenprojectors. The KDQs on the state $\rho$ for the observable's eigenprojector $A_j$ and the energy eigenprojector $\Pi_n$, as well as for $A_j$ and the charge eigenprojector $Q^a_\mu$, are defined by $K_{j,n}(\rho) := \Tr( A_j \rho \Pi_n)$ and $K_{j,\mu}(\rho) := \Tr(A_j \rho Q^a_\mu)$, respectively. The weak values can be obtained from these KDQs by taking the expectation value $\sum_n E_n K_{j,n}(\rho)$ and dividing by the marginal $p_j = \sum_n K_{j,n}(\rho)$ (and analogously for the charges' KDQs). At equilibrium, i.e., on the $\rho_\infty$, $K_{j,n}(\rho_\infty) = p_n \sum_s |\braket{j,s}{\lambda_n}|^2$ is classical. However, in general $[Q^a_\mu, \rho_\infty] \neq 0$, so the KDQ for the non-commuting charges can be non-classical even at equilibrium. We stress that, if instead we had a non-degenerate Hamiltonian, the KDQs $K_{j,\mu}(\rho)$ corresponding to \textit{any} conserved quantity would always be classical at equilibrium.

\section{Numerical results} 
\label{sec:numerical_results}

\subsection{Model and methods} \label{subsec:numerics_model_methods}
To test the analytical predictions, we study numerically the one-dimensional spin-1/2 non-integrable XXX model, described by the Hamiltonian
\begin{equation} \label{eqn:hamiltonian}
    H = \sum_{a=x,y,z} \left\{ \sum_{i=1}^{N-1} J_i \sigma^a_i \sigma^a_{i+1} + \sum_{i=1}^{N-2} J_i \sigma^a_i \sigma^a_{i+2} \right\},
\end{equation}
where $i$ is an index that runs over the $N$ lattice sites, and $\sigma^a_i$ ($a=x,y,z$) represents the Pauli operator with Pauli matrix $\sigma^a$ acting on the lattice site $i$, i.e.
\begin{equation} \label{eqn:local_pauli}
 \sigma^a_i \coloneqq \mathbb{I}^{\otimes (i-1)} \otimes \sigma^a \otimes \mathbb{I}^{\otimes(N-i)}.
\end{equation}
$H$ has SU(2)-symmetry, i.e., it conserves the mean total magnetizations along $x$, $y$ and $z$, defined as 
\begin{equation} \label{eqn:total_mag_def}
    Q^a := \frac{1}{N} \sum_{i=1}^{N} \sigma^a_i, \quad a = x, y, z .
\end{equation}
These charges are conserved globally but transported locally \cite{yunger_halpern_how_2022}, and they do not commute with each other: $[Q^a, Q^b] \neq 0 \quad \forall \, a \neq b$.
Note that, following Ref.~\cite{lasek_numerical_2024}, we break all other symmetries to focus on SU(2) only. Thus, we use open boundary conditions to break translation invariance; we include the next-nearest neighbour interaction term to break integrability; and add a small one-site perturbation of strength $\epsilon = 0.3$ to the interaction coefficient $J_i = 1 + \epsilon \, \delta_{i,3}$ to break also spatial-inversion symmetry. 

We consider a one-parameter family of initial states given by
\begin{equation} \label{eqn:psi_theta_m}
    \ket{\psi_0(\theta)} = R_y(\theta) \otimes R_y(-\theta) \otimes \ldots \otimes R_y(\theta) \otimes R_y(-\theta) \ket{01\ldots01},
\end{equation}
where $R_y(\theta) = \exp(-iY\theta/2)$ is the rotation operator along the $y$ axis, and the rotations alternate between the angles $\theta$ and $- \theta$. Such rotation is chosen to be SU(2)-breaking, so that we can explore different average energy subspaces. Our numerics focus on  five points, $\theta = \{0, \frac{\pi}{16}, \frac{\pi}{8}, \frac{3\pi}{16}, \frac{\pi}{4}\}$, avoiding $\theta = \pi/2$ since it is a ground state of the Hamiltonian. 

We employ the two-site time-dependent variational principle (TDVP) ~\cite{haegeman_time-dependent_2011,haegeman_unifying_2016,collura_tensor_2024} to simulate the state's evolution, $|\psi_t\rangle = e^{-iH dt}|\psi_0 \rangle$. We simulate three different system sizes, $N=12, 14, 16$; however, all the plots in the main text report the results for $N=16$ unless otherwise stated. The simulation is performed with a finite time step of $dt = 0.02$ and final time $t_f = 40.02$. The truncation error is kept below $\sim 10^{-8}$ per time step throughout the evolution.

In the numerical study, we consider the SU(2)-breaking one-body observables $\sigma^a_i$ $(a = x, y, z)$ of the form given in Eq.~\eqref{eqn:local_pauli}, and two-body observables of the form
\begin{equation} \label{eqn:two-body_obs}
 \sigma^a_i \sigma^a_{i+1} = \mathbb{I}^{\otimes (i-1)} \otimes \sigma^a \otimes \sigma^{a} \otimes \mathbb{I}^{\otimes(N-1-i)}.
\end{equation}
Their coarse-grained eigenprojectors are respectively
\begin{gather} \label{eqn:coarse-grained_projectors}
    A_j^{a,i} = \frac{\mathbb{I} + (-1)^{j} \sigma^a_i}{2} \\
    A_{jk}^{a,i} = A_j^{a,i} A_k^{a,i+1},
\end{gather}
with $j,k \in \{0,1\}$. The two-body observables have four distinct eigenvalues, $\{00,01,10,11\}$. We assume open boundary conditions, and consider the observables with support on the central lattice site $i=N/2$ (for even $N$) and, in the case of two-body observables, also on the subsequent one. 
All extensive quantities at equilibrium are rescaled by $1/N$ to ensure we can compare them among different system sizes.

\subsection{Anomalous weak values} \label{subsec:numerics_anomalous_weak_values}
 In Section~\ref{sec:interpretation_KDQ_WV}, we concluded that for degenerate Hamiltonians it may be possible to witness anomalous (or `non-classical') behavior of the weak values and their corresponding KDQs even at equilibrium. 
 Such anomalies can emerge in two ways: the imaginary parts of the weak values can be non-zero or their real parts can take values outside of the range of the observable's spectrum~\cite{arvidsson-shukur_properties_2024}. 
 
 As the Hamiltonian commutes with the equilibrium state, but the charges in general do not, we study weak values and KDQs corresponding to an observable $A$ and charge $Q^a$ such that $[A, Q^a] \neq 0$. To consider equilibrium values, we compute the time-averaged quantities $\overline{R_j^a(t)} \approx \Re[\Tr(\rho_\infty Q^a A_j)]$, $\overline{\Im[\Tr(\rho(t) Q^a A_j)]} \approx \Im[\Tr(\rho_\infty Q^a A_j)]$, and $\overline{p_j(t)} \approx \Tr(\rho_\infty A_j)$, where the overline denotes the finite-time average. The corresponding weak value parts are $q_j^a = \Re[Q^a_w (\rho_\infty, A_j)] \approx \overline{R_j^a(t)}/\overline{p_j(t)}$ and $\Im[Q^a_w (\rho_\infty, A_j)] \approx \overline{\Im[\Tr(\rho(t) Q^a A_j)]}/\overline{p_j(t)}$.
 
We find that $\Im[Q^a_w (\rho_\infty, A_j)]$ is non-zero for some combinations of charges and observable eigenprojectors, thus witnessing non-classical behavior even at equilibrium. The relevant combinations involve the charge $Q^z$ and a $y$-observable, or viceversa the charge $Q^y$ and a $z$-observable. This is because the commutator of these two quantities is related to $x$-observables, and for $\theta > 0$ the expectation value $\langle Q^x \rangle$, and thus the chemical potential $\mu^x$, is non-zero. An example of this anomalous weak value is given in Fig.~\ref{fig:Im_WV_EE1_all_m}, where we plot the imaginary part of the weak value of the charge $Q^z$ at equilibrium $\Im[Q^z_w (\rho_\infty, A_{00}^{y, N/2})]$ against the initial state, for one of the eigenprojectors of the observable $\sigma^y_\frac{N}{2} \sigma^y_{\frac{N}{2}+1}$.

\begin{figure}
    \centering
    \includegraphics[width=\linewidth]{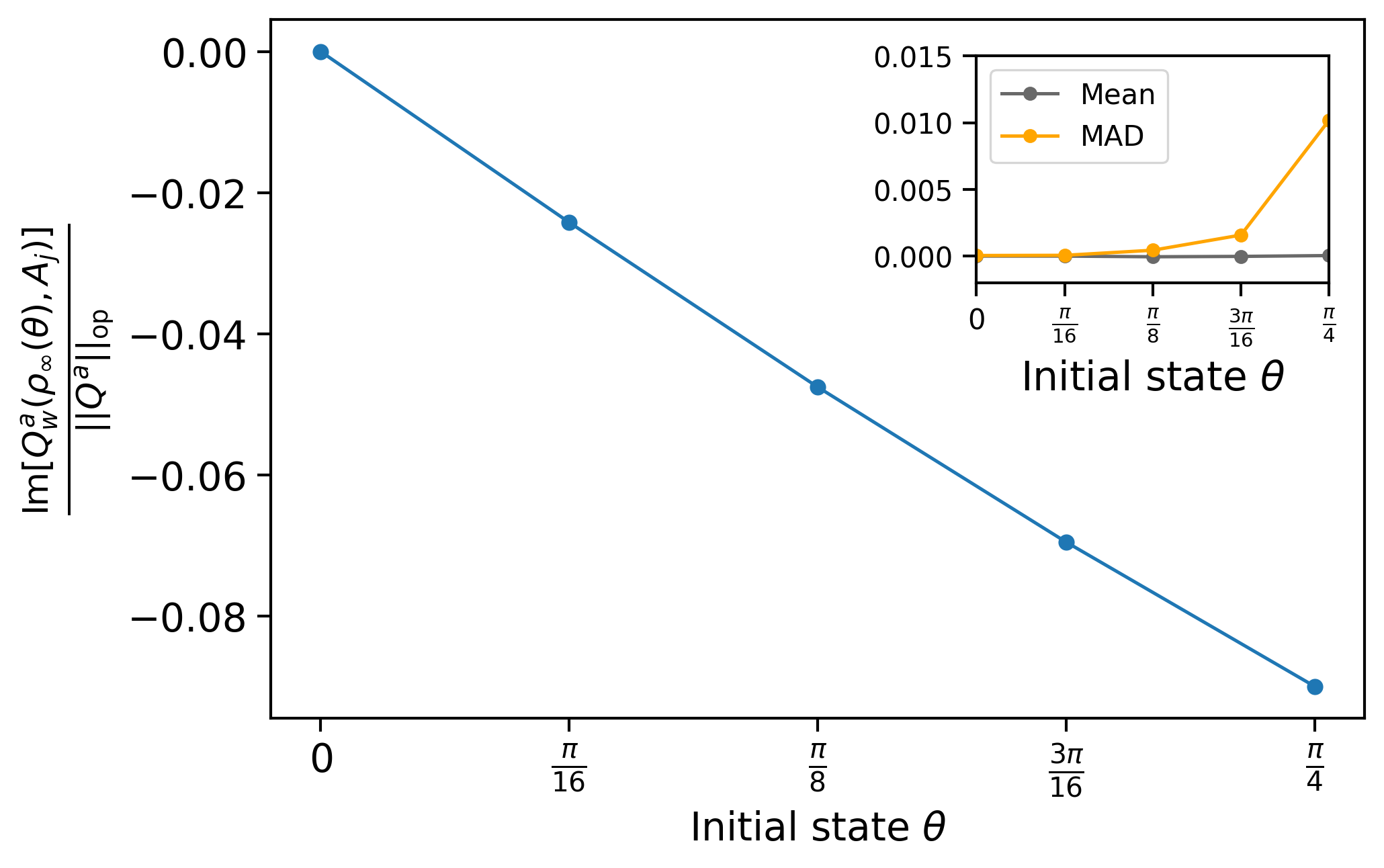}
    \caption{The imaginary part $\frac{\Im\left[Q^z_w \left(\rho_\infty(\theta), A_{00}^{y, N/2} \right)\right]}{||Q^z||_\textrm{op}}$ of the weak value of the charge $Q^z$ at equilibrium, normalized by $||Q^z||_\textrm{op} = 1$, for different initial states parametrized by the angle $\theta$, for $N=16$. Such quantity is computed with respect to the eigenprojector $A_{00}^{y, N/2}$, which projects onto the eigensubspace corresponding to outcome $00$ of the observable $\sigma^y_\frac{N}{2} \sigma^y_{\frac{N}{2}+1}$. The non-zero values of $\Im\left[Q^z_w \left(\rho_\infty(\theta), A_{00}^{y, N/2} \right)\right]$ prove non-classical (contextual) behaviour~\protect\cite{lostaglio_certifying_2020}. 
    The inset concerns the First Equilibrium Equation corresponding to the same eigensubspace of the same observable. We plot the absolute value of its temporal mean, and its mean absolute deviation (MAD) against the initial state. As both quantities are small, this shows that the First Equilibrium Equation is still respected despite the anomaly, thanks to the smallness of the chemical potentials.}
    \label{fig:Im_WV_EE1_all_m}
\end{figure}

The anomalies are small but non-negligible, so it is natural to wonder how this affects the First Equilibrium Equation, Eq.~\eqref{eqn:EE1WV}, which has to be respected at equilibrium, $\mathcal{E}_A(\rho, j) := \beta_A \Im[E_w (\rho, j)] + \sum_a \mu^a_A \Im[Q^a_w (\rho, j)] \mbeq 0$. 
We find that this equation is well respected in all cases despite the anomalies due to the smallness of the corresponding chemical potentials; see Appendix~\ref{app_subsec:reWV_param_space_exploration}. We show an example of this in the inset of Fig.~\ref{fig:Im_WV_EE1_all_m}, where we consider the LHS of the first Equilibrium Equation corresponding to the same observable's eigenspace as in the main plot, namely, the eigenspace corresponding to outcome $00$ of the observable $\sigma^y_\frac{N}{2} \sigma^y_{\frac{N}{2}+1}$. We see that the absolute value of the temporal mean is essentially zero, i.e., $|\overline{\mathcal{E}_A(\rho, j)}| \approx 0$, and the fluctuations, quantified by the mean absolute deviation (MAD), are also small, namely, $\overline{|\mathcal{E}_A(\rho, j) - \overline{\mathcal{E}_A(\rho, j)}|} \ll 1$. Nevertheless, we should note that, while the First Equilibrium Equation is respected in all cases for the model and the class of initial states considered, it could in principle be violated in other settings. Other plots regarding the anomalous imaginary weak values and the corresponding equilibrium equations can be found in in Section~\ref{app_subsec:additional_plots} of the Appendices. 

Weak values can be anomalous also by having real parts that take values outside of the range of the observable's spectrum. We investigated this too, but did not find any violation  for the cases studied. In fact, we explored a large parameter space by varying the Hamiltonian parameters and considering general rotations of the initial state in terms of three Euler angles. We witnessed anomalous behavior in the imaginary part of the weak value in $\sim 20\%$ of cases for SU(2)-breaking rotations, but never in its real part. It would be interesting to find regimes where they appear, possibly in models with different non-Abelian symmetries such as SU(3), or understand why they don’t. More details on this numerical exploration are available in Section \ref{app_subsec:reWV_param_space_exploration} of the Appendices.

\subsection{Predicting the real parts of the weak values}
 
\begin{figure}
    \centering
    \includegraphics[width=\linewidth]{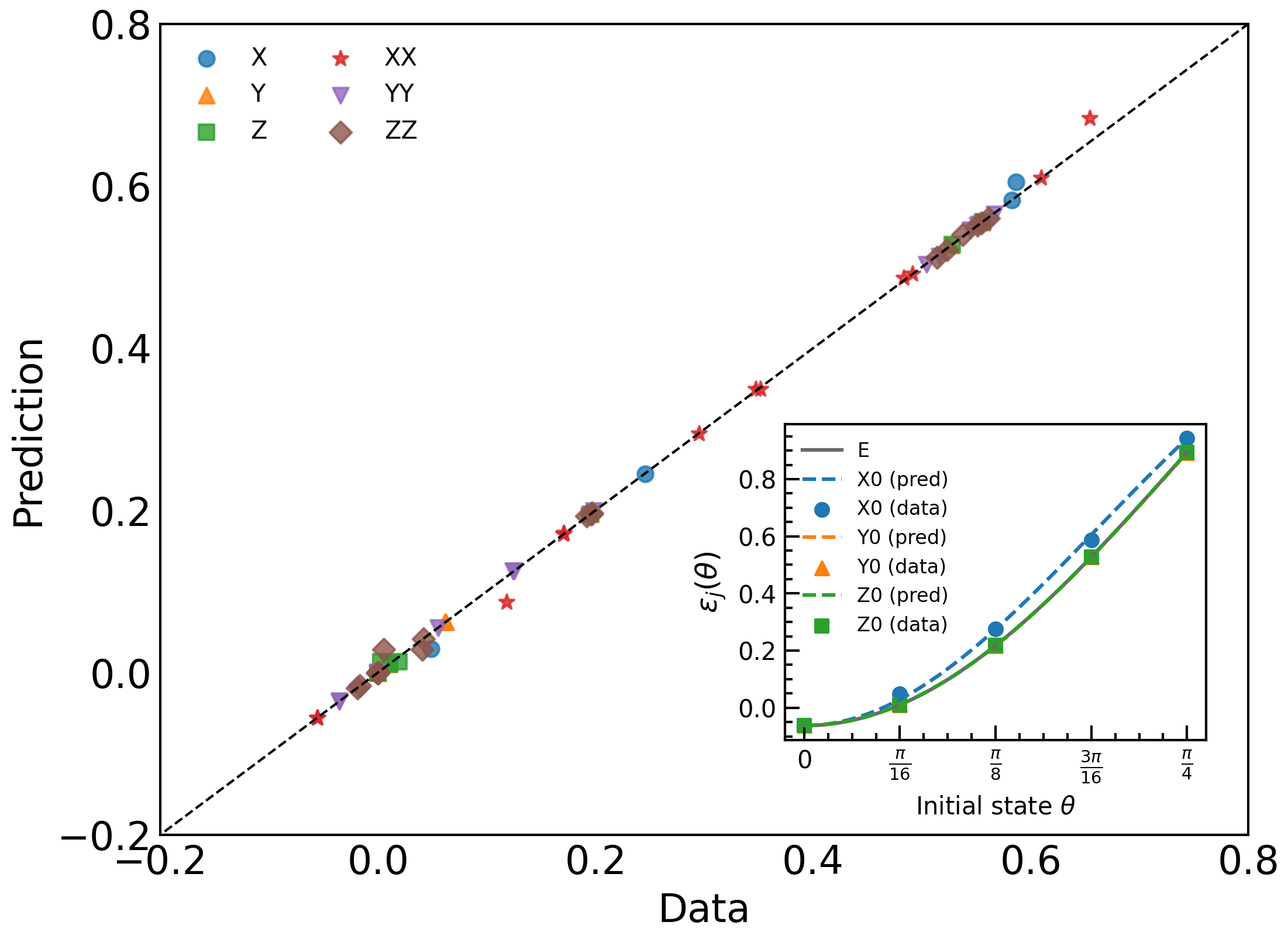}
    \caption{Comparison between our predictions for the real parts of weak values $\varepsilon_j$ and $q^a_j$, and the numerical data for $N=16$. As described in the text, we use two or three data points to fix the parameters in Eqs.~\eqref{eqn:epsj_approx_E_dE} and \eqref{eqn:qaj_approx_qa} and then compare the predictions with the remaining data. The black diagonal dashed line represent perfect agreement between data and predictions. The plot includes $\varepsilon_j$ and $q^a_j$ for all observables and all their (independent) eigenvalues. In the inset, we give an example of these predictions for $\varepsilon_j(\theta)$ plotted against the initial state, parametrized by $\theta$, for the eigenvalue $j=0$ of one-body observables, for $N=16$. For simplicity, we denote ``$X0$" the case corresponding to the eigenvalue $j=0$ of the observable $X \coloneqq \sigma^x_\frac{N}{2}$ etc., and ``pred" stands for ``prediction". The grey line (in this case coinciding with the green and orange lines) is the average energy $E(\theta)$. Note that we computed analytically the average and standard deviation of both the energy and the charges; the explicit analytical expressions can be found in Section~\ref{subsec:average_std_conserved_quantities} of the Appendices.}
    \label{fig:Prediction_Re_WV}
\end{figure}

From Eq.~\eqref{eqn:general_solution_EEs}, the real parts of the weak values $\varepsilon_j$ and $\{q^a_j\}_a$ are crucial to determine our prediction for an observable's probability distribution at equilibrium. Therefore, it is essential to have an easy way to compute them that does not require access to the energy eigenvalues and eigenvectors. In Ref.~\cite{scarpa_observable_2025}, it was found empirically that $\varepsilon_j$ could be well approximated by a linear function in the energy average and standard deviation. While a proper analytical derivation of this result has not yet been found, this can be understood by considering the constraints $\sum_j \varepsilon_j \, p_j^\textrm{est} = E$. If $\varepsilon_j$ and $p_j^\textrm{est}$ were independent, then we would indeed expect a linear relation with the energy. However, since $p_j^\textrm{est}$ actually depends on $\varepsilon_j$, we have a correction in terms of the standard deviation. Given that the same type of constraints also hold for the charges, we expect this same linear relation to hold here too. Indeed, this is what we observe numerically and in fact for the charges we do not even need to include the information about the standard deviation. Thus, we have
\begin{gather}
    \varepsilon_j \approx \gamma_j E + \eta_j \Delta E + \chi_j \label{eqn:epsj_approx_E_dE}\\
    q^a_j \approx \gamma_j^a q^a + \chi^a_j, \label{eqn:qaj_approx_qa}
\end{gather}
where $\{ \gamma_j, \eta_j, \chi_j \}$ and their charge counterparts are all coefficients that need to be fixed using some data or analytical calculation. 

In Fig.~\ref{fig:Prediction_Re_WV}, we show that such approximations hold well in all cases considered. To do so, given an observable and one of its (independent) eigenvalues, we consider $\varepsilon_j(\theta)$ and $q^a_j(\theta)$ as functions of $\theta$, which parametrizes the initial state. We take two or three data points and use them to fix the parameters in Eqs.~\eqref{eqn:epsj_approx_E_dE} and~\eqref{eqn:qaj_approx_qa}. Then, we use these equations to make predictions and compare them to the remaining data points. The main figure shows this comparison for $\varepsilon_j(\theta)$ and $\{q^a_j(\theta)\}_a$ for all observables and all (independent) eigenvalues. The inset shows an example of these predictions for the $\{\varepsilon_j(\theta)\}$ of one-body observables. In this case, the $\varepsilon_j$'s corresponding to $\sigma^y_\frac{N}{2}$ and $\sigma^z_\frac{N}{2}$ essentially correspond to the average energy, while $\varepsilon_j$ relative to $\sigma^x_\frac{N}{2}$ needs also the information about the energy standard deviation to be properly approximated. Thus, we have a simple way to predict $\varepsilon_j(\theta)$ and $q^a_j(\theta)$ which only requires information about the average energy/charge and in some cases also information about the energy standard deviation, besides a couple of data points to fix the parameters.

\subsection{Predicting the equilibrium distribution}
Here, we compare our predictions for the equilibrium probability distribution $p_j^\textrm{est}$ with the simulated  equilibrium value, i.e., the temporal average $\overline{p_j(t)}$. First, we numerically confirmed that the observable's entropy dynamically reaches its constrained maximum, as shown in Appendix~\ref{app_subsec:additional_plots}. In the same appendix we also confirmed that the Second Equilibrium Equation is well respected. This supports the framework's applicability. 

To compute our estimate $p_j^\textrm{est} \propto e^{-\beta_A \varepsilon_j - \sum_a \mu^a q^a_j}$, we first obtain $\varepsilon_j \approx \overline{R_j(t)}/\overline{p_j(t)}$ and $q_j^a \approx \overline{R_j^a(t)}/\overline{p_j(t)}$. Then, we calculate $\beta_A$ and the $\{\mu_A^a\}$ either analytically or via numerical optimization. We quantify the closeness with $\overline{p_j(t)}$ using the total variation distance, defined as $D(p, q):= \frac{1}{2} \sum_k \left| p_k - q_k \right|$ for two probability distributions $p$ and $q$. Note that $0 \leq D(p, q) \leq 1$.

In Fig.~\ref{fig:TVD_pred_vs_DE_with_charges_against_m_inset_nats}, we show that our predictions work well, with an error under $1\%$ for all cases considered. 
Further, as anticipated in Sec.~\ref{subsec:nats}, we compare our estimates with those obtained from the relevant subsystem's thermal state, i.e., the local NATS given in Eq.~\ref{eqn:local_NATS}. In Appendix~\ref{app_subsec:comparing_predictions_to_nats} we determine analytically the predictions for the observable's probability distributions based on the NATS, and determine the inverse temperature $\beta$ and the chemical potentials $\{\mu^a\}$ by fitting the data. We then compute the total variation distance $D(\overline{p_j(t)}, \, p_j^\textrm{NATS})$ between the temporal average $\overline{p_j(t)}$ and the estimate made via the NATS $p_j^\textrm{NATS} := \Tr(\rho_\textrm{NATS} A_j)$. This is shown in the inset of Fig.~\ref{fig:TVD_pred_vs_DE_with_charges_against_m_inset_nats}. 

For one-body observables, our predictions [Eq.~\eqref{eqn:general_solution_EEs}] and those of the NATS work equally well, with total variation distances of the order of $10^{-10}$. However, for two-body observables our predictions are better than the ones based on the NATS. We quantify this in Fig.~\ref{fig:TVD_pred_vs_DE_relative_improvement_nats_vs_with_charges_obs=two} using the relative improvement $\Delta_r := 1 - \frac{D(\overline{p_j(t)}, \, p_j^\textrm{est})}{D(\overline{p_j(t)}, \, p_j^\textrm{NATS})}$, which tells us how much better we can do if we switch from $p_j^\textrm{NATS}$ to $p_j^\textrm{est}$. 

As discussed above, one of the necessary conditions for the local NATS to be predictive is weak coupling. In Appendix~\ref{app_subsec:non-weak-coupling}, we review various definitions of weak coupling found in the literature and study them for the model and parameter values considered. According to all three definitions, including the possibly most physical one, i.e., $\beta \|H_\textrm{int}\| \ll 1$ \cite{riera_thermalization_2012, gogolin_equilibration_2016}, where $\|H_\textrm{int}\|$ is the operator norm of the interaction Hamiltonian, the coupling does not appear to be sufficiently weak to fully justify the use of the local NATS. This explains why our predictions for two-body observables, which do not rely on the assumption of weak coupling, outperform the NATS. The fact that the NATS makes good predictions for one-body observables can instead be ascribed to the one-site subsystem Hamiltonian being zero, and hence to an effective infinite temperature. 
Appendix~\ref{app_subsec:additional_plots} includes further details comparing our predictions to those of NATS.
 
\begin{figure}
    \centering
    \includegraphics[width=\linewidth]{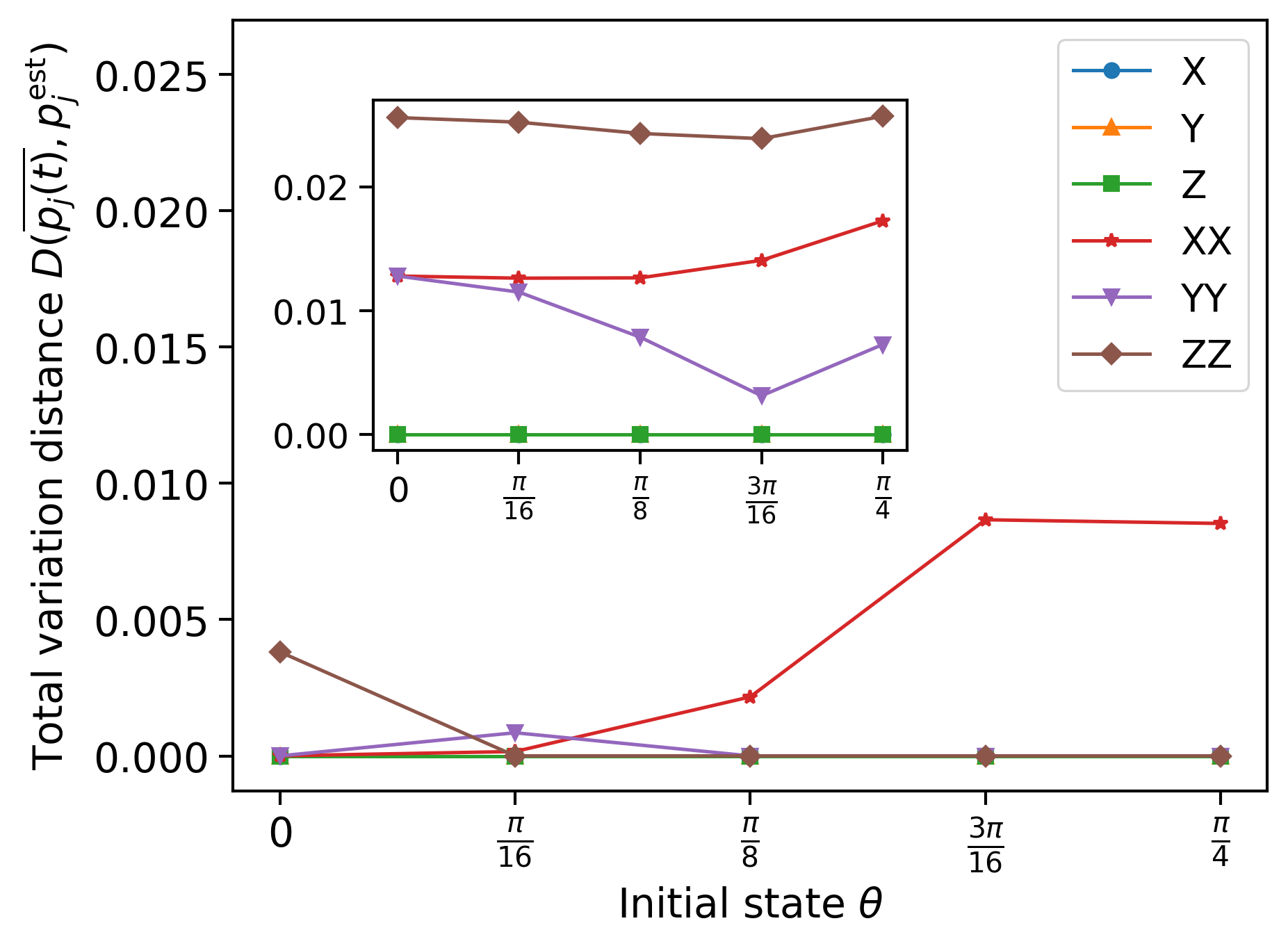}
    \caption{The total variation distance $D(\overline{p_j(t)}, p_j^\textrm{est})$ between the true probability distribution at equilibrium $\overline{p_j(t)}$ and our estimate for it, $p_j^\textrm{est}$, plotted against the initial state parametrized by an angle $\theta$. The figure includes all observables considered, for $N=16$. In the inset, we plot the same quantity but for the estimate based on the NATS $p_j^\textrm{NATS} := \Tr(\rho_\textrm{NATS} A_j)$, i.e., $D(\overline{p_j(t)}, p_j^\textrm{NATS})$. Note for simplicity we call $X \coloneqq \sigma^x_\frac{N}{2}$, $XX \coloneqq \sigma^x_\frac{N}{2} \sigma^x_{\frac{N}{2}+1}$ etc.}
    \label{fig:TVD_pred_vs_DE_with_charges_against_m_inset_nats}
\end{figure}

\begin{figure}
    \centering
    \includegraphics[width=\linewidth]{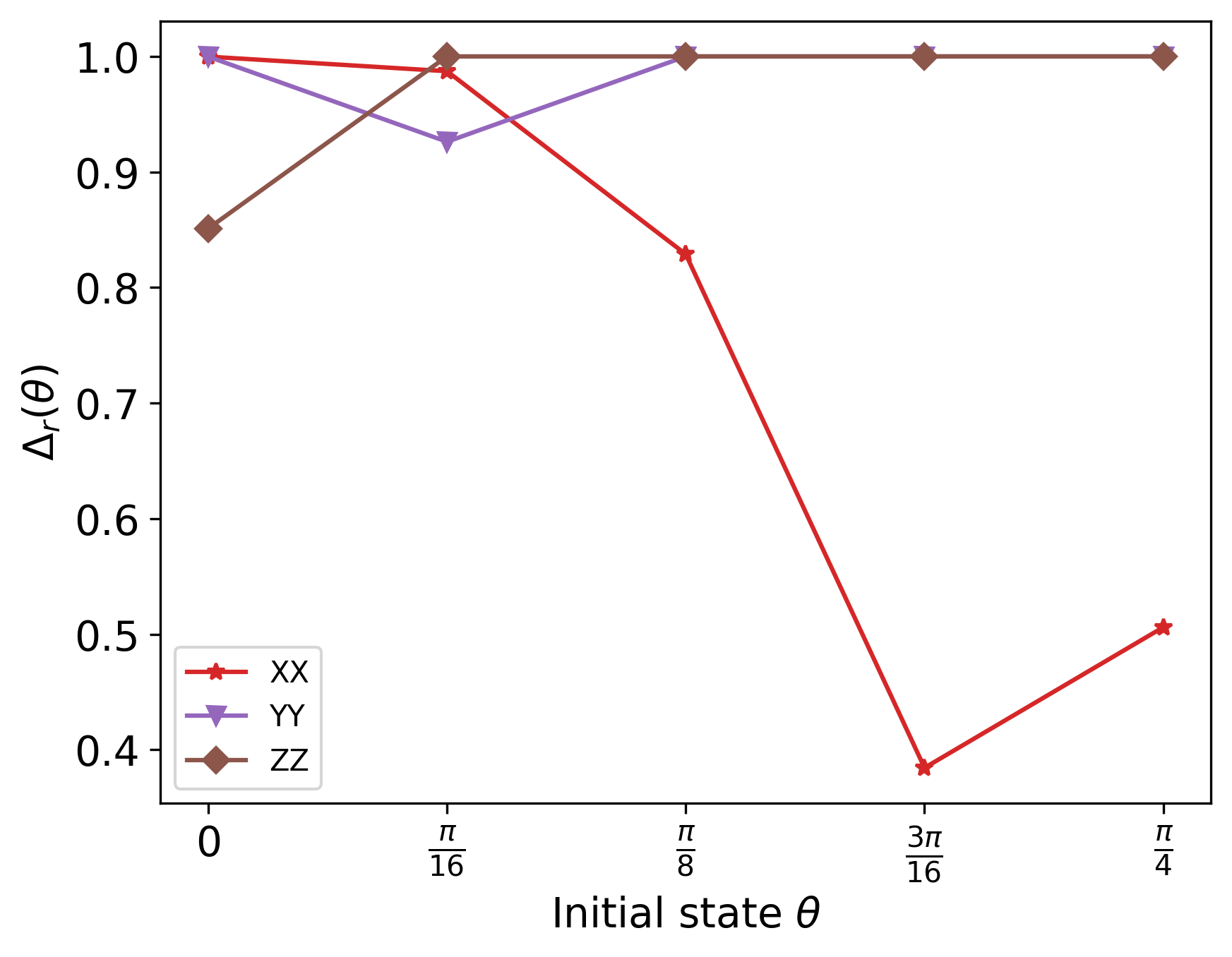}
    \caption{The relative improvement in total variation distance (TVD) $\Delta_r := 1 - \frac{D(\overline{p_j(t)}, \, p_j^\textrm{est})}{D(\overline{p_j(t)}, \, p_j^\textrm{NATS})}$ plotted against the initial state parametrized by the angle $\theta$, for two-body observables and $N=16$. This quantifies how much better we do when switching from the estimate based on the NATS (with TVD $D(\overline{p_j(t)}, p_j^\textrm{NATS})$) to our prediction Eq.~\eqref{eqn:general_solution_EEs} (with TVD $D(\overline{p_j(t)}, p_j^\textrm{est})$). $\Delta_r = 0$ means no improvement, and the higher the $\Delta_r$ the better our prediction is compared to the NATS one. Note for simplicity we call $XX \coloneqq \sigma^x_\frac{N}{2} \sigma^x_{\frac{N}{2}+1}$ etc.} 
    \label{fig:TVD_pred_vs_DE_relative_improvement_nats_vs_with_charges_obs=two}
\end{figure}

\section{Discussion} \label{sec:discussion}
 
In this work, we studied the effect of non-commuting charges on the equilibration of isolated quantum systems. We have done so via the lens of Observable Statistical Mechanics, an observable-centric framework that allows making computationally easy predictions for the stationary behavior of coarse observables. Importantly, this framework works even beyond the usual realm of applicability of equilibrium statistical mechanics, as it does not rely on the assumptions of weak coupling and locality of observables. It would be interesting to further characterize the limitations to the validity of such a framework. 

We obtained three main results. Firstly, we showed that degeneracies in the Hamiltonian are present if and only if the Hamiltonian has a non-Abelian symmetry. Physically, this means that these two aspects are equivalent, and no effect can be ascribed to only one of them. This holds in finite and infinite-dimensional systems. 
The Hamiltonian's degeneracies, caused by the non-Abelian symmetry, could influence equilibration processes, as many bounds link equilibration to the number of distinct energy levels populated~\cite{reimann_foundation_2008, short_quantum_2012, malabarba_quantum_2014}.

Secondly, we confirmed that the predictions of this framework work well even in the presence of non-commuting charges --- the error is under $10^{-10}$ for one-body observables and always under $10^{-2}$ for two-body ones. In fact, for two-body observables we show that our predictions are better than those given by the non-Abelian thermal state due to non-weak coupling. As our approach can also apply to highly degenerate subspaces of non-local observables, it would be a natural generalization to consider them in future work.

Finally, we showed that the equilibrium distribution and the necessary conditions for thermal-like behavour are characterized by weak values of the non-commuting charges. Thus, weak values and their associated Kirkwood-Dirac quasiprobability distributions, which have been widely investigated in quantum foundations and quantum thermodynamics as probes of non-classical behaviour, emerge naturally via the framework of Observable Statistical Mechanics. Studying them allows to witness non-classical behavior even at equilibrium, which would not arise if the charges commuted. It would be interesting to explore if the real parts of the weak values, which determine the equilibrium thermal-like distribution in Eq.~\eqref{eqn:general_solution_EEs}, can take anomalous values beyond the charge's expected ranges.

\begin{acknowledgments}
     L.S. thanks Shayan Majidy, Aroosa Ijaz, Samuel Slezak, Benedikt Placke, Jonathan Classen-Howes and Giuseppe Di Pietra for useful discussions. L.S. and A.T. are grateful to Robert Neagu for helpful comments on a previous version of the proof of Theorem \ref{thm:non-Abelian_degenerate_Hamiltonian}. A.T. also thanks Eliahu Cohen for useful discussions. F.A. thanks Francesco Caravelli, Sebastian Deffner and Nicole Yunger-Halpern for helpful discussions on non-commuting charges. L.S. thanks the ``Angelo Della Riccia" Foundation for their financial support during part of this project.
     L.S. also acknowledges support from the U.S. Department of Energy (DOE) through a quantum computing program sponsored by the Los Alamos National Laboratory (LANL) Information Science \& Technology Institute and from the DOE Office of Advanced Scientific Computing Research, Accelerated Research for Quantum Computing program, Fundamental Algorithmic Research for Quantum Computing (FAR-QC) project. Moreover, L.S. and V.V. acknowledge the support of the John Templeton Foundation through the ID No. 62312 grant, as part of the “The Quantum Information Structure of Spacetime” project (QISS). V.V. also thanks the Gordon and Betty Moore Foundation for their support. The opinions expressed in this publication are those of the authors and do not necessarily reflect the views of the John Templeton Foundation. A.T. was partially supported by the European Union’s Horizon Europe research and innovation programme under grant agreement No. 101178170 and by the Israel Science Foundation under grant agreement No. 2208/24. 
     L.P.G.P. acknowledges support from the Beyond Moore’s Law project of the Advanced Simulation and Computing Program at LANL,  the Laboratory Directed Research and Development (LDRD) program of LANL under project number 20230049DR, and the DOE Office of Advanced Scientific Computing Research, Accelerated Research for Quantum Computing program, Fundamental Algorithmic Research toward Quantum Utility (FAR-Qu) project.
\end{acknowledgments}

\bibliography{references}

\clearpage
\appendix
\onecolumngrid

\section*{APPENDIX}
\setcounter{secnumdepth}{1}

\renewcommand{\thesection}{A\arabic{section}}

\setcounter{section}{0}
\setcounter{equation}{0}
\renewcommand{\theequation}{A\arabic{equation}}

These appendices are organized as follows. In Sec.~\ref{sec:app_deg_H_non-Abelian_symmetry}, we prove the infinite-dimensional version of Theorem~\ref{thm:non-Abelian_degenerate_Hamiltonian}. This completes our argument that a Hamiltonian is degenerate if and only if it has a non-Abelian symmetry, or, equivalently, non-commuting charges. In Sec.~\ref{sec_app:small_mutual_info_extensive_entropy}, we discuss in more detail the bounds and heuristics arguing that the mutual information of coarse observables with both the energy and the charges is much smaller than its upper bound (the Holevo bound), and that coarse observables have extensive entropy. Some of these results had been obtained in Ref.~\cite{scarpa_observable_2025} and here we have adapted them to the new setup with a degenerate Hamiltonian, but Sec.~\ref{subsec_app:small_mutual_info_charges} is a new contribution. Such arguments guarantee the applicability of the Maximum Observable Entropy Principle. The small mutual information guarantees a few constraints suffice in entropy maximization, while the extensivity of the observable's entropy justifies the use of a maximum entropy principle in the first place. 

In Sec.~\ref{sec:AGH_theorem} we slightly generalize the main theorem in Ref.~\cite{anza_eigenstate_2018}, which states how degenerate an observable's eigensubspace needs to be so that the probability $p_{js}$ can be chosen to be constant in $s$ at equilibrium. This allows to turn non-linear expressions in $p_{js}$ into expressions in $p_j$, which is the observable's eigenvalue probability distribution and is of more physical interest. In Sec.~\ref{sec_app:generalized_EEs}, we derive the Equilibrium Equations of Observable Statistical Mechanics in the presence of non-commuting charges, and make use of the aforementioned theorem to express them in terms of $p_j$. 

Finally, in Sec.~\ref{app_sec:additional_info_numerics} we provide additional information and plots regarding our numerical simulations. These include analytical calculations for the specific model, initial states and observables used; the analytical construction of the local NATS; and a discussion on the coupling strength. In addition, we also give details of the numerical exploration we conducted to search for anomalous weak values in a larger parameter space.

\section{Degenerate Hamiltonians and non-Abelian symmetries} \label{sec:app_deg_H_non-Abelian_symmetry}
Let $\mathcal{H}$ be a complex infinite-dimensional separable Hilbert space, and let $H$ be a self-adjoint operator on $\mathcal{H}$. Naively, we would have liked to diagonalize $H$ in the usual sense, i.e. write $H = \sum_{n=1}^\infty E_n P_n$ for orthogonal projections $P_n$ ($E_m \neq E_n$ for $m \neq n$).
The projections determine an orthogonal direct sum $\mathcal{H} = \bigoplus_{n=1}^\infty \mathcal{H}_n$, where $\mathcal{H}_n$ is the eigenspace corresponding to the eigenvalue $E_n$.
The degeneracies are defined as $d_n \vcentcolon= \dim \mathcal{H}_n$.
If $H$ is bounded then we can define the symmetry group $\mathcal{G} \left( H \right)$ to comprise all unitaries that commute with $H$. The degeneracy of $H$ is equivalent to non-commutativity of $\mathcal{G} \left( H \right)$; one can show this in the same vein as in the finite-dimensional case (see Section \ref{sec:NonCommQ}).
However, the above reasoning relies on two assumptions which are always satisfied in finite dimensions, but may or may not hold in the infinite-dimensional case.

First, $H$ may not be diagonalizable in the aforementioned sense. It can be written as $H = \sum_{n=1}^\infty E_n P_n$ only if it has \textit{pure point spectrum}, meaning that every element of its spectrum is an eigenvalue. Recall that the spectrum $\sigma \left( H \right)$ is defined as the set of complex numbers $\lambda$ such that $H-\lambda I$ does not have a bounded inverse. In the finite-dimensional case every such $\lambda$ is an eigenvalue, so $H$ always has pure point spectrum; but in the infinite-dimensional case this is not necessarily true.
For example, the ``position'' operator $\hat{x}$ (multiplication by $x$) on $L^2 \left[ 0, 1 \right]$ is self-adjoint, but its spectrum $\sigma \left( \hat{x} \right) = \left[ 0, 1 \right] $ does not contain even a single eigenvalue. And yet, any element in the spectrum of $H$ physically corresponds to an energy. Thus, we have a conceptual issue: how do we define degeneracy for a Hamiltonian whose spectrum is not pure-point?

Second, if $H$ is unbounded, then the notion of commutation with $H$ becomes subtle.
Indeed, if $H$ is unbounded then it need not be well-defined on every vector $\ket{\psi} \in \mathcal{H}$; instead, it is equipped with a domain $\mathcal{D} \left( H \right)$, which is a dense subspace of $\mathcal{H}$. Now, let $U$ be a unitary operator; then it is bounded, thus everywhere-defined. The product $UH$ is defined on $\mathcal{D} \left( H \right)$, but $HU$ is defined on $U^{-1}\mathcal{D} \left( H \right) $, so it may not make sense to compare the two. Thus, it no longer makes sense to define the symmetry group $\mathcal{G} \left( H \right)$ as the set of unitaries that commute with $H$. Instead, we require a stronger notion of commutation.

Here we propose solutions for both issues. We look into the spectral theorem and identify appropriate definitions of degeneracy and symmetry.
We then state and prove the infinite-dimensional version of theorem, i.e., that degeneracy and non-Abelian symmetries are equivalent. This appendix makes use of machinery from functional analysis; see Ref.~\cite[Chapters 6-10]{hall_quantum_2013} for a good exposition on the spectral theorem. We also use further properties of direct integrals, which appear, e.g., in Ref.~\cite[Chapter IV.8]{takesaki_theory_1979}. Throughout this appendix we assume that $\mathcal{H}$ is a separable infinite-dimensional complex Hilbert space. $\mathfrak{B} \left( \mathcal{H} \right) $ denotes the set of bounded linear operators on $\mathcal{H}$.

\subsection{The spectral theorem}\label{sec:spectral}
Let $H$ be a self-adjoint operator (not necessarily bounded). The general form of the spectral theorem establishes that $\mathcal{H}$ can be written in the form of a \textit{direct integral} on the spectrum of $H$. A direct integral is somewhat analogous to a direct sum, but there are several key differences between the two. A formal definition of direct integrals is beyond the scope of this appendix, but can be found in Ref.~\cite[Chapter 7]{hall_quantum_2013}.

The spectral theorem asserts that there exists a measure $\mu$ on the spectrum $\sigma(H)$ and a family of Hilbert spaces $\left\{ \mathcal{H}_\lambda \right\}_{\lambda \in \sigma(H)} $, such that $\mathcal{H}$ is unitarily equivalent to a direct integral:
\begin{equation}\label{direct_integral_spec_H}
    \mathcal{H} \cong \int_{\sigma(H)}^{\oplus} \mathcal{H}_\lambda d\mu \left( \lambda \right) .
\end{equation}
Although we do not define direct integrals here, let us mention that the elements of $\int_{\sigma(H)}^{\oplus} \mathcal{H}_\lambda d\mu \left( \lambda \right)$ are \textit{sections}, i.e. $\mu$-measurable functions $s : X \rightarrow \bigcup_{\lambda \in \sigma(H)} \mathcal{H}_\lambda$ such that $s \left( \lambda \right) \in \mathcal{H}_\lambda$.

The spectral theorem further states that $H$ acts on the direct integral as multiplication by $\lambda$. Moreover, if we assume $\dim \mathcal{H}_\lambda > 0$ almost everywhere, then this direct integral is unique in the following sense. Suppose there is another direct integral $\int_{\sigma(H)}^{\oplus} \mathcal{\tilde{H}}_\lambda d\tilde{\mu} \left( \lambda \right) $ with the same properties as above; then the measure $\tilde{\mu}$ has the same sets of measure zero as $\mu$, and $ \dim \mathcal{\tilde{H}}_\lambda = \dim \mathcal{H}_\lambda $ almost everywhere. Therefore, $\dim \mathcal{H}_\lambda $ is uniquely determined by $H$, up to a set of measure zero. 

Intuitively, we can think of $\dim \mathcal{H}_\lambda$ as the multiplicity of the element $\lambda \in \sigma \left( H \right)$. In the special case that $\lambda$ is an eigenvalue, $\mu \left( \left\{ \lambda \right\} \right) > 0 $, and the well-defined $\dim \mathcal{H}_\lambda$ is indeed the degeneracy of the eigenvalue $\lambda$.
If $ \dim \mathcal{H}_\lambda = d $ almost everywhere for some fixed $d>0$, we say that $H$ has \textit{uniform multiplicity} $d$. These insights motivate the following definition.
\begin{definition}
Let $H$ be a self-adjoint operator on $\mathcal{H}$. We say that $H$ is \textit{non-degenerate} if $H$ has uniform multiplicity $1$.
\end{definition}
If $\sigma \left( H \right)$ consists entirely of eigenvalues (a pure-point spectrum), then this definition coincides with the conventional definition, i.e. that the eigenspaces of $H$ are all one-dimensional.

\subsection{Decomposable operators and strong-commutation}
Recall that in the finite-dimensional case, the spectral theorem determined a decomposition of $\mathcal{H}$ as a direct sum; and the operators that commute with $H$ were those which are block-diagonal with respect to this sum.
An analogous result holds in the infinite-dimensional case as well. Fix the direct integral decomposition Eq.~\eqref{direct_integral_spec_H}. We start by defining decomposable operators; intuitively, these are the block-diagonal operators with respect to the direct integral.
A bounded operator $O$ is said to be \textit{decomposable}, if there exists an essentially-bounded measurable family of operators $ \left\{ o \left( \lambda \right) \right\}_{\lambda \in \sigma (H)} $ such that $O$ acts on sections $s \in \int_{\sigma (H)}^{\oplus} \mathcal{H}_\lambda d\mu \left( \lambda \right) $ via:
\begin{equation}
    O s = \int_{\sigma (H)}^{\oplus} o \left( \lambda \right) s \left( \lambda \right) d\mu \left( \lambda \right) .
\end{equation}
In this case, we denote $ O = \int_{\sigma (H)}^{\oplus} o \left( \lambda \right) d\mu \left( \lambda \right) $. If every linear operator $o \left( \lambda \right) \in \mathfrak{B} \left( \mathcal{H}_\lambda \right) $ is a scalar operator, then we say that the operator $O$ is \textit{diagonal}. 
\begin{proposition}[Corollary IV.8.16 from Ref.~\cite{takesaki_theory_1979}]\label{prop:Takesaki_8_16}
    The decomposable operators are precisely those that commute with the algebra of all diagonal operators (the \textit{diagonal algebra}).
\end{proposition}

We now define the spectral projections of $H$. As we mentioned, points in the spectrum $\sigma \left( H \right)$ correspond to possible values of the energy. If $\lambda \in \sigma \left( H \right)$ is not an eigenvalue, then $\mu \left( \left\{ \lambda \right\} \right) = 0$, and the probability to measure the energy to be $\lambda$ vanishes. However, an energy measurement can tell us that the energy lies within some set, for example an interval: $E_{\min} \leq E \leq E_{\max}$.
More generally, let $E \subseteq \sigma \left( H \right)$ be a measurable subset of the spectrum. The corresponding \textit{spectral projection} of $H$, denoted $\mu_E$, is defined as the diagonal operator:
\begin{equation}
    \mu_E \vcentcolon= \int_{\sigma \left( H \right)}^{\oplus} \mathbb{I}_E \left( \lambda \right) d\mu \left( \lambda \right) ,
\end{equation}
where $\mathbb{I}_E \left( \lambda \right)$ is the indicator function of the set $E$. Physically, $\mu_E$ is the projection corresponding to measuring the energy to be in the set $E$. Now, let $B \in \mathfrak{B} \left( \mathcal{H} \right)$ be a bounded operator. We say that $B$ \textit{strong-commutes} with $H$ if $B$ commutes with \textit{every} spectral projection of $H$.
Let $ \mathcal{G} \left( H \right) \subseteq \mathfrak{B} \left( \mathcal{H} \right) $ denote the set of all unitary operators that strong-commute with $H$. Therefore, we think of a unitary $U$ as a symmetry of $H$ only if it commutes with all of the spectral projections $\mu_E$.

\subsection{Statement and proof of the theorem}
We now state and prove our theorem.
\begin{theorem}
    Let $H$ be a (not necessarily bounded) self-adjoint operator. Then $ \mathcal{G} \left( H \right) $ is Abelian iff $H$ has uniform multiplicity $1$.
\end{theorem}
\begin{proof}
    A von Neumann algebra is generated by its projections. The diagonal algebra is a von Neumann algebra, and the diagonal projections are precisely the spectral projections. To see this, take the square of an arbitrary diagonal operator $ O = \int_{\sigma (H)}^{\oplus} f \left( \lambda \right) d\mu \left( \lambda \right) $ (here $f$ is a real-valued measurable function). To obtain $P^2 = P$, clearly we must have $f^2 = f$ almost everywhere, hence $f \left( \lambda \right) \in \left\{ 0, 1 \right\}$ almost everywhere. Since $f$ is measurable, the set $ E = \left\{ \lambda \in \sigma (H) \mid f \left( \lambda \right) = 1 \right\} $ is measurable, hence $O = \mu_E$.
    
    Thus, the spectral projections generate the diagonal algebra. This implies that a bounded operator commutes with every spectral projection iff it commutes with the entire diagonal algebra. The symmetry group $\mathcal{G} \left( H \right)$ comprises the unitaries that commute with every spectral projection, hence the unitaries that commute with the diagonal algebra. Using Proposition \ref{prop:Takesaki_8_16}, the latter are precisely the decomposable unitary operators. 

    Now, if $\dim \mathcal{H}_\lambda = 1$ almost everywhere, then the decomposable operators are all diagonal, hence commuting. Otherwise, there exists a set $E$ of positive measure where $\dim \mathcal{H}_\lambda = d$, for some $d>1$ ($d$ may be $\infty$). Choose some identification of these Hilbert spaces, $ \mathcal{H}_\lambda \cong \mathcal{V} $ for all $\lambda \in E$. Clearly the unitary group $\mathrm{U} \left( \mathcal{V} \right)$ is non-Abelian; we would like to show that we can fit a copy of it as decomposable operators.
    For every $U \in \mathrm{U} \left( \mathcal{V} \right)$, define:
    \begin{equation}
        O_U = \int_{\sigma \left( H \right)}^{\oplus} \left[ \mathbb{I}_E \left( \lambda \right) U + \mathbb{I}_{E^c} \right] d\mu \left( \lambda \right) ,
    \end{equation}
    that is, $O_U$ acts as on $s \left( \lambda \right)$ as $U$ for any $\lambda \in E$, and as the identity operator for any $\lambda \notin E$ ($E^c$ denotes $\sigma \left( H \right) \setminus E$). $O_U$ is clearly decomposable. Moreover, for any two elements $U,V \in \mathrm{U} \left( \mathcal{V} \right)$, we have $O_{UV} = O_U O_V$. Therefore $O$ defines a faithful representation of $\mathrm{U} \left( \mathcal{V} \right)$. In particular, if $U, V$ do not commute, then $O_U, O_V$ are two non-commuting decomposable operators. This completes the proof.
\end{proof}

\newpage
\section{Coarse observables have small mutual information with the energy and charges, and extensive entropy} \label{sec_app:small_mutual_info_extensive_entropy}
 
Here, we provide heuristic arguments for coarse observables having small mutual information with the energy in the case of a degenerate Hamiltonian. For Hamiltonians with non-commuting charges that are sums of identical one-site operators, we show that coarse observables also have small mutual information with the charges. We also show that the arguments given in Ref.~\cite{scarpa_observable_2025} regarding coarse observables having extensive entropy apply under non-commuting charges, too.

In the following, we consider a coarse observable as described in Section \ref{subsec:observables}. That is, we define 
\begin{equation}
    A := \sum_{j=1}^{n_A} a_j A_j := \sum_{j=1}^{n_A} a_j \sum_{s=1}^{d_j} \ketbra{j, s}, 
\end{equation}
where $A_j$ is the eigenprojector onto the subspace corresponding to eigenvalue $a_j$, and $\{ \ket{j,s} \}_{js}$ is a complete basis with the index $s$ labelling the degeneracies for each eigensubspace. We assume the number of outcomes $n_A$ is much smaller than the Hilbert space dimension, and grows at most as a small polynomial in the system size, i.e., $n_A \leq N^k$ for some integer $k \sim O(1)$. We are interested in highly degenerate subspaces of these observables, i.e., subspaces for which $d_j \gg 1$ and in particular $d_j(d_j-1) \geq n_E + 1$, where $n_E$ is the number of distinct energy eigenvalues.

We define two Shannon entropies relative to the observable. The first one corresponds to the measurement of the complete eigenbasis $\{ |j,s \rangle \}_{j,s}$, which has equilibrium probability distribution $\{p_{js} = |\langle \lambda_n|j,s\rangle|^2\}_{js}$, while the second one is related to the coarse-grained outcomes $a_j$, with equilibrium distribution $\{p_j = \langle \lambda_n|A_j |\lambda_n\rangle =: [A_j]_{nn}\}_j$. Thus, we consider
\begin{gather}
    S_A(\{p_{js}\}) := - \sum_{j,s} p_{js} \log p_{js} \\
    S_A(\{p_{j}\}) := - \sum_j p_j \log p_j.
\end{gather}

\subsection{Small mutual information with the energy}
Consider the equilibrium state $\rho_\infty = \sum_n p(n) |\lambda_n \rangle\langle \lambda_n|$. Since $[H, |\lambda_n \rangle\langle\lambda_n|]=0$, it is natural to consider measuring the energy first. The post-measurement state is $\rho_H = \frac{\Pi_n \rho_\infty \Pi_n}{p(n)} = |\lambda_n \rangle\langle\lambda_n| = \frac{\Pi_n |\psi(0) \rangle\langle \psi(0)| \Pi_n}{p(n)}$. Subsequently, we consider a measurement in the complete eigenbasis $\{|j,s\rangle\}_{j,s}$ of the observable $A$. The equilibrium probability of measuring state $|j, s\rangle$ after the outcome $E_n$ of the Hamiltonian is $p(j,s|n) = \langle j,s|\rho_H|j,s\rangle = |\langle \lambda_n | j,s \rangle|^2 = \frac{|\langle \psi(0)|\Pi_n |j,s\rangle|^2}{p(n)}$. 

The joint probability of both measurements is $p(j,s; n) = p(n) p(j,s|n) = p(n) |\langle \lambda_n | j,s \rangle|^2 = |\langle \psi(0)|\Pi_n |j,s\rangle|^2$, and the marginals are obtained by summing over $n$, $p(j,s) = \sum_n p(n) p(j,s | n) = \sum_n p(n) |\langle \lambda_n | j,s \rangle|^2$. Since we have fixed the order of measurements, we can use Bayes' theorem to obtain the conditional probability $p(n|j,s) = \frac{p(n) p(j,s|n)}{p(j,s)} = \frac{p(n) |\langle \lambda_n | j,s \rangle|^2}{\sum_n p(n) |\langle \lambda_n | j,s \rangle|^2}$. 

Consider an alternative scenario where the coarse-grained outcomes $a_j$ of the observable are measured after the energy measurement. In this case, we obtain $p(j|n) = \langle \lambda_n| A_j |\lambda_n\rangle =: [A_j]_{nn}$, $p(j;n) = p(n) p(j|n) = p(n) [A_j]_{nn}$, $p(j) = \sum_n p(n) p(j|n) = \sum_n p(n) [A_j]_{nn}$ and finally $p(n|j) = \frac{p(n) [A_j]_{nn}}{\sum_n p(n) [A_j]_{nn}}$. 

In all subspaces such that $d_j(d_j-1) \geq n_E + 1$ we can apply Theorem 1 in Ref.~\cite{anza_eigenstate_2018} (see also Appendix~\ref{sec:AGH_theorem}), which implies that $|\langle \lambda_n |j,s\rangle|^2 = \frac{[A_j]_{nn}}{d_j}$. We thus have $p(j,s|n) = \frac{p(j|n)}{d_j}$, $p(j,s;n) = \frac{p(j;n)}{d_j}$, $p(j,s) = \frac{p(j)}{d_j}$ and $p(n|j,s) = p(n|j)$. 

The mutual information between two random variables $X$ and $Y$ is defined as 
\begin{equation}
    I(X, Y) := \sum p(X;Y) \log \left( \frac{p(X;Y)}{p(X) p(Y)} \right),
\end{equation}
where $p(X;Y)$ is the joint distribution, and $p(X)$ and $p(Y)$ are the marginals. Define $\tilde{I}_{eq}(H, A)$ and $I_{eq}(H, A)$ to be respectively the mutual information calculated with respect to the measurement of the observable's complete eigenbasis and the observable's coarse-grained outcomes. Then, if the aforementioned theorem can be applied, we find $\tilde{I}_{eq}(H, A) = I_{eq}(H, A)$. We can now apply the information-theoretic interpretation in terms of communication channels, which was discussed in Section \ref{subsec:MOEP_constraints}, to $I_{eq}(H, A)$. We have that $\frac{I_{eq}(H, A)}{S_H(\{p_n\})} \leq \frac{S_A(\{p_j\})}{S_H(\{p_n\})} \leq \frac{\log n_A}{S_H(\{p_n\})} \lesssim \frac{\log N}{N}$, where we have assumed $S_H(\{p_n\}) \sim N$ for realistic initial states when the Hamiltonian is not too highly degenerate, i.e., $n_E \sim D$. Hence, the mutual information is much smaller than its Holevo bound: $I_{eq}(H, A) \ll S_H(\{p_n\})$.

\subsection{Small mutual information with the charges} \label{subsec_app:small_mutual_info_charges}
Consider the spectral decomposition of a charge 
\begin{equation}
    Q^a := \sum_\mu q^a_\mu Q^a_\mu,
\end{equation}
where $\{q^a_\mu\}_\mu$ are the charge's distinct eigenvalues and $\{Q^a_\mu\}_\mu$ are the coarse-grained eigenprojectors.

Once again, we consider the equilibrium state $\rho_\infty = \sum_n p(n) |\lambda_n \rangle\langle \lambda_n|$. Let us start by first measuring the charge $Q^a$. The state after measurement is $\rho_{Q^a} = \frac{Q^a_\mu \rho_\infty Q^a_\mu}{p(\mu)}$ with $p(\mu) = \Tr(\rho_\infty Q^a_\mu)$. We subsequently measure the complete eigenbasis $\{|j,s\rangle\}_{j,s}$ of the observable $A$. The probability of measuring state $|j, s\rangle$ after the outcome $q^a_\mu$ of the charge is $p(j,s|\mu) = \langle j,s|\rho_{Q^a}|j,s\rangle = \frac{\langle j,s|Q^a_\mu \rho_\infty Q^a_\mu|j,s\rangle}{p(\mu)}$. We can now also compute the joint probability $p(j,s; \mu) = p(\mu) p(j,s|\mu) = \langle j,s|Q^a_\mu \rho_\infty Q^a_\mu|j,s\rangle$ and the marginal by summing over $\mu$, $p(j,s) = \sum_\mu p(\mu) p(j,s | \mu) = \sum_\mu \langle j,s|Q^a_\mu \rho_\infty Q^a_\mu|j,s\rangle$. We can also repeat the same calculations but inverting the order of measurements.

Consider a one-dimensional spin-1/2 chain of $N$ lattice sites, which evolves under a Hamiltonian with non-commuting charges of the form $Q^a = \sum_{i=1}^N Q_{(i)}^a$, where the $\{Q_{(i)}^a\}$ are identical one-site operators, each acting on a different lattice site. Assume every $Q_{(i)}^a$ has the same $d$ eigenvalues $V:=\{v_k\}_{k=1}^d$. Then, each $Q^a$ has $n_{Q^a} \in O(N^{d-1})$ distinct outcomes. To see this, note that we draw $N$ values (with repetition) from $V$, so we have $d^N$ possible combinations $(x_1, x_2, \ldots, x_N)$. We want to find the number of distinct values that the sum $X= \sum_{i=1}^N x_i$ takes. Define the number of times a given eigenvalue appears in a combination as $n_k :=|\{i: x_i=v_k\}|$ with $n_k \geq 0$ and $\sum_k n_k = N$. Then, $X = \sum_{k=1}^d n_k v_k$ and each weak composition $(n_1,\dots,n_d)$ of $N$ gives rise to one sum. The number of these weak compositions, i.e., the number of non-negative integer solutions to $\sum_k n_k = N$, is given by $\binom{N + d - 1}{d - 1} \in O(N^{d-1})$, and hence $S_{Q^a}(\{p_\mu\}) \lesssim (d-1)\log N$. In particular, if the eigenvalues in $V$ are equally spaced, that is, $v_k - v_{k-1}= \xi$, then $v_k = v_1 + (k-1) \xi$ and $X = \sum_{i=1}^N x_i = N v_1 + \xi \sum_{i=1}^N (k_i - 1)$. The second term in the sum can take integer values in $[0, N(d-1)]$ since each $k_i-1$ runs over the integers $\{0,1,\dots,d-1\}$. Hence, the number of distinct possible sums is $N(d-1) + 1 \in O(N)$. For instance, this is true for the SU(2)-symmetric model we study in the numerics, where the non-commuting charges are the total magnetizations $Q^a := \frac{1}{N} \sum_{i=1}^N \sigma^a_i$ with $a =x, y, z$, which have $N+1$ eigenvalues. 

We can now show that the mutual information is small compared to its Holevo bound:
\begin{equation}
    \frac{\tilde{I}(Q^a, A)}{S_A(\{p_{js}\})} \leq \frac{S_{Q^a}(\{p_\mu\})}{S_A(\{p_{js}\})} \lesssim \frac{\log N}{N},
\end{equation}
where we have assumed that $S_A(\{p_{js}\})$ is extensive; see the next subsection for a detailed justification of this assumption. Moreover, we note that coarse-graining the outcomes of a variable is just a form of post-processing, and thus we have $I(Q^a, A) \leq \tilde{I}(Q^a, A)$ thanks to the data processing inequality \cite{nielsen_quantum_2010}.

\subsection{Extensive entropy}
For completeness, we summarize here the arguments given in Ref.~\cite{scarpa_observable_2025} for the extensivity of the Shannon entropy of the observable's complete eigenbasis $S_A(\{p_{js}\})$. Firstly, by the Schur-Horn theorem and the Schur-concavity of the entropy, we have that $S_A(\{p_{js}\}) \geq S_H(\{p_n\})$, so at equilibrium $S_A(\{p_{js}\})$ is extensive whenever $S_H(\{p_n\})$ is. This is true for any observable at equilibrium. Moreover, one can show that when $\exists \, j: \max_s p_{js} = \max_{j's} p_{j's} \land d_j(d_j-1) \geq D+1$, we have $S_A(\{p_{js}\}) \geq \min_j \log d_j$. In particular, this is the case for observables with constant degeneracy $d^*$ such that $d^*(d^*-1) \geq D+1$. If such observables have support on $N_S$ sites, then $\log d^* \sim N-N_S$, so $S_A(\{p_{js}\}) \gtrsim N - N_S$ at equilibrium, independently of the initial state. For a more detailed discussion, we refer the reader to Ref.~\cite{scarpa_observable_2025}.

\newpage
\section{Generalization of the Anza-Gogolin-Huber theorem}
\label{sec:AGH_theorem}
In Ref.~\cite{anza_eigenstate_2018}, Anza, Gogolin and Huber showed how degenerate an observable's eigensubspace has to be to allow choosing a complete orthonormal basis within that subspace that it is unbiased with respect to a set of orthonormal vectors on the full Hilbert space. This latter set could be taken to be the energy eigenbasis. Here, we generalize this theorem to consider unbiasedness with respect to a set of orthogonal projectors of rank greater or equal to 1 instead of just a set of orthonormal vectors. We now state the generalization and prove it.

\begin{theorem}[Generalization of the Anza-Gogolin-Huber theorem]
    Let $\{\Pi_n\}_{n=1}^{M} \subset \mathcal{B}(\mathcal{H})$ be a set of $M$ rank-$d_n$ orthogonal projectors in the space $\mathcal{B}(\mathcal{H})$ of bounded linear operators on a Hilbert space $\mathcal{H}$ of dimension $D$. Let $A := \sum_{j=1}^{n_A} a_j A_j$ be an operator on $\mathcal{H}$ with $n_A \leq D$ distinct eigenvalues $a_j$ and corresponding eigenprojectors $A_j$. Decompose $\mathcal{H} = \bigoplus_{j=1}^{n_A} \mathcal{H}_j$ into a direct sum such that each $\mathcal{H}_j$ is the image of the corresponding $A_j$ with dimension $d_j$. For each $j$ for which $d_j(d_j -1) \geq M + 1$ there exists an orthonormal basis $\{|j,s\rangle\}_{s=1}^{d_j} \subset \mathcal{H}_j$ such that for all $s$, $n$, 
    \begin{equation} \label{eqn:thm_AGH_app}
        \langle j, s| \Pi_n |j,s \rangle = \frac{\Tr(\Pi_n A_j)}{d_j}.
    \end{equation}
\end{theorem}

\begin{proof}
    The proof follows closely the original argument by Anza, Gogolin and Huber \cite{anza_eigenstate_2018}. We begin by noting that we can always decompose a rank-$d_n$ projector $\Pi_n$ in terms of $d_n$ rank-1 projectors $\{\Pi_{nr}\}_{r=1}^{d_n}$, i.e., $\Pi_n = \sum_{r=1}^{d_n} \Pi_{nr}$. Let us consider the normalized projection of each $\Pi_{nr}$ and $\Pi_n$ onto the eigen-subspace $\mathcal{H}_j$, namely,
    \begin{gather}
        \Pi_{nr}^{(j)} := \frac{A_j \Pi_{nr} A_j}{\Tr(A_j \Pi_{nr})}, \qquad 
        \Pi_n^{(j)} := \frac{A_j \Pi_n A_j}{\Tr(A_j \Pi_n)} = \sum_{r=1} ^{d_n} \frac{\Tr(A_j \Pi_{nr})}{\Tr(A_j \Pi_n)} \Pi_{nr}^{(j)}.
    \end{gather}
    Now, given a vector $|\varphi\rangle \in \mathcal{H}_j$, we can write $\langle \varphi| \Pi_n | \varphi\rangle = \langle \varphi| A_j \Pi_n A_j | \varphi\rangle = \Tr(A_j \Pi_n) \langle \varphi| \Pi_n^{(j)} | \varphi\rangle = \sum_{r=1} ^{d_n} \Tr(A_j \Pi_{nr}) \langle \varphi| \Pi_{nr}^{(j)} | \varphi\rangle$. Thanks to the isomorphism $\mathcal{H} \cong \mathbb{S}^{D^2-2}$ between a $D$-dimensional separable complex Hilbert space $\mathcal{H}$ and the unit $(D^2-2)$-sphere $\mathbb{S}^{D^2-2}$, one can associate to every normalized rank-1 projector $|\psi\rangle\langle \psi|$ (with $\ket{\psi} \in \mathcal{H}$) a generalized Bloch vector $\Vec{b}(\psi) \in \mathbb{S}^{D^2-2} \subset \mathbb{R}^{D^2-1}$ satisfying
    \begin{equation}
        |\psi \rangle\langle \psi| = \frac{\mathbb{I}}{D} + \sqrt{\frac{D-1}{D}} \Vec{b}(\psi) \cdot \Vec{\gamma},
    \end{equation}
    where $\Vec{\gamma}$ is a vector with entries $\gamma_i := \hat{\gamma}_i/\sqrt{2}$ and $\hat{\gamma}_i$ are the $D^2-1$ generators of SU($D$) \cite{bengtsson_geometry_2006, anza_eigenstate_2018}.
    Using the generalized Bloch vector parametrization, we can express $\langle \varphi| \Pi_{nr}^{(j)} | \varphi\rangle$ 
    as 
    \begin{equation}
        \langle \varphi| \Pi_{nr}^{(j)} | \varphi\rangle = \frac{1}{d_j} + \frac{d_j - 1}{d_j} \Vec{b}_{nr}^{(j)} \cdot \Vec{b},
    \end{equation}
    where $\Vec{b}_{nr}^{(j)}$ and $\Vec{b}$ are the (generalized) Bloch vectors associated to $\Pi_{nr}^{(j)}$ and $|\varphi\rangle\langle \varphi|$ respectively. Thus, we obtain
    \begin{equation}
        \langle \varphi| \Pi_n | \varphi\rangle = \frac{\Tr(\Pi_n A_j)}{d_j} + \frac{d_j - 1}{d_j} \Vec{b}_n^{(j)} \cdot \Vec{b},
    \end{equation}
    with $\Vec{b}_n^{(j)} := \sum_{r=1} ^{d_n} \Tr(A_j \Pi_{nr}) \Vec{b}_{nr}^{(j)}$. To obtain Eq.~\eqref{eqn:thm_AGH_app} we therefore need to impose two types of conditions. Firstly, we need to require that the generalized Bloch vectors are orthogonal for every $n$, i.e., $\Vec{b}_n^{(j)} \cdot \Vec{b} \mbeq 0 \; \forall n$. Therefore, given $\Vec{b}_n^{(j)}, \Vec{b} \in \mathbb{S}^{d_j^2-2}$, imposing these $M$ conditions leaves a subspace of dimension $d_j^2 - 2 - M$ for choosing $\Vec{b}$. Secondly, in this remaining subspace we want to pick $d_j$ generalized Bloch vectors such that their corresponding state vectors $\{ |j,s\rangle \}_{s=1}^{d_j}$ form a complete orthonormal basis for $\mathcal{H}_j$. Thus, we must have
    \begin{gather} 
        \mathbb{I} \mbeq \sum_s |j,s \rangle\langle j,s| = \mathbb{I} + \sqrt{\frac{d_j-1}{d_j}} \left( \sum_s \Vec{b}_{js} \right) \cdot \Vec{\gamma} \quad \Longleftrightarrow \quad \sum_s \Vec{b}_{js} \mbeq 0 \label{eqn:conditions_completeness_bloch_vectors_app} \\
        \delta_{ss'} \mbeq |\langle j,s| j,s' \rangle|^2 = \frac{1}{d_j} + \frac{d_j - 1}{d_j} \Vec{b}_{js} \cdot \Vec{b}_{js'} \quad \Longleftrightarrow \quad \Vec{b}_{js} \cdot \Vec{b}_{js'} \mbeq \frac{d_j}{d_j - 1} \delta_{ss'} - \frac{1}{d_j - 1}. \label{eqn:conditions_orthonormality_bloch_vectors_app}
    \end{gather}
    If the remaining subspace is large enough to host a $(d_j -1)$-dimensional regular simplex, the Minkowski-Weyl theorem \cite{anza_eigenstate_2018} guarantees finding $d_j$ generalized Bloch vectors $\Vec{b}_{js}$ that obey the conditions in Eqs.~\eqref{eqn:conditions_completeness_bloch_vectors_app} and \eqref{eqn:conditions_orthonormality_bloch_vectors_app} by choosing them to be the facet vectors of such a simplex. We thus require $d_j^2 - 2 - M \geq d_j - 1$, equivalently $d_j(d_j - 1) \geq M + 1$.
\end{proof}

\begin{remark}
    If the projectors $\{ \Pi_n \}$ are rank-1, i.e., $d_n=1 \; \forall n$, we retrieve the  theorem by Anza, Gogolin and Huber~\cite{anza_eigenstate_2018}. 
\end{remark}
\newpage

\section{The generalized Equilibrium Equations and their solutions} \label{sec_app:generalized_EEs}

\subsection{Derivation of the generalized Equilibrium Equations}
We here derive the generalization of the Equilibrium Equations of Observable Statistical Mechanics \cite{scarpa_observable_2025, anza_information-theoretic_2017, anza_pure_2018} to systems evolving under a degenerate Hamiltonian with non-commuting charges.\\

\noindent We consider the setup introduced in Section \ref{sec:setup}, but allow for mixed states $\rho = \sum_\alpha t_\alpha \ketbra{\psi_\alpha}$ with $\sum_\alpha t_\alpha = 1$. Thus, we consider a degenerate Hamiltonian $H$ and an observable $A := \sum_{j=1}^{n_A} a_j A_j$ with $n_A$ distinct eigenvalues $a_j$ and coarse-grained eigenprojectors $A_j := \sum_{s=1}^{d_j} A_{js} := \sum_{s=1}^{d_j} |j, s \rangle\langle j,s|$ of rank $d_j$. For future convenience, we define the overlaps $D_{js}^\alpha := \braket{j,s}{\psi_\alpha}$, their complex conjugates $\overline{D}_{js}^\alpha$, and 
\begin{align}
R_{js} &\coloneqq \frac{1}{2} \Tr(\rho \{ A_{js}, H \}) , \qquad
    R_j \coloneqq \sum_s R_{js} = \frac{1}{2} \Tr(\rho \{ A_j, H \}).
\end{align} 
Further, we consider $n_Q$ (linearly independent) non-commuting charges $\{Q^a\}_{a=1}^{n_Q}$ which are such that they all commute with the Hamiltonian, i.e., $[H, Q^a] = 0 \; \forall a$, but do not commute with each other, namely, $[Q^a, Q^b] \neq 0 \; \forall a, b : a \neq b$. For each of the charges, we define 
\begin{align}
    R_{js}^a \coloneqq \frac{1}{2} \Tr(\rho \{ A_{js}, Q^a \}), \qquad R_j^a \coloneqq \sum_s R_{js}^a = \frac{1}{2} \Tr(\rho \{ A_j, Q^a \}).
\end{align}

We want to maximize the observable's Shannon entropy
\begin{equation}
    \tilde{S}_A[\rho] \coloneqq -\sum_j p_{js} \log p_{js} = - \sum_{\alpha,js} t_\alpha \abs{D_{js}^\alpha}^2 \log( \sum_\beta t_\beta \abs{D_{js}^\beta}^2 ),
\end{equation}
where $p_{js} \coloneqq p(\ket{j,s})$ is the probability distribution of the complete basis $\{|j,s\rangle\}_{js}$ diagonalizing all our measurements. 

Besides imposing the usual constraints of state normalization and fixed average energy, here we also include information about the non-commuting charges by fixing their average values. The constraints take the form
\begin{align}
    &\mathcal{C}_N := \Tr(\rho) - 1 = \sum_{\alpha js} t_\alpha \abs{D_{js}^\alpha}^2 - 1 \mbeq 0, \\
    &\mathcal{C}_E := \Tr(\rho H) - E = \sum_{\alpha jj' ss'} t_\alpha \overline{D}_{js}^\alpha H_{js,j's'} D_{j's'}^\alpha - E \mbeq 0, \\
    &\mathcal{C}_{Q^a} := \Tr(\rho Q^a) - q^a = \sum_{\alpha jj' ss'} t_\alpha \overline{D}_{js}^\alpha (Q^a)_{js,j's'} D_{j's'}^\alpha - q^a \mbeq 0 \quad \forall a,
\end{align}
where we used the shorthand notations $H_{js,j's'} := \mel{j,s}{H}{j',s'}$ and $(Q^a)_{js,j's'} := \mel{j,s}{Q^a}{j',s'}$. The choice of constraints is justified in Section \ref{subsec:MOEP_constraints}. Note that strictly speaking we have two more constraints: the positivity of the density matrix and the unitarity of the $D_{js}^{(n)}$. However, we can simply discard solutions not compatible with these constraints.

We solve this constrained maximization problem using the Lagrange multipliers technique \cite{anza_information-theoretic_2017, riley_mathematical_2006, reichl_modern_2016}. To each of the constraints $\mathcal{C}_N$, $\mathcal{C}_E$ and $\{ \mathcal{C}_{Q^a} \}_a$ corresponds a Lagrange multiplier $\lambda_N$, $\beta_A$ and $\{ \mu^a_A \}_a$ respectively. We can thus define the auxiliary function
\begin{equation}
    \tilde{\Lambda}_A[\rho; \lambda_N, \beta_A, \{ \mu^a_A \}] := \tilde{S}_A[\rho] + \lambda_N \mathcal{C}_N + \beta_A \mathcal{C}_E + \sum_a \mu^a_A \mathcal{C}_{Q^a}.
\end{equation}
The constrains are enforced by imposing the derivatives with respect to the Lagrange multipliers to be zero, i.e., $\frac{\delta \tilde{\Lambda}_A}{\delta \lambda_N} = \frac{\delta \tilde{\Lambda}_A}{\delta \beta_A} = \frac{\delta \tilde{\Lambda}_A}{\delta \mu^a_A} \mbeq 0 \; \forall a$. We can vary the auxiliary function with respect to the overlaps $D_{js}^\alpha$ and $\overline{D}_{js}^\alpha$, and the statistical coefficients $t_\alpha$. However, only two sets of these equations are independent, namely $\frac{\delta \tilde{\Lambda}_A}{\delta D_{js}^\alpha} = 0$ and $\frac{\delta \tilde{\Lambda}_A}{\delta \overline{D}_{js}^\alpha} = 0$. We thus take functional derivatives of the Shannon entropy and of the constraints with respect to $D_{js}^\alpha$ and $\overline{D}_{js}^\alpha$. We obtain
\begin{align}
\begin{split}
    \frac{\delta \tilde{\Lambda}_A}{\delta D_{js}^\alpha} &= \frac{\delta \tilde{S}_A}{\delta D_{js}^\alpha} + \lambda_N \frac{\delta \mathcal{C}_N}{\delta D_{js}^\alpha} + \beta_A \frac{\delta \mathcal{C}_E}{\delta D_{js}^\alpha} + \sum_a \mu^a_A \frac{\delta \mathcal{C}_{Q^a}}{\delta D_{js}^\alpha} = \\
    &= - t_\alpha \overline{D}_{js}^\alpha \left( 1 + \log( \sum_\beta t_\beta \abs{D_{js}^\beta}^2 ) \right) + \lambda_N t_\alpha \overline{D}_{js}^\alpha\\ 
    &\quad + \beta_A \sum_{j's'} t_\alpha \overline{D}_{j's'}^\alpha H_{j's',js} + \sum_a \mu^a_A \sum_{j's'} t_\alpha \overline{D}_{j's'}^\alpha (Q^a)_{j's',js} \mbeq 0
\end{split} \\
    \frac{\delta \tilde{\Lambda}_A}{\delta \overline{D}_{js}^\alpha} &= \overline{\frac{\delta \tilde{\Lambda}_A}{\delta D_{js}^\alpha}} \mbeq 0.
\end{align}
We now take the linear combinations
\begin{gather}
    \frac{\delta \tilde{\Lambda}_A}{\delta D_{js}^\alpha} D_{js}^\alpha -  \overline{D}_{js}^\alpha \frac{\delta \tilde{\Lambda}_A}{\delta \overline{D}_{js}^\alpha} \mbeq 0, \\
    \frac{1}{2} \left( \frac{\delta \tilde{\Lambda}_A}{\delta D_{js}^\alpha} D_{js}^\alpha + \overline{D}_{js}^\alpha \frac{\delta \tilde{\Lambda}_A}{\delta \overline{D}_{js}^\alpha} \right) \mbeq 0.
\end{gather}
Finally, we sum both equations over $\alpha$ to obtain
\begin{gather}
    \beta_A \Tr(\rho [A_{js}, H]) + \sum_a \mu^a_A \Tr(\rho [A_{js}, Q^a]) \mbeq 0, \label{eqn:EE1_pjs_app} \\
    - p_{js} \log p_{js} \mbeq \lambda_N p_{js} + \beta_A R_{js} + \sum_a \mu^a_A R_{js}^{a}, \label{eqn:EE2_pjs_app}
\end{gather}
where we have also re-defined the Lagrange multipliers $\lambda_N \to 1 -\lambda_N$, $\beta_A \to -\beta_A$ and $\mu^a_A \to -\mu^a_A \; \forall a$.

\subsection{Equilibrium Equations for highly degenerate observables} 
We have obtained two sets of equilibrium equations for $p_{js}$, but we are interested in the coarse-grained probability distribution $p_j$ as this represents the distribution of measurement outcomes. In fact, for highly degenerate observables we can obtain equilibrium equations for $p_j$ by applying Theorem 1 from Ref.~\cite{anza_eigenstate_2018}. Given such an observable $A$ and a set of orthonormal vectors $\{\ket{\psi_n}\}_{n=1}^M$, this theorem states that for each $j$ for which $d_j(d_j-1) \geq M + 1$, we can choose an orthonormal eigenbasis $\{ \ket{j,s} \}_{js}$ of $A$ such that for all $s, n$ we have 
\begin{equation}
    |\braket{\psi_n}{j,s}|^2 = \frac{\expval{A_j}{\psi_n}}{d_j}.
\end{equation}
Thus, Theorem 1 from Ref.~\cite{anza_eigenstate_2018} quantifies the degeneracy needed to find an orthonormal basis within a degenerate eigenspace such that the basis is also mutually unbiased with respect to a set of vectors, i.e., such that each basis element in the subspace has the same overlap with each vector in the set. A full statement and proof of the theorem is available in Ref.~\cite{anza_eigenstate_2018} or in Appendix~\ref{sec:AGH_theorem}, where we have also generalized the theorem to a set of orthogonal projectors $\{ \Pi_n \}_{n=1}^M$ of rank $d_n$.
By applying the theorem using the set of vectors $\{\ket{\lambda_n}\}_{n=1}^{n_E}$, we have that at equilibrium, i.e., on the $\rho_\infty$, the two distributions are related in the following way:
\begin{equation} \label{eqn:pjs=pj_over_dj}
    p_{js}^\infty = \sum_n p_n |\braket{\lambda_n}{j,s}|^2 = \sum_n p_n \frac{\expval{A_j}{\lambda_n}}{d_j} = \frac{p_j^\infty}{d_j}.
\end{equation}
We remark that while in the case of a non-degenerate Hamiltonian we need to consider the relation between $\{\ket{j,s}\}_{s=1}^{d_j}$ and the complete basis $\{E_n\}_{n=1}^D$ for every $j$, here we just need to consider the set of $n_E < D$ orthonormal vectors $\{\ket{\lambda_n}\}_{n=1}^{n_E}$. Thus, the condition of the theorem will be satisfied also by observables that are less degenerate than those required in the case of a non-degenerate Hamiltonian. 

Hence, to obtain the equilibrium equations for $p_j$, we can first sum both sets of equations over $s$ and then use Eq.~\eqref{eqn:pjs=pj_over_dj} to find
\begin{gather}
    \beta_A \Tr(\rho [A_j, H]) + \sum_a \mu^a_A \Tr(\rho [A_j, Q^a]) \mbeq 0, \label{eqn:EE1_app} \\
    - p_j \log p_j + p_j \log d_j \mbeq \lambda_N p_j + \beta_A R_j + \sum_a \mu^a_A R_j^{a}. \label{eqn:EE2_app}
\end{gather}
This proves the First and Second Equilibrium Equations in the main text.

\newpage
\section{Additional details about the numerics} \label{app_sec:additional_info_numerics}
\subsection{Analytical calculations for the non-integrable XXX model} \label{app_subsec:analytical_calculations_non-integrable_XXX_model}
Consider the Hamiltonian
\begin{equation} \label{eqn:hamiltonian_def_app}
    H = \sum_{a=x,y,z} \left\{ \sum_{i=1}^{N-1} J_i \sigma^a_i \sigma^a_{i+1} + \sum_{i=1}^{N-2} J_i \sigma^a_i \sigma^a_{i+2} \right\},
\end{equation}
with open boundary conditions and $J_i = 1 + \epsilon \, \delta_{i,3}$ with $\epsilon=0.3$, following Ref.~\cite{lasek_numerical_2024}. \\
Due to the SU(2)-symmetry of the model, the non-commuting charges are
\begin{equation} \label{eqn:def_charges_app}
    Q^a := \frac{1}{N} \sum_{i=1}^N \sigma^a_i \quad \forall a \in \{x,y,z\}.
\end{equation}
We consider a one-parameter family of initial states given by
\begin{equation} \label{eqn:psi_theta_m_app}
    \ket{\psi_0(\theta)} = R_y(\theta) \otimes R_y(-\theta) \otimes \ldots \otimes R_y(\theta) \otimes R_y(-\theta) \ket{01\ldots01},
\end{equation}
where $R_y(\theta) = \exp(-iY\theta/2)$ is the rotation operator along the $y$ axis, and the rotations alternate between angles $\theta$ and $- \theta$. We note that this rotation is chosen to be SU(2)-breaking, so that we can explore different average energy subspaces. We look at five points: $\theta = \{0, \frac{\pi}{16}, \frac{\pi}{8}, \frac{3\pi}{16}, \frac{\pi}{4}\}$, avoiding $\theta = \pi/2$ because it is a ground state of the Hamiltonian. 

\subsubsection{Averages and standard deviations of the conserved quantities} \label{subsec:average_std_conserved_quantities}
The energy average and variance are respectively defined as $E(\theta) := \Tr(\rho_0(\theta) H)$ and $\Delta E^2(\theta) := \Tr(\rho_0(\theta) H^2) - E^2(\theta)$. For the charges, we analogously define $q^a(\theta) := \Tr(\rho_0(\theta) Q^a)$ and $\Delta^2 q^a(\theta) := \Tr(\rho_0(\theta) (Q^a)^2) - (q^a)^2(\theta)$ for $a \in \{x, y, z\}$.
We compute explicitly the averages and variances of the conserved quantities for the family of initial states parametrized by $\theta$:
\begin{equation}
    \frac{E(\theta)}{N} = \frac{2(N-1+\epsilon)\sin^2(\theta) -1}{N}
\end{equation}
\begin{equation} \label{eqn:energy_std_app}
      \frac{\Delta E^2(\theta)}{N^2} =
      \frac{4 \cos^2\theta}{N^2} [ (1 + 3\sin^2\theta)N
      - 5\sin^2\theta - 1
      + 2\epsilon (1 + 3\sin^2\theta)
      + \epsilon^2 (1 + \sin^2\theta)
      ]
\end{equation}
\begin{align} 
    &q^x(\theta) = \sin(\theta); \qquad
    && \Delta^2 q^x(\theta) = \frac{\cos^2(\theta)}{N} \label{eqn:charge_avg_std_x_app}\\
    &q^y(\theta) = 0; 
    && \Delta^2 q^y(\theta) = \frac{1}{N} \label{eqn:charge_avg_std_y_app} \\
    &q^z(\theta) = \frac{\cos(\theta)}{N} \llbracket \text{N is odd} \rrbracket; 
    && \Delta^2 q^z(\theta) = \frac{\sin^2(\theta)}{N} \label{eqn:charge_avg_std_z_app}
\end{align}
where the Iverson bracket is defined as
\begin{equation*}
    \llbracket P \rrbracket \coloneqq 
    \begin{cases}
        1 &\; \text{if $P$ is true} \\
        0 &\; \text{otherwise}
    \end{cases}
\end{equation*}

Note that here we have defined the charges divided by the system size $N$, as in Eq.~\eqref{eqn:def_charges_app}. If we instead used the definition $Q^a := \sum_{i=1}^N \sigma^a_i$, then the expectation values $q^a$ shown here would need to be multiplied by $N$, and similarly the variances $\Delta^2 q^a$ by a factor of $N^2$. Hence, we indeed have $q^a \sim N$ and $\Delta q^a \sim \sqrt{N}$.

\subsubsection{Predicting \texorpdfstring{$R_j^a$}{Rₐⱼ}: a special case}
We briefly note a special case for which we can obtain the value of $R_j^z(t)$ analytically. The first initial state, i.e., $|\psi_0(\theta=0)\rangle = |01\ldots01\rangle$, is an eigenstate of $Q^z$ with eigenvalue $0$. As $[Q^z, H]=0$, we have that $R_j^z(t) = \frac{1}{2} \langle \psi_0(0)| e^{iHt} \{A_j, Q^z\} e^{-iHt} |\psi_0(0)\rangle) = 0$, and thus $q^z_j = 0$, for any observable and any time. More generally, whenever the initial state is an eigenstate of a charge, we have that $R_j^a = \lambda^a_\psi \, p_j$ with $Q^a |\psi_0\rangle = \lambda^a_\psi |\psi_0\rangle$.

\subsection{Additional plots} \label{app_subsec:additional_plots}

\subsubsection{Time evolution of the Shannon entropy}
In this subsection, we give more details about the numerical results. We begin by discussing the maximization of the observables' entropies. As discussed in Sec.~\ref{sec:generalizing_obs_stat_mech}, the Maximum Observable Entropy Principle is based on the constrained maximization of $S_A(\{p_{js}\})$. However, for an observable $A$ such that $d_j = d^* \; \forall j$ and $d^*(d^*-1) \geq n_E + 1$ (so we can apply the aforementioned Theorem 1 in Ref.~\cite{anza_eigenstate_2018}), at equilibrium we have that 
\begin{equation}
    S_A(\{p_{js}\}) = S_A\left(\left\{ \frac{p_j}{d^*}\right\}\right) = S_A(\{p_j\}) + \log d^*.
\end{equation}
Thus, under these conditions requiring the maximization of $S_A(\{p_j\})$ is the same as requiring the maximization of $S_A(\{p_{js}\})$. This is indeed the case for the one-body and two-body observables considered in our numerical analysis (see Sec.~\ref{sec:numerical_results}). Hence, given the time-evolved probability distributions of measurement outcomes $p_j(t) := \Tr(\rho(t) A_j)$ obtained from the simulation for each observable, we compute the time evolution of the corresponding Shannon entropy
\begin{equation}
    S_A(t) := S_A(\{p_j\})(t) := - \sum_j p_j(t) \log p_j(t).
\end{equation} 
This is shown in Fig.~\ref{fig:shannon_all_m_all_obs} for all one-body and two-body observables. Each subplot represents a different observable, and the different colors correspond to different initial states as indicated in the legends. While $y$- and $z-$ observables show rapid equilibration to the absolute maximum for all initial states, the $x$-observables display a dependence on the initial state. We still see equilibration, albeit with larger fluctuations and reaching a constrained maximum that is lower than the absolute one. This is because, as $\theta$ increases, the expectation value of the charge $Q^x$ also increases and hence its chemical potential $\mu^x$ plays a bigger role in the equilibrium entropy $S_A(p_j)$, which takes the same form as Eq.~\ref{eqn:equilibrium_shannon_entropy}. The constraints lower the entropy from its absolute maximum $\log n_A$, and the strength of their effect depends on their corresponding Lagrange multiplier. Nevertheless, the entropy still equilibrates to the \textit{constrained} maximum, so these plots indeed support an entropic maximization approach.

\begin{figure}
    \centering
    \includegraphics[width=\linewidth]{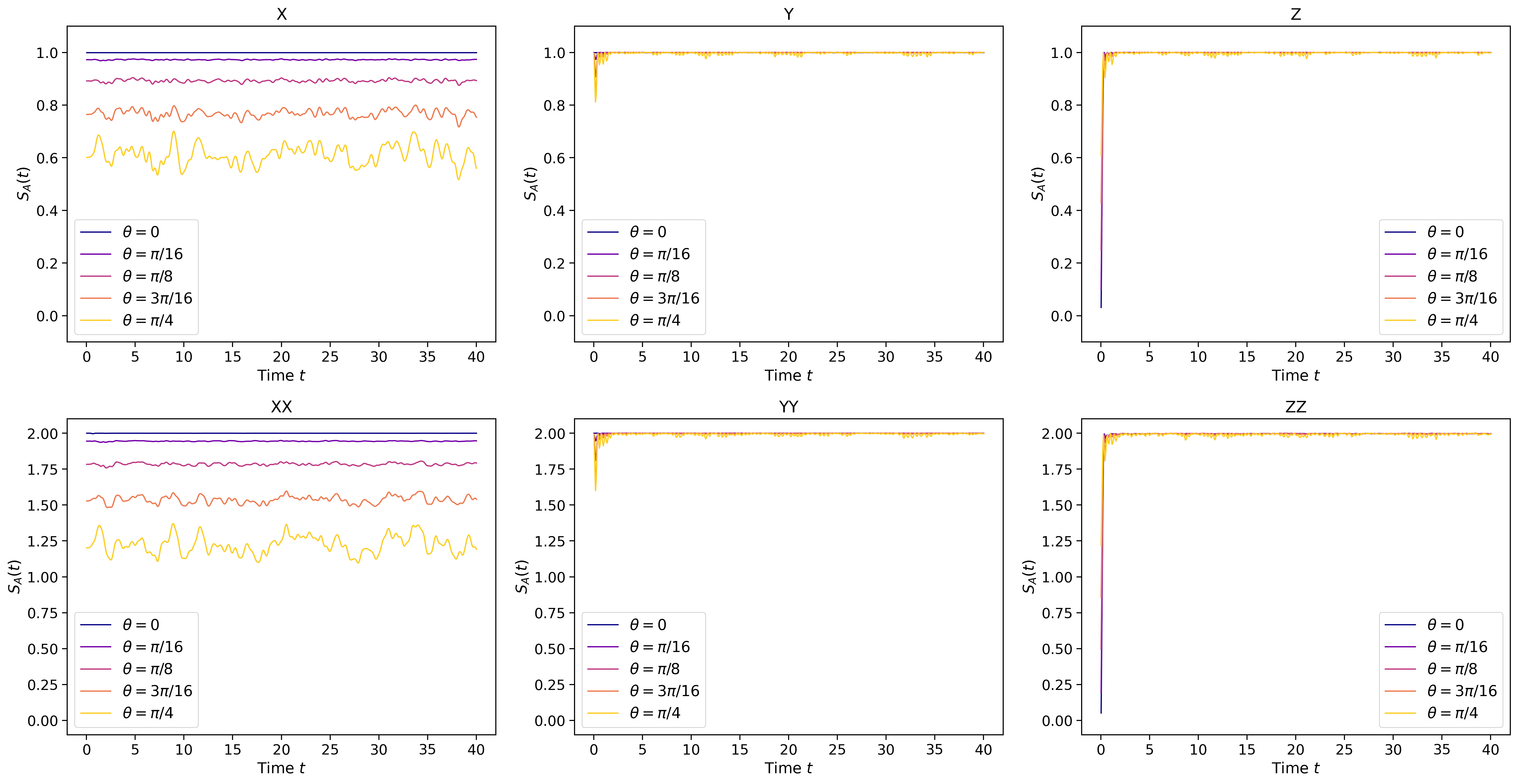}
    \caption{Time evolution of the Shannon entropy of the probability distributions $\{p_j(t)\}$ of all observables considered, for $N=16$. Each subplot corresponds to an observable, and the different colors indicate the initial state. Note for simplicity we call $X \coloneqq \sigma^x_\frac{N}{2}$, $XX \coloneqq \sigma^x_\frac{N}{2} \sigma^x_{\frac{N}{2}+1}$ etc.}
    \label{fig:shannon_all_m_all_obs}
\end{figure}

\subsubsection{Anomalous imaginary parts of weak values and the First Equilibrium Equation}
We then move on to providing more details regarding the anomalous imaginary parts of weak values observed for the model, initial states and observables considered. As noted in the main text, we witness this non-classical behavior at equilibrium only between $y$-observables and the charge $Q^z$, or viceversa between $z$-observables and the charge $Q^y$. This is linked to the distinct behavior of $x$-observables just noted above, given that $[\sigma^y, \sigma^z] \propto \sigma^x$. Moreover, for two-body observables we see this behavior only for eigenvalues $00$ and $11$, not for $01$ or $10$. 

Figure~\ref{fig:Im_WV_EE1_all_m_all_anomalous} shows all the anomalies found numerically for the setup presented in Section~\ref{sec:numerical_results} (or see Section~\ref{app_subsec:analytical_calculations_non-integrable_XXX_model} of this Appendix), for $N=16$. Every plot corresponds to a combination of observable, eigenvalue and charge for which we observed an anomaly. We plot the dimensionless quantity $\frac{\Im[Q^a_w \left(\rho_\infty, A_j \right)]}{||Q^a||_\textrm{op}}$, i.e., the imaginary part of the weak value of the charge $Q^a$ at equilibrium normalized by the operator norm of the same charge $||Q^a||_\textrm{op} = 1$, against the initial state parametrized by the angle $\theta$. In the inset, we show that the corresponding First Equilibrium Equation is still well respected despite this anomalous behavior, due to the fact that we have to consider also the size of the Lagrange multipliers and the other terms in the sum. To show it, we study the left-hand side of the First Equilibrium Equation, namely, 
\begin{equation} \label{eqn_app:EE1_LHS}
    \mathcal{E}_A(t, j) := \mathcal{E}_A(\rho(t), j) := \beta_A \Im[E_w (\rho(t), j)] + \sum_a \mu^a_A \Im[Q^a_w (\rho(t), j)].
\end{equation}
We look at the $\mathcal{E}_A(t, j)$ corresponding to the same observable's eigenspace considered in the main plot. We then plot both the absolute value of its temporal mean $|\overline{\mathcal{E}_A(t, j)}|$, and its fluctuations around equilibrium quantified by the mean absolute deviation (MAD) $\overline{|\mathcal{E}_A(t, j) - \overline{\mathcal{E}_A(t, j)}|}$. As both are small in all cases, we conclude these non-classical effects do not lead to significant enough violations of the First Equilibrium Equation to affect the framework. Nevertheless, it remains possible that for other models, initial states or observables this behavior might indeed lead to a larger violation. \\
We further investigate that indeed the First Equilibrium Equation is well respected in all cases considered by looking at the histograms of the values of the absolute value of the temporal mean $|\overline{\mathcal{E}_A(t, j)}|$ and of the MAD $\overline{|\mathcal{E}_A(t, j) - \overline{\mathcal{E}_A(t, j)}|}$ for all initial states and observables' eigenspaces, for $N=16$. As is shown in Fig.~\ref{fig:Confirming_EE1}, these histograms are sharply peaked around zero, hence confirming that there is no violation of this equilibrium equation. 

\begin{figure}
    \centering
    \includegraphics[width=\linewidth]{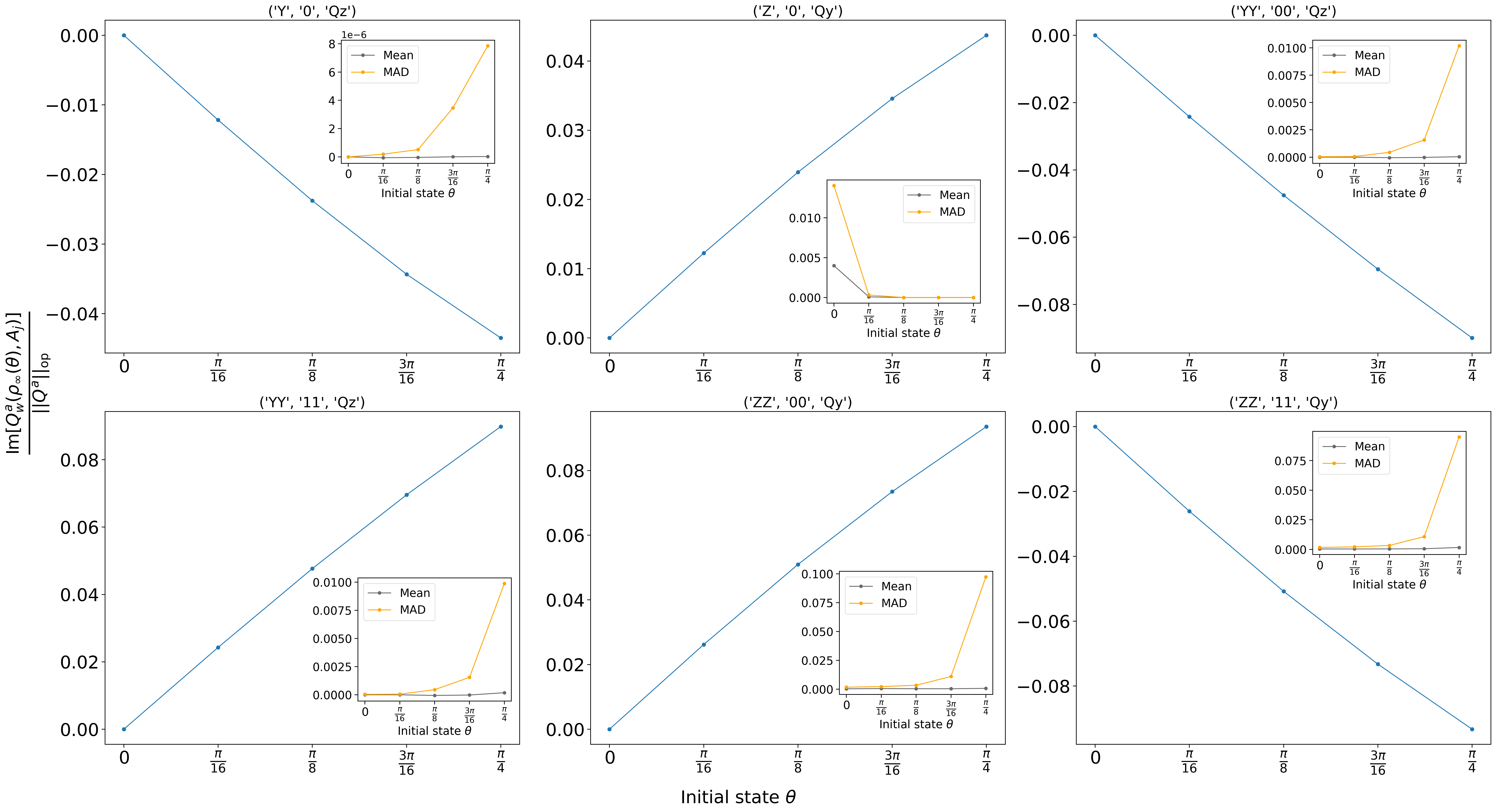}
    \caption{
    Here we show all anomalous imaginary parts of weak values found in the numerical simulations described in Section \ref{subsec:numerics_model_methods}. We observed anomalies between $y$-observables and $Q^z$, or viceversa between $z$-observables and $Q^y$. Moreover, for two-body observables we only see this behavior for eigenvalues $00$ and $11$. The titles of each subplot represent the observable $A$, eigenvalue $j$ and charge $Q^a$ considered in the corresponding subplot. Note for simplicity we call $Y \coloneqq \sigma^y_\frac{N}{2}$, $YY \coloneqq \sigma^y_\frac{N}{2} \sigma^y_{\frac{N}{2}+1}$ etc. Every main plot shows the imaginary part of the weak value of the charge $Q^a$ at equilibrium normalized by $||Q^a||_\textrm{op} = 1$, $\frac{\Im[Q^a_w \left(\rho_\infty(\theta), A_j \right)]}{||Q^a||_\textrm{op}}$, plotted against the initial state parametrized by the angle $\theta$, for $N=16$. This quantity is computed with respect to the eigenprojector $A_j$, which projects onto the eigensubspace corresponding to outcome $j$ of the observable $A$. Each inset concerns the first Equilibrium Equation corresponding to the same eigensubspace of the same observable. We plot the absolute value of its temporal mean, and mean absolute deviation (MAD) against the initial state to show that the corresponding Equilibrium Equation is still respected despite the anomaly. }
    \label{fig:Im_WV_EE1_all_m_all_anomalous}
\end{figure}

\subsubsection{The Second Equilibrium Equation}
Further, we also confirm that the Second Equilibrium Equation is respected. In Section~\ref{subsec:generalized_EEs_solution}, we determined that the most general solution of the Second Equilibrium Equation Eq.~\eqref{eqn:EE2} is given by Eq.~\eqref{eqn:general_solution_EEs}. Plugging it back into the equilibrium equation gave us the conditions $R_j = \varepsilon_j \, p_j$ and $\{R_j^a = q^a_j \, p_j\}_a$, which define $\varepsilon_j$ and $q^a_j$. Of course our state will never be truly stationary, so strictly speaking equilibrium is never reached. Nonetheless, for observables that equilibrate in the usual sense of ``on average", we expect the equilibrium equations to hold approximately after a sufficiently long time has passed. So we will have $R_j(t) \approx \varepsilon_j \, p_j(t)$ and $\{R_j^a(t) \approx q^a_j \, p_j(t)\}_a$. We now want to check that these equations are indeed respected. By the definition of $\varepsilon_j$ we have that $\overline{R_j(t)} = \varepsilon_j \overline{p_j(t)}$, and similarly for $R_j^a$, so we are left with looking at the time average of the fluctuations, quantified by the mean absolute deviation $\overline{|R_j(t)-\varepsilon_j p_j(t)|}$. In Fig.~\ref{fig:Linearity_of_pjRj}, we plot the histogram of these fluctuations for all initial states, observables and (independent) eigenvalues. We see that it is indeed sharply peaked at zero, meaning that the Second Equilibrium Equation is well respected. In the inset, we give an example of the evident linearity between $R_j(t)$ and $p_j(t)$.

\begin{figure}[ht]
    \centering
    \subfloat[Histogram of the absolute value of the temporal average of the left-hand side of the First Equilibrium Equation $\mathcal{E}_A(t, j)$, given in Eq.~\eqref{eqn_app:EE1_LHS}. The inset shows the histogram of the temporal average of the absolute value of the fluctuations of the First Equilibrium Equation, i.e., $\overline{|\mathcal{E}_A(t, j) - \overline{\mathcal{E}_A(t, j)}|}$. Note the significantly small range of values on the $x$-axis for both histograms. \label{fig:Confirming_EE1}]{\includegraphics[width=0.45\textwidth]{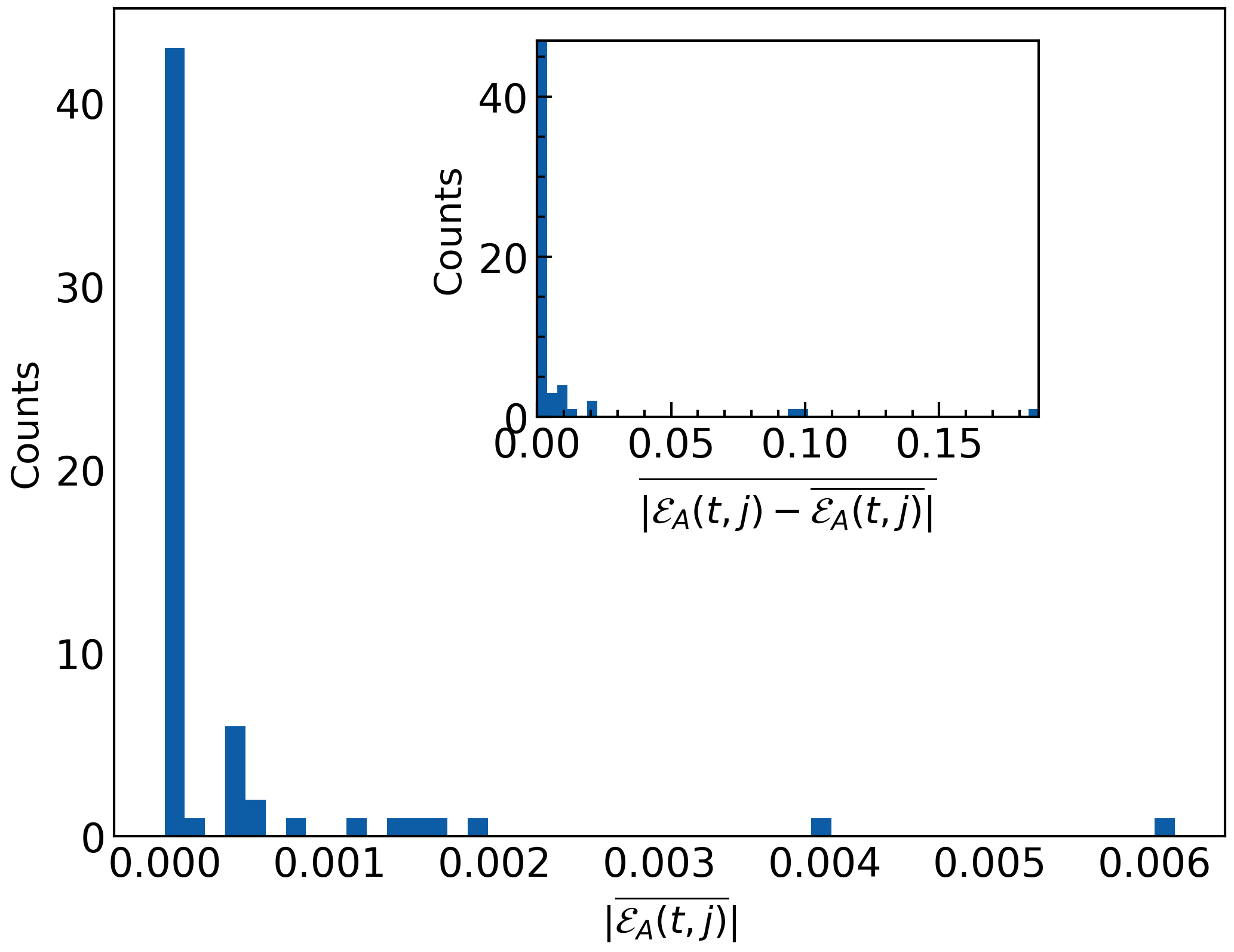}} \hspace{1cm}
    \subfloat[ Histogram of the time average of the absolute value of the fluctuations of the Second Equilibrium Equation. The inset is a time-implicit plot $(p_j(t), R_j(t))$ for the eigenvalue $j=00$ of the observable $\sigma^y_\frac{N}{2} \sigma^y_{\frac{N}{2}+1}$, initial state $\theta = \pi/4$ and $N=16$.\label{fig:Linearity_of_pjRj}] {\includegraphics[width=0.45\textwidth]{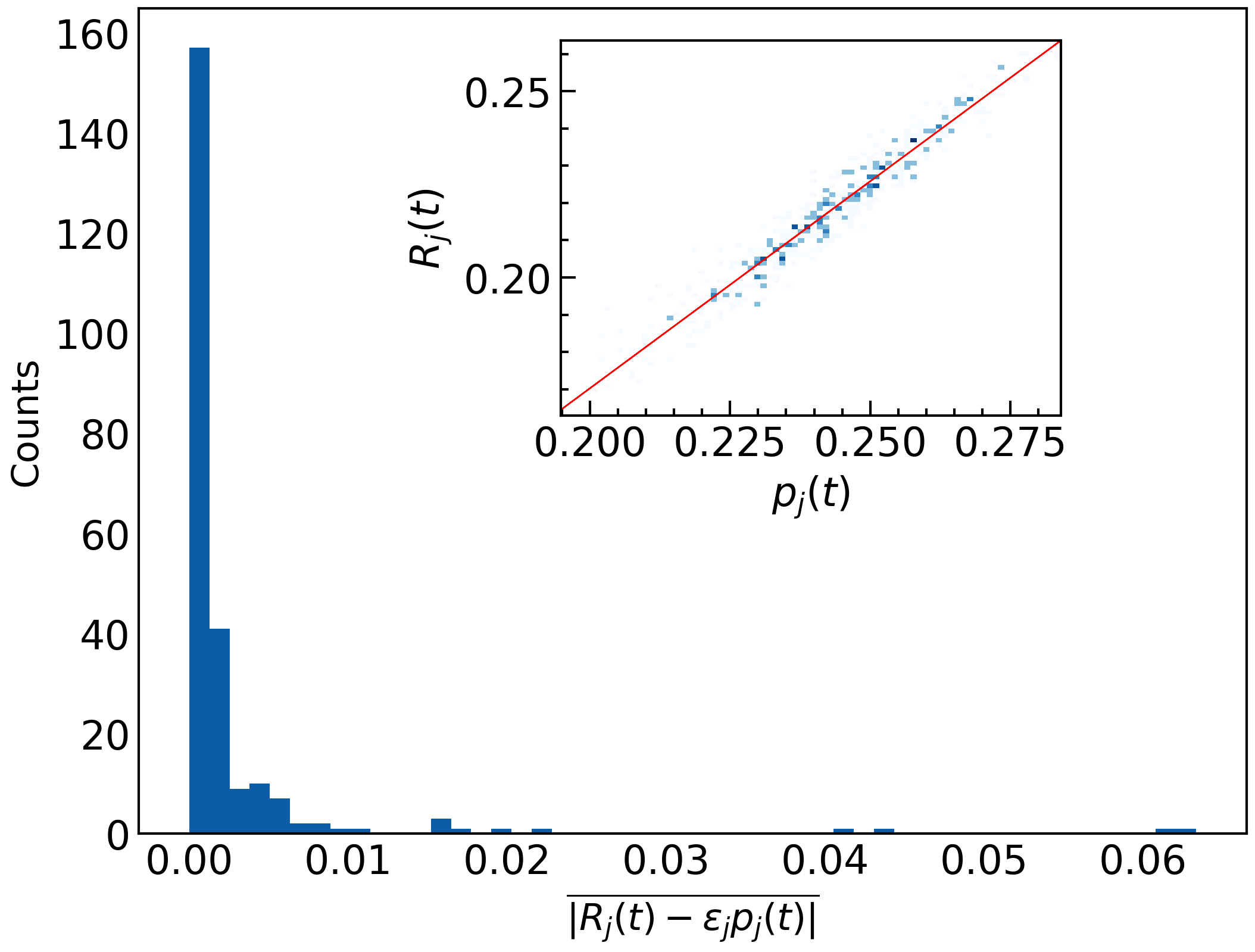}}
    \caption{We confirm the Equilibrium Equations are respected for all initial states, observables and (independent) eigenvalues considered, for $N=16$.} \label{fig:confirming_equilibrium_equations}
\end{figure}

\subsubsection{Estimating the equilibrium probability distribution}
Finally, we give further details about our estimates of the observables' probability distributions at equilibrium. 
In the main text, we studied how close our predictions $p_j^\textrm{est}$ are to the actual equilibrium value of the probability distribution, i.e., to the temporal average $\overline{p_j(t)}$. This was quantified using the total variation distance $D(\overline{p_j(t)}, p^\textrm{est})$. We showed that our predictions indeed work very well and they are always under $1\%$ error. We also compared our predictions to the estimates made using the non-Abelian thermal state (NATS). By plotting the relative improvement $\Delta_r := 1 - \frac{D(\overline{p_j(t)}, \, p_j^\textrm{est})}{D(\overline{p_j(t)}, \, p_j^\textrm{NATS})}$, we showed that for two-body observables our predictions are better than the ones made using the NATS. 

To complement the plots in the main text, in Fig.~\ref{fig:Histogram_CDF_TVD_pred_vs_DE_with_charges_and_NATS} we show histograms of the total variation distance between the temporal average and respectively our predictions (left subplot) and the predictions based on the NATS (right subplot). We also plot (in red) the cumulative distribution function in both cases. We highlight that for our predictions the error is always under $1\%$, while when using the NATS it is under $1\%$ in about $60\%$ of cases. 

In Fig.~\ref{fig:TVD_pred_vs_DE_absolute_improvement_nats_vs_with_charges_onetwo} we plot the absolute improvement (i.e., the difference) between the total variation distances computed in both cases $D(\overline{p_j(t)}, p_j^\textrm{NATS}) - D(\overline{p_j(t)}, p_j^\textrm{est})$ against the initial state parametrized by the angle $\theta$. Here the first total variation distance corresponds to the prediction based on the NATS, while the second one to our estimate. From Fig.~\ref{fig:TVD_pred_vs_DE_absolute_improvement_nats_vs_with_charges_one}, where this difference is plotted for one-body observables, it is clear that there is no advantage in using our estimate compared to the one using the NATS in this case. 
This is different for two-body observables, shown in Fig.~\ref{fig:TVD_pred_vs_DE_absolute_improvement_nats_vs_with_charges_two}, where the improvement from switching to our estimate can be significant. 
The plots shown so far have all concerned the $N=16$ spin case, so we also compare the framework's predictions across the system sizes considered, namely, $N=12, 14, 16$. In Fig.~\ref{fig:CDF_all_L_TVD_pred_vs_DE_with_charges_and_NATS_cdf=manual} we plot the cumulative distribution function of the total variation distance for all three system sizes, for both our estimates (left) and those based on the NATS (right). In both cases we see an improvement with system size, but our estimate does well regardless of the system size. For instance, note that for $N=12$ the total variation distance is under $0.01$ in around $90\%$ of cases. 

\begin{figure}
    \centering
    \includegraphics[width=\linewidth]{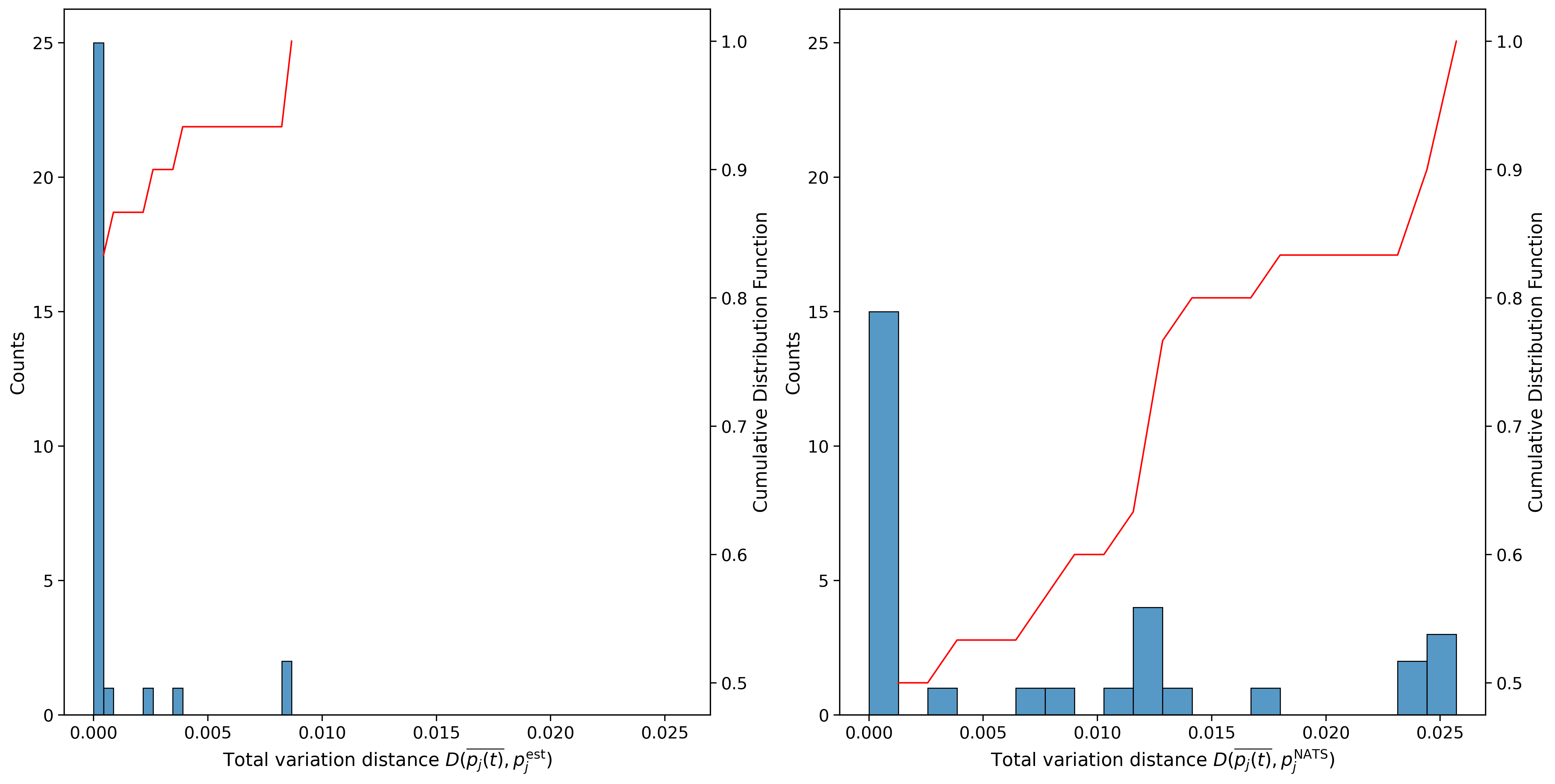}
    \caption{In blue: histogram of the total variation distance $D(\overline{p_j(t)}, p_j^\textrm{est})$ (left) and $D(\overline{p_j(t)}, p_j^\textrm{NATS})$ (right). Here $p_j^\textrm{est}$ is our prediction, $p_j^\textrm{NATS}$ is the prediction made using the NATS, and $\overline{p_j(t)}$ is the true value of the observable's probability distribution at equilibrium (the temporal average). Both plots include the data for all initial states and observables, and for $N=16$. In red: the cumulative distribution functions of the same quantities, respectively. Note the corresponding $y$-axis starts around $0.5$.}
    \label{fig:Histogram_CDF_TVD_pred_vs_DE_with_charges_and_NATS}
\end{figure}

\begin{figure}[ht]
    \centering
    \subfloat[One-body observables. \label{fig:TVD_pred_vs_DE_absolute_improvement_nats_vs_with_charges_one}]{\includegraphics[width=0.45\textwidth]{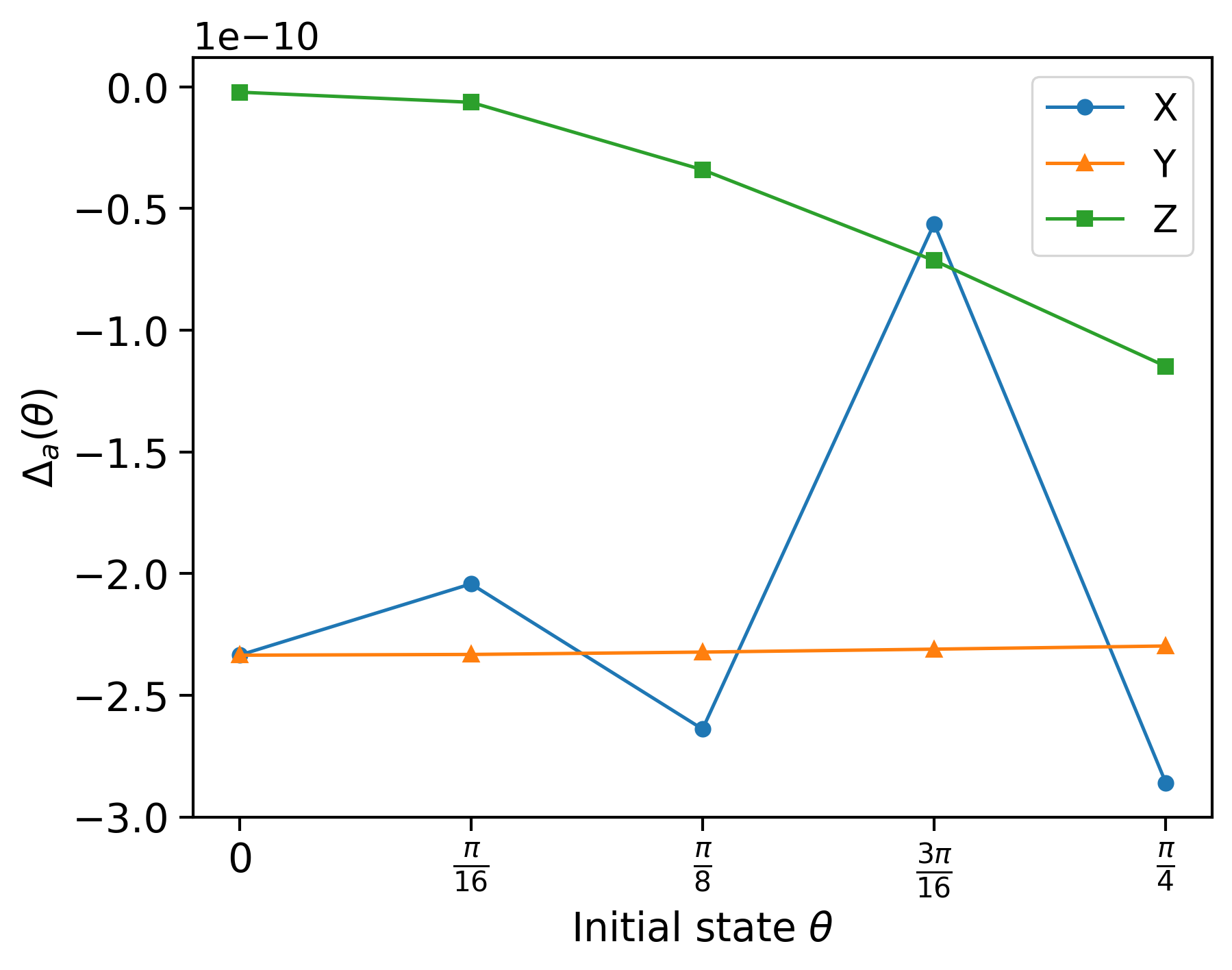}} \hspace{1cm}
    \subfloat[Two-body observables. \label{fig:TVD_pred_vs_DE_absolute_improvement_nats_vs_with_charges_two}] {\includegraphics[width=0.45\textwidth]{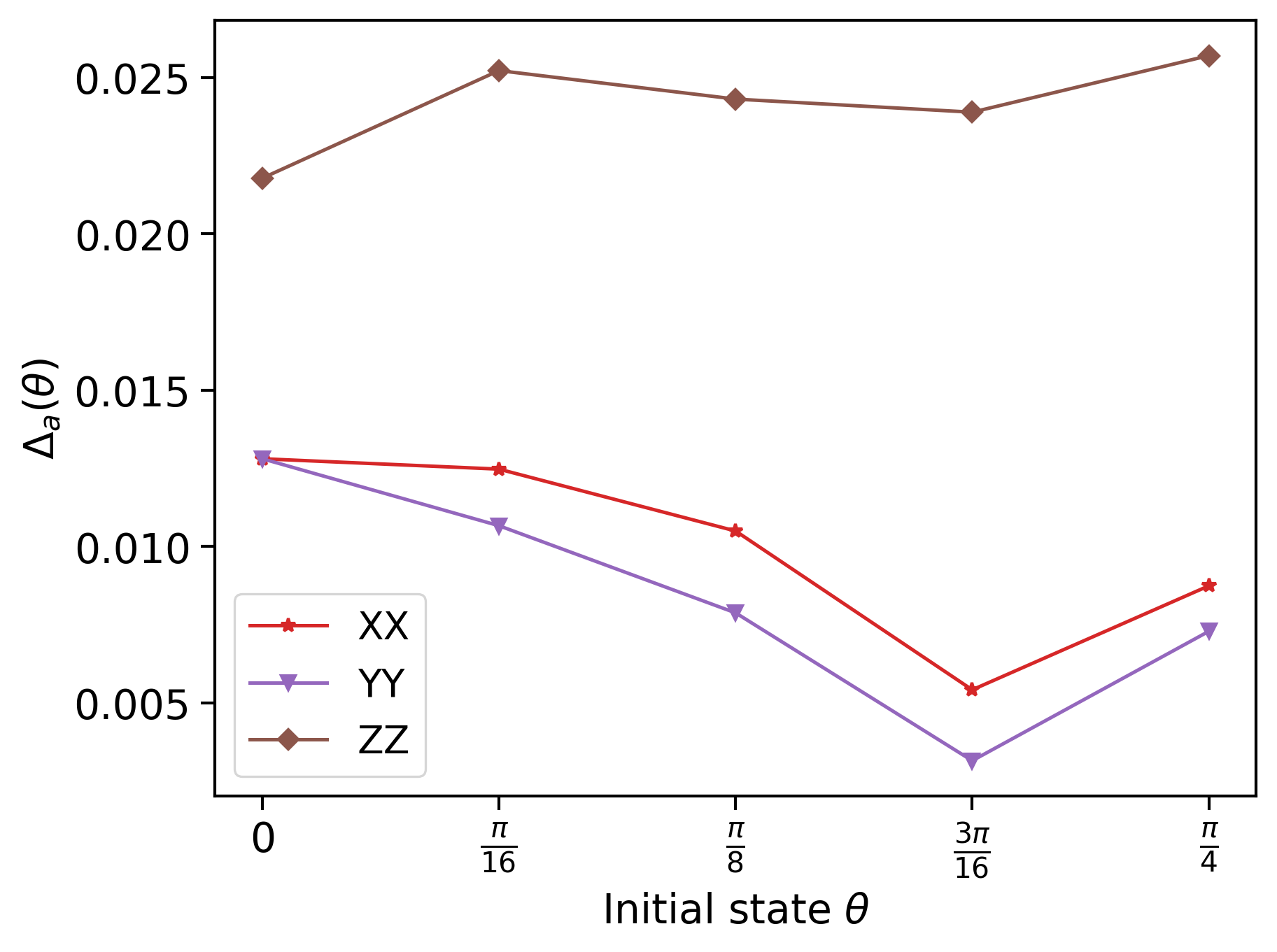}}
    \caption{The difference (absolute improvement) in total variation distance (TVD) $\Delta_a := D(\overline{p_j(t)}, p_j^\textrm{NATS}) - D(\overline{p_j(t)}, p_j^\textrm{est})$ plotted against the initial state parametrized by the angle $\theta$, for (a) one-body and (b) two-body observables and for $N=16$. This quantifies the improvement when switching from the prediction based on the NATS (with TVD $D(\overline{p_j(t)}, p_j^\textrm{NATS})$) to our estimate (with TVD $D(\overline{p_j(t)}, p_j^\textrm{est})$). Note for simplicity we call $X \coloneqq \sigma^x_\frac{N}{2}$, $XX \coloneqq \sigma^x_\frac{N}{2} \sigma^x_{\frac{N}{2}+1}$ etc. We also highlight that the range of the y-axis in subplot (a) is extremely small.} \label{fig:TVD_pred_vs_DE_absolute_improvement_nats_vs_with_charges_onetwo}
\end{figure}

\begin{figure}
    \centering
    \includegraphics[width=\linewidth]{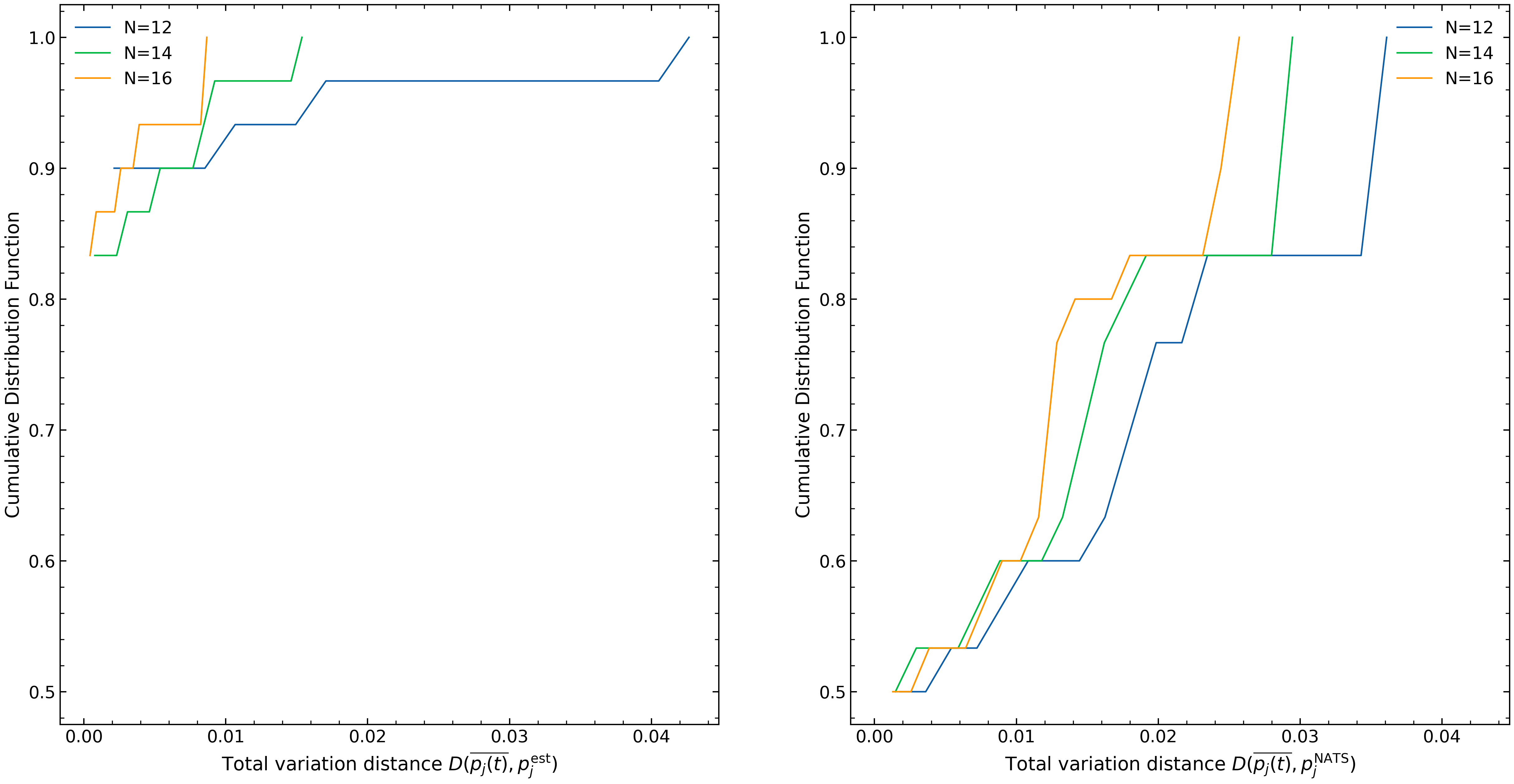}
    \caption{The cumulative distribution function of the total variation distance $D(\overline{p_j(t)}, p_j^\textrm{est})$ (left) and $D(\overline{p_j(t)}, p_j^\textrm{NATS})$ (right). Here $p_j^\textrm{est}$ is our prediction, $p_j^\textrm{NATS}$ is the prediction made using the NATS, and $\overline{p_j(t)}$ is the true value of the observable's probability distribution at equilibrium (the temporal average). Both plots include the data for all initial states and observables studied. Each colored curve corresponds to one of the system sizes considered, namely, $N=12, 14, 16$. Note the reduced range of the $y$-axis, starting around $0.5$.}
    \label{fig:CDF_all_L_TVD_pred_vs_DE_with_charges_and_NATS_cdf=manual}
\end{figure}

\FloatBarrier
\subsection{Large parameter space exploration in search of anomalous weak values} \label{app_subsec:reWV_param_space_exploration}
In Section \ref{subsec:numerics_anomalous_weak_values} of the main text, we discussed finding anomalies in the imaginary parts of the charges' weak values, but not observing any anomalies in the real part. We remind the reader that a ``standard" expectation value $\Tr(\rho O)$ of an observable $O$ on a state $\rho$ is real and cannot lie outside the spectrum of $O$. Instead, a weak value is a kind of generalized expectation value that can be anomalous, in the sense that it can violate these two conditions: it can have non-zero imaginary part and can take values outside of the range of the observable's spectrum \cite{arvidsson-shukur_properties_2024}. We thus conducted a large parameter space exploration to search for other anomalies. While we were able to find several anomalies in the imaginary parts, we never witnessed any violation in the real parts. It is unclear to us why this is the case, and it would be interesting to explore other models with different non-Abelian symmetries, as well as other parameter regimes, to see whether we can indeed find anomalies also in the real parts of the weak values. Here we present the details of this analysis and our results. \\

Consider the family of SU(2)-symmetric Hamiltonians given by 
\begin{equation}
    H = \sum_{a=x,y,z} \left\{ \sum_{i=1}^{N-1} J_i \sigma^a_i \sigma^a_{i+1} + \sum_{i=1}^{N-2} K_i \sigma^a_i \sigma^a_{i+2} \right\},
\end{equation}
where $i$ is an index that runs over the $N=6$ lattice sites, and $\sigma^a_i$ ($a=x,y,z$) represents the Pauli operator with Pauli matrix $\sigma^a$ acting on the lattice site $i$, i.e., $\sigma^a_i \coloneqq \mathbb{I}^{\otimes (i-1)} \otimes \sigma^a \otimes \mathbb{I}^{\otimes(N-i)}$. Note we are using open boundary conditions. The interaction coefficients are $J_i = 1 + \epsilon \, \delta_{i,3}$ and $K_i = K + \epsilon \, \delta_{i,3}$, where for our analysis we consider $K \in \{0, 0.1, 1\}$ and $\epsilon \in \{0, 0.3\}$.\\
Due to the SU(2)-symmetry of the model, we study the following three non-commuting charges:
\begin{equation}
    Q^a := \frac{1}{N} \sum_{i=1}^N \sigma^a_i \quad \forall a \in \{x,y,z\}.
\end{equation}
We consider the most general family of initial states possible by rotating the antiferromagnetic state along $z$, $\ket{01\ldots01}$, via a unitary parametrized by three Euler angles $\mathcal{R}(\alpha, \beta, \gamma)$ with $\alpha, \gamma \in [0, 2\pi)$ and $\beta \in [0, \pi]$. We consider two kinds of rotations: an SU(2)-symmetric one , which maps the starting state to another state with the same average energy, and an SU(2)-breaking one, which can map it to different average energy subspaces. These are respectively given by
\begin{gather} 
 \mathcal{R}^\textrm{symm}(\alpha, \beta, \gamma) := \bigotimes_{i=1}^{N} R_x(\alpha) R_y(\beta) R_x(\gamma) \\
 \mathcal{R}^\textrm{break}(\alpha, \beta, \gamma) := \bigotimes_{i=1}^{N/2} \left[ R_x(\alpha) R_y(\beta) R_x(\gamma) \otimes R_x(-\gamma) R_y(-\beta) R_x(-\alpha)\right]
\end{gather}
where $R_x(\theta) = \exp(-iX\theta/2)$ and $R_y(\theta) = \exp(-iY\theta/2)$ are respectively the rotation operator along the $x$ and $y$ axis. We consider four values of $\alpha$ and $\gamma$, namely, $\alpha, \gamma \in \left\{ 0, \frac{\pi}{2}, \pi, \frac{3\pi}{2} \right\}$ and five values for the remaining angle, $\beta \in \left\{ 0, \frac{\pi}{4}, \frac{\pi}{2}, \frac{3\pi}{4}, \pi \right\}$. For each combination of these angles, we have an SU(2)-symmetric and an SU(2)-breaking initial state:
\begin{gather}
    |\psi_0^\textrm{symm}(\alpha, \beta, \gamma) \rangle = \mathcal{R}^\textrm{symm}(\alpha, \beta, \gamma) \ket{01\ldots01}, \\
    |\psi_0^\textrm{break}(\alpha, \beta, \gamma) \rangle = \mathcal{R}^\textrm{break}(\alpha, \beta, \gamma) \ket{01\ldots01}.
\end{gather}

As we are interested in the equilibrium behavior of the weak values, we construct the exact stationary state $\rho_\infty = \sum_n \Pi_n |\psi_0(\alpha, \beta, \gamma) \rangle\langle \psi_0(\alpha, \beta, \gamma)| \Pi_n$. To compute it, we first diagonalize the Hamiltonian. Then, we transform $\rho_0(\alpha, \beta, \gamma) := |\psi_0(\alpha, \beta, \gamma)\rangle \langle \psi_0(\alpha, \beta, \gamma)|$ into the energy basis and set to zero all off-diagonal elements connecting different energy eigensubspaces. We consider all eigenvalues $E_n, E_m$ such that $|E_n - E_m| \leq \tau$ to be equal, where the tolerance $\tau$ is chosen as follows. First, we estimate the diagonalization error, quantified by the Frobenius-norm residual $\delta_\textrm{diag} \;=\;\bigl\|\,U^\dagger H\,U - \mathrm{diag}(\{E_n\})\bigr\|_F$, where $U$ is the unitary diagonalizing $H$ constructed with the eigenvectors, and $\mathrm{diag}(\{E_n\})$ is a diagonal matrix whose entries are the energy eigenvalues. Then, we sort the energy eigenvalues in increasing order and compute all differences $\Delta_n := E_{n+1} - E_n$. We select all such non-zero differences, take their logarithm $\log_{10} \Delta_n$ and again sort them in increasing order. Next, we identify the largest ``jump" in these log-gaps. This separates the gaps due to noise from the larger physical level spacings. We then consider the log-gap just below the jump and return to linear space: $\Delta_n^\textrm{jump}$. We multiply this gap by a factor of 10 to avoid any floating-point errors. Moreover, to make sure this is greater than the diagonalizing error ($\sim 10^{-13}$), we take the maximum value between the two, i.e., $\tau = \max(\delta_\textrm{diag}, 10 \Delta_n^\textrm{jump})$. Finally, we rotate back the block-decohered initial state into its original basis. $\rho_\infty$ is now stable by construction, in the sense that changing the tolerance $\tau$ does not affect the grouping of energy eigenstates because of the large difference between gaps due to noise and those due to actual physical separation. We tested this and confirmed the state does not change when increasing the tolerance by several orders of magnitude (we increased it by factors of 10 to $10^5$). \\

Similarly to what we did for the initial states, we also consider a general observable parametrized by three Euler angles. We take as a starting observable the Pauli operator $\sigma^z_i\coloneqq \mathbb{I}^{\otimes (i-1)} \otimes \sigma^z \otimes \mathbb{I}^{\otimes(N-i)}$ acting non-trivially on the lattice site $i=\frac{N}{2}+1$. We then apply a local rotation $\mathcal{R}^\textrm{loc}(\alpha, \beta, \gamma) := R_x(\alpha) R_y(\beta) R_x(\gamma)$ to obtain an observable 
\begin{equation}
    O(\alpha, \beta, \gamma) \coloneqq \mathbb{I}^{\otimes (i-1)} \otimes \mathcal{R}^\textrm{loc}(\alpha, \beta, \gamma) \, \sigma^z \, \left( \mathcal{R}^\textrm{loc}(\alpha, \beta, \gamma) \right)^\dagger \otimes \mathbb{I}^{\otimes(N-i)}.
\end{equation}
We then construct its two coarse-grained eigenprojectors given by $A_j(\alpha, \beta, \gamma) = \frac{\mathbb{I} + (-1)^{j} O(\alpha, \beta, \gamma)}{2}$ with $j \in \{0,1\}$. We consider the same values of the Euler angles as for the initial states. \\

To summarize, we are considering 3 values of $K$, 2 values of $\epsilon$, $4 \times 5 \times 4 = 80$ combinations of Euler angles for both the initial state and the observable (separately), and we have 2 eigenprojectors and 3 charges. Hence, for each type of initial state rotation (SU(2) symmetric or SU(2) breaking), we are calculating $3 \times 2 \times 80 \times 80 \times 2 \times 3 = 230,400$ weak values; or $460,800$ in total. \\
For each of these cases, we compute the weak value
\begin{equation}
    Q^a_w(\rho_\infty, A_j) = \frac{\Tr\left( A_j Q^a \rho_\infty\right)}{\Tr(A_j \rho_\infty)}.
\end{equation}
(Although note we exclude cases in which $p_j^\infty := \Tr(A_j \rho_\infty)$ is smaller than the diagonalization error, which is of the order of $10^{-13}$ for $N=6$, because they are essentially indistinguishable from zero. This excludes 1,440 probabilities, or $1,440 \times 3= 4,320$ cases out of 460,800, given we have three charges. This happens only for SU(2)-breaking rotations of the initial state.) We then study the real and imaginary parts of the weak value. If the absolute value of the imaginary part is greater than $10^{-2}$, we consider it an anomaly. Similarly, we check whether the absolute value of the real part is greater than $1 + 10^{-2}$ since the charges' spectra range from $-1$ to $1$. Note that we tried other values for the threshold, from $10^{-13}$ to $10^{-2}$, and the results were unchanged.\\
We observe 47,208 anomalies in the imaginary part of the weak value for SU(2)-breaking rotations of the initial state. That is, roughly $20\%$ of $230,400$ cases are anomalous. However, we do not observe any anomaly for SU(2)-symmetric rotations of the initial state. This is because on the antiferromagnetic state along $z$, $|01...01\rangle$ all charges have mean zero, as can also be seen in the analytical calculations: Eqs.~\eqref{eqn:charge_avg_std_x_app}, \eqref{eqn:charge_avg_std_y_app} and \eqref{eqn:charge_avg_std_z_app}. An SU(2)-symmetric rotation can change the direction of the average charge vector but not its magnitude, hence we cannot have anomalies in the whole average energy subspace of $|01...01\rangle$. Thus, in total, the anomalies constitute approximately $10\%$ of $460,800$ cases. We also could not find any anomaly in the real parts of the weak values, neither for SU(2)-breaking nor SU(2)-symmetric rotations of the initial state. It is currently unclear to us why this is the case, but we intend to investigate this further to obtain a more complete understanding of this emergence of non-classicality at equilibrium.

\subsection{Computing the non-Abelian thermal state (NATS)}
\label{app_subsec:comparing_predictions_to_nats}

\subsubsection{Analytical construction of the Gibbs ensemble and the NATS}

The Hamiltonian we have been studying is 
\begin{equation}
    H = \sum_{a=x,y,z} \left\{ \sum_{i=1}^{N-1} J_i \sigma^a_i \sigma^a_{i+1} + \sum_{i=1}^{N-2} J_i \sigma^a_i \sigma^a_{i+2} \right\},
\end{equation}
with open boundary conditions and $J_i = J + \epsilon \, \delta_{i,3}$, where $J=1$ and $\epsilon=0.3$. Due to the SU(2)-symmetry of the model, we have the following non-commuting charges:
\begin{equation}
    Q^a := \sum_{i=1}^N \sigma^a_i \quad \forall a \in \{x,y,z\}.
\end{equation} 
(Note here we are not dividing the charges by the system size $N$.)
We now want to construct the local versions of the Gibbs ensemble and the non-Abelian thermal state (NATS) for one-site and two-site subsystems, given the Hamiltonian above. These two ensembles are respectively given by 
\begin{equation}
    \rho_G := \frac{e^{-\beta H_S}}{Z_G}; \quad \qquad \rho_\textrm{NATS} := \frac{e^{-\beta H_S - \sum_a \mu^a Q^a_S}}{Z_N},
\end{equation}
where the partition functions are $Z_G := \Tr\left( e^{-\beta H_S}\right)$ and $Z_N := \Tr\left( e^{-\beta H_S - \sum_a \mu^a Q^a_S} \right)$. The inverse temperature $\beta$ and the chemical potentials $\mu$ need to be fixed by imposing the constraints of average energy and average charge conservation on the global thermal state of the system. Note we always choose subsystems with support on $i=\frac{N}{2}$ ($N$ even) and in the two-site case we also have support on $i+1=\frac{N}{2}+1$. \\

\paragraph*{One-site subsystem.---}We begin by identifying the subsystem Hamiltonian and the local charges. For one-site subsystems, we have $H_S^{(1)} = 0$ and $Q_{(1)}^a = \sigma^a$ with $a=x,y,z$, where $\sigma^a$ is a Pauli matrix. Thus, the Gibbs ensemble and the NATS take the form
\begin{equation}
    \rho_G^{(1)} = \frac{1}{2} \mathbb{I}_2; \quad \qquad \rho_\textrm{NATS}^{(1)} = \frac{e^{-\sum_a \mu^a \sigma^a}}{Z_N}.
\end{equation}
We can simplify the expression for the $\rho_\textrm{NATS}$ in the following way. Define the vector of chemical potentials $\Vec{\mu} := (\mu^x, \mu^y, \mu^z)^T$ and its magnitude $\mu \coloneqq \|\Vec{\mu}\| := \sqrt{(\mu^x)^2 + (\mu^y)^2 + (\mu^z)^2}$. The corresponding unit vector is $\hat{\Vec{n}} := \frac{\Vec{\mu}}{\mu}$. We can then rewrite $\sum_a \mu^a \sigma^a = \mu (\hat{\Vec{n}} \cdot \Vec{\sigma})$ and we note that $(\hat{\Vec{n}} \cdot \Vec{\sigma})^2 = \sum_a n_a^2 \, \mathbb{I}_2 = \mathbb{I}_2$ since $\hat{\Vec{n}}$ is a unit vector. For any operator $A$ such that $A^2 = \mathbb{I}$, we have that $e^{tA} = \cosh(t) \mathbb{I} + \sinh(t) A$, so $e^{-\mu (\hat{\Vec{n}} \cdot \Vec{\sigma})} = \cosh(\mu) \mathbb{I} - \sinh(\mu) (\hat{\Vec{n}} \cdot \Vec{\sigma})$. From this expression, we compute the partition function: $\Tr\left( e^{-\mu (\hat{\Vec{n}} \cdot \Vec{\sigma})} \right) = 2\cosh(\mu)$, so $\rho_\textrm{NATS} = \frac{1}{2} \left[ \mathbb{I}_2 - \tanh(\mu) (\hat{\Vec{n}} \cdot \Vec{\sigma}) \right]$. 
For one-site subsystems, our observables of interest are precisely the Pauli matrices $\sigma^a$, which have eigenprojectors $A_j^a := \frac{\mathbb{I}_2 + (-1)^j \sigma^a}{2}$. Defining $\Vec{A}_j := (A_j^x, A_j^y, A_j^z)^T$, we can compute the probabilities
\begin{align} \label{eqn:prob_gibbs_one-site}
    \vec{p}_j^G &= \Tr(\rho_G^{(1)} \Vec{A}_j) = \left(\frac{1}{2}, \frac{1}{2}, \frac{1}{2}\right)^T,
\\
\label{eqn:prob_nats_one-site}
    \vec{p}_j^{NATS} &= \Tr(\rho_\textrm{NATS}^{(1)} \Vec{A}_j) = \frac{1}{2} - \frac{(-1)^j}{2} \tanh(\mu) \, \hat{\vec{n}}.
\end{align}

\paragraph*{Two-site subsystem.---}For the two-site case, we have $H_S^{(2)} = J \sum_a \sigma^a \otimes \sigma^a = J \Vec{\sigma}_1 \cdot \Vec{\sigma}_2$ and $Q_{(2)}^a = \sigma^a \otimes \mathbb{I}_2 + \mathbb{I}_2 \otimes \sigma^a = \sigma_1^a + \sigma_2^a$, where $\sigma_k^a := \mathbb{I}_2^{\otimes k-1} \otimes \sigma^a \otimes \mathbb{I}_2^{\otimes 2-k}$ for $k=1,2$. 
To construct the Gibbs state, note $H_S$ is SU(2)-symmetric. In fact, it is a linear function of $S^2$: $S^2 = (\Vec{S}_1 + \Vec{S}_2)^2 = \frac{3}{2} \mathbb{I}_4 + \frac{H_S}{2J}$, where $\Vec{S}_k := \frac{1}{2} \Vec{\sigma}_k$ for $k=1,2$. Thus, $H_S = J \left( 2S^2 - 3 \mathbb{I}_4 \right)$. This means the Gibbs state must be a Werner state: $\rho_G = p \Pi_\textrm{asym} + \frac{(1-p)}{3} \Pi_\textrm{sym}$, where $\Pi_\textrm{asym}$ and $\Pi_\textrm{sym}$ project onto the anti-symmetric and symmetric subspaces respectively. In this case, the former is the singlet subspace spanned by $|\psi_-\rangle := \frac{|01\rangle - |10\rangle}{\sqrt{2}}$, i.e., $\Pi_\textrm{asym} := |\psi_-\rangle\langle \psi_-|$, while the latter is the triplet subspace spanned by $\{ |00\rangle, |\psi_+\rangle, |11\rangle \}$ with $|\psi_+\rangle := \frac{|01\rangle + |10\rangle}{\sqrt{2}}$. Indeed, $S^2 |\psi_-\rangle\langle \psi_-| = 0$ and $S^2 \Pi_\textrm{sym} = 2 \Pi_\textrm{sym}$, where we denoted the projector onto the triplet subspace as $\Pi_\textrm{sym} := |00\rangle\langle 00| + |\psi_+\rangle\langle \psi_+| + |11\rangle\langle 11|$. So $H_S = J (4 \Pi_\textrm{sym} - 3 \mathbb{I}_4) = - 3J\, \Pi_\textrm{asym} + J \, \Pi_\textrm{sym}$. Thus, exponentiating and imposing state normalization we get
\begin{equation} \label{eqn:rho_gibbs_two-site_app}
    \rho_G^{(2)} = \frac{e^{3\beta J}}{e^{3\beta J} + 3 e^{-\beta J}} \Pi_\textrm{asym} + \frac{e^{-\beta J}}{e^{3\beta J} + 3 e^{-\beta J}} \Pi_\textrm{sym},
\end{equation}
which is indeed of the Werner form with $p=\frac{e^{3\beta J}}{e^{3\beta J} + 3 e^{-\beta J}}$. The partition function is $Z_G^{(2)} = e^{3\beta J} + 3 e^{-\beta J}$. From $\mathbb{I}_4 = \Pi_\textrm{asym} + \Pi_\textrm{sym}$ and $H_S = - 3J\, \Pi_\textrm{asym} + J \, \Pi_\textrm{sym}$, we get $\Pi_\textrm{asym} = \frac{1}{4} \left[ \mathbb{I}_4 - \frac{H_S}{J} \right]$ and $\Pi_\textrm{sym} = \frac{1}{4} \left[ 3\mathbb{I}_4 + \frac{H_S}{J} \right]$. Substituting these expressions into Eq.~\ref{eqn:rho_gibbs_two-site_app}, we can rewrite it as 
\begin{equation}
    \rho_G^{(2)} = \frac{1}{4} \left[ \mathbb{I}_4 + \frac{1 - e^{4\beta J}}{3 + e^{4\beta J}} \frac{H_S}{J}\right].
\end{equation}
For the two-site subsystem, the observables we consider are $\sigma^a \otimes \sigma^a$. These have eigenprojectors
\begin{equation} \label{eqn:obs_2-site_eigenproj}
    A^a_{jk} = \left(\frac{\mathbb{I}_2 + (-1)^j \sigma^a}{2} \right) \otimes \left(\frac{\mathbb{I}_2 + (-1)^k \sigma^a}{2} \right) = \frac{1}{4} \left\{ \mathbb{I}_4 + (-1)^j [\sigma^a \otimes \mathbb{I}_2 + (-1)^{j+k} \, \mathbb{I}_2 \otimes \sigma^a ]  + (-1)^{j+k} \, \sigma^a \otimes \sigma^a \right\}.
\end{equation}
We now compute the corresponding probabilities $p_{jk}^a := \Tr(A^a_{jk} \rho_G)$ for the Gibbs state. Given $\Tr\left(\rho_G (\sigma^a \otimes \mathbb{I}_2)\right) = \Tr\left(\rho_G (\mathbb{I}_2 \otimes \sigma^a)\right) = 0$ and $\Tr\left(\rho_G (\sigma^a \otimes \sigma^a)\right) = \frac{1 - e^{4\beta J}}{3 + e^{4\beta J}}$, we find
\begin{equation} \label{eqn:prob_gibbs_two-site}
    \vec{p}_{jk}^G = \Tr\left(\rho_G \vec{A}_{jk}\right) = \frac{1}{4} \left[ 1 + (-1)^{j+k} \frac{1 - e^{4\beta J}}{3 + e^{4\beta J}} \right] \left( 1, 1, 1 \right)^T,
\end{equation}
where $\vec{A}_{jk} = (A^x_{jk}, A^y_{jk}, A^z_{jk})^T$. \\

We now proceed to construct the two-site NATS. In this case, we have the effective Hamiltonian $H_\textrm{eff} := \beta H_S + \Vec{\mu} \cdot \Vec{Q}_{(2)} = \beta J \left( 2S^2 - 3 \mathbb{I}_4 \right) + 2 \mu S_n = -3\beta J \Pi_\textrm{asym} + \beta J \Pi_\textrm{sym} + 2 \mu S_n$, where $S_n := \hat{\Vec{n}} \cdot \Vec{S}$, $\Vec{S} := \Vec{S}_1 + \Vec{S}_2$ is the total spin vector and $\Vec{Q}_{(2)} = \left(Q^x_{(2)}, Q^y_{(2)}, Q^z_{(2)} \right)^T$. $H_\textrm{eff}$ has eigenstates $|s, m\rangle$, where $s \in \{0, 1\}$ is the total spin and $m \in \{ -s, \ldots, s\} = \{-1, 0, 1\}$ are the eigenvalues of $S_n := (\hat{\Vec{n}} \cdot \Vec{S})$. So we have $S^2 |s, m\rangle = s(s+1) |s, m\rangle$ and $S_n |s, m\rangle = m |s, m\rangle$. From this it follows that the spectral decomposition of $S_n$ is $S_n = |1, 1 \rangle\langle 1, 1| - |1, -1 \rangle\langle 1, -1|$, and hence $S_n^3 = S_n$. Taylor expanding $e^{-2\mu S_n}$, reducing all powers $S_n^k$ for $k\geq 3$ and collecting terms, we get $e^{-2\mu S_n} = \mathbb{I}_4 - \sinh(2\mu) S_n + (\cosh(2\mu) - 1) S_n^2$. Thus, 
\begin{align}
    e^{-H_\textrm{eff}} &= e^{-\beta H_S} e^{-2\mu S_n} = \left( e^{3\beta J} \Pi_\textrm{asym} + e^{-\beta J} \Pi_\textrm{sym} \right) \left[ \mathbb{I}_4 - \sinh(2\mu) S_n + (\cosh(2\mu) - 1) S_n^2 \right] = \\
    &= e^{3\beta J} \Pi_\textrm{asym} + e^{-\beta J} \left[ \Pi_\textrm{sym} - \sinh(2\mu) S_n + (\cosh(2\mu) - 1) S_n^2 \right].
\end{align}
The partition function is given by $Z_N^{(2)} = \Tr(e^{- H_\textrm{eff}}) = e^{3\beta J} + e^{-\beta J} \left[ 1 + 2 \cosh(2\mu) \right]$, where we used $\Tr(S_n^2) = 2$. Simplifying the factor of $e^{-\beta J}$, we get
\begin{equation}
    \rho_\textrm{NATS}^{(2)} = \frac{e^{4\beta J}}{\tilde{Z}_N^{(2)}} \Pi_\textrm{asym} + \frac{1}{\tilde{Z}_N^{(2)}} \left[ \Pi_\textrm{sym} - \sinh(2\mu) S_n + (\cosh(2\mu) - 1) S_n^2 \right],
\end{equation}
with $\tilde{Z}_N^{(2)} := e^{4\beta J} + 1 + 2 \cosh(2\mu)$. Note that for $\mu = 0$ we retrieve the Gibbs ensemble for the two-site subsystem.
To compute the probabilities $p_{jk}^a := \Tr(\rho_\textrm{NATS}^{(2)} A_{jk}^a)$, consider again the eigenprojectors in Eq.~\eqref{eqn:obs_2-site_eigenproj}. The only non-zero overlaps are $\Tr\left( S_n [\sigma^a \otimes \mathbb{I}_2 + (-1)^{j+k} \, \mathbb{I}_2 \otimes \sigma^a ] \right) = 4 n^a \delta_{jk}$, $\Tr(\Pi_\textrm{asym} (\sigma^a \otimes \sigma^a)) = -1$, $\Tr(\Pi_\textrm{sym} (\sigma^a \otimes \sigma^a)) = 1$ and $\Tr(S_n^2 (\sigma^a \otimes \sigma^a)) = 2 (n^a)^2$. Therefore, 
\begin{equation} \label{eqn:prob_nats_two-site}
    p_{jk}^{a, NATS} = \frac{1}{4} + \frac{(-1)^{j+k}}{4 \tilde{Z}_N^{(2)}} (1 - e^{4\beta J}) + \frac{(-1)^{j+1}}{\tilde{Z}_N^{(2)}} n^a \delta_{jk} \sinh(2\mu) + \frac{(-1)^{j+k}}{2 \tilde{Z}_N^{(2)}} (n^a)^2 (\cosh(2\mu) - 1)  .
\end{equation}

\subsubsection{Determining the inverse temperature and the chemical potentials} \label{subsubsec:determining_beta_mu}

In Eqs.~\eqref{eqn:prob_nats_one-site} and \eqref{eqn:prob_nats_two-site}, we have expressed the probability distributions of one-body and two-body observables based on the NATS in terms of the inverse temperature $\beta$ and the chemical potentials $\{\mu^a\}$. To determine these parameters, one would need to assume that the equilibrium behavior of local observables is well described by the \textit{global} NATS and then impose the global constraints of average energy and average charge conservation on it, as explained in Ref.~\cite{yunger_halpern_noncommuting_2020}. That is, $\Tr\left( \rho_\textrm{NATS}^\textrm{tot} H\right) \mbeq E$ and $\Tr\left( \rho_\textrm{NATS}^\textrm{tot} Q^a\right) \mbeq q^a$, where $\rho_\textrm{NATS}^\textrm{tot} \propto e^{-\beta H - \sum_a \mu^a Q^a}$ and $H$ and $\{Q^a\}$ are the global Hamiltonian and charges respectively. However, this is very computationally demanding, and the approximation used in Ref.~\cite{yunger_halpern_noncommuting_2020} is not applicable in our case, so we instead fit the analytical expressions we computed to the data. 
To determine the magnitude $\mu$ and unit vector $\hat{\vec{n}}$ of the vector of chemical potentials $\vec{\mu}$, we use $\vec{p}_1 - \vec{p}_0 = \tanh(\mu) \, \hat{\vec{n}}$, where $\vec{p}_j = (p^x_j, p^y_j, p^z_j)$. Thus, $\mu = \operatorname{arctanh}(|\vec{p}_1 - \vec{p}_0|)$ and $\hat{\vec{n}} = \frac{\vec{p}_1 - \vec{p}_0}{|\vec{p}_1 - \vec{p}_0|}$. Then, we use Eq.~\eqref{eqn:prob_nats_two-site} and the data for two-body observables to determine $\beta$ via numerical optimization.

\subsection{The coupling is non-weak} \label{app_subsec:non-weak-coupling}

A crucial assumption in the derivation of the local non-Abelian thermal state (NATS) is that of weak coupling. Here we provide evidence that for the model and the parameter regime considered, the coupling is not sufficiently weak under any of the definitions existing in the literature, at least for two-site subsystems. This limits the ability of the local NATS to predict the stationary behavior of local observables. One could instead take the partial trace of the global NATS. However, this requires diagonalizing the global Hamiltonian and global charges to be constructed, which rapidly becomes unfeasible due to the exponential increase of the Hilbert space with system size. \\

Given the Hamiltonian in Eq.~\eqref{eqn:hamiltonian_def_app}, for one-site subsystems the local Hamiltonian is $H_S^{(1)} = 0$ and the interaction term is $H_\textrm{int}^{(1)} = J \left[ \vec{\sigma}_{i-2} \cdot \vec{\sigma}_i + \vec{\sigma}_{i-1} \cdot \vec{\sigma}_i + \vec{\sigma}_{i} \cdot \vec{\sigma}_{i+1} + \vec{\sigma}_{i} \cdot \vec{\sigma}_{i+2} \right]$, where $i$ is the subsystem's lattice site, $\vec{\sigma}_i = (\sigma_i^x, \sigma_i^y, \sigma_i^z)^T$ is a vector of Pauli operators, and we take $J=1$. Note we are assuming for simplicity that we are sufficiently far from the ends of the spin chain such that the small perturbation $\epsilon$ does not play a role, and for all relevant lattice sites we have $J_i = J$. For two-site subsystems (with support on sites $i$ and $i+1$), we have instead $H_S^{(2)} = J \, \vec{\sigma}_i \cdot \vec{\sigma}_{i+1}$ and $H_\textrm{int}^{(2)} = J \left[ \vec{\sigma}_{i-2} \cdot \vec{\sigma}_i + \vec{\sigma}_{i-1} \cdot \vec{\sigma}_i + \vec{\sigma}_{i-1} \cdot \vec{\sigma}_{i+1} + \vec{\sigma}_i \cdot \vec{\sigma}_{i+2} + \vec{\sigma}_{i+1} \cdot \vec{\sigma}_{i+2} + \vec{\sigma}_{i+1} \cdot \vec{\sigma}_{i+3}\right]$. More generally, for a subsystem of size $n_S$ with support on sites $\{i_k\}_{k=1}^{n_S}$, $H_S^{(n_S)} = J \sum_{k=1}^{n_S-1} \vec{\sigma}_{i_k} \cdot \vec{\sigma}_{i_k + 1}$, while $H_\textrm{int}^{(n_S)} = H_\textrm{int}^{(2)} \quad \forall \, n_S > 1$.\\

A first na\"ive approach to investigate whether the coupling is weak is to compare the relevant energies (e.g., as in Ref.~\cite{yunger_halpern_noncommuting_2020}), which we here quantify using the operator norm. Via analytical or numerical calculations, we find that the operator norms of the Hamiltonians above are $\| H_S^{(1)} \| = 0$, $\| H_\textrm{int}^{(1)} \| = 6J$, $\| H_S^{(2)} \| = 3J$ and $\| H_\textrm{int}^{(2)} \| \approx 10 \,J$. For a general subsystem of size $n_S$ we have the upper bound $\| H_S^{(n_S)} \| \leq 3 (n_S - 1) J$. Noting that the fully polarized state along $z$ on the $n_S$ sites $|00\ldots0\rangle$ is always an eigenstate of $H_S^{(n_S)}$ with eigenvalue $+1$, we also find a lower bound: $\|H_S^{(n_S)}\| \geq (n_S - 1)J$, so this operator norm scales as $\|H_S^{(n_S)}\| \sim n_S J$. Comparing the operator norms of the interaction and subsystem Hamiltonians in both the one-site and two-site cases clearly indicates the coupling is not weak according to this definition: $\frac{\| H_\textrm{int}^{(1)} \|}{\| H_S^{(1)} \|} = \infty$ and $\frac{\| H_\textrm{int}^{(2)} \|}{\| H_S^{(2)} \|} \approx \frac{10}{3}$. In Ref.~\cite{yunger_halpern_noncommuting_2020}, it is remarked that the coupling becomes weak in the thermodynamic limit, meant as sending both $n_S, N \to \infty$, since the interaction Hamiltonian does not vary for $n_S > 2$ so $\frac{\| H_\textrm{int}^{(n_S)} \|}{\|H_S^{(n_S)}\|} \sim \frac{1}{n_S} \xrightarrow{n_S \to \infty} 0$. However, one often considers a different thermodynamic limit, in which the subsystem size is kept fixed while the total system size $N$ is increased, which is indeed what we do here and what Ref.~\cite{yunger_halpern_noncommuting_2020} itself does in its own numerics, to our best understanding. In this case, the ratio between the two operator norms remains constant. We also note that if $n_S \to \infty$, say as a fraction of the total system size $N$, then the complexity of constructing the local NATS also increases exponentially, even though it anyway remains relatively easier to compute than its global counterpart. \\

Alternatively, a common approach, rooted in perturbation theory, is to compare the operator norm of the interaction Hamiltonian with the minimum energy difference $\Delta E_S$ between non-degenerate energy levels of the subsystem's Hamiltonian \cite{tasaki_quantum_1998, reimann_canonical_2010, riera_thermalization_2012, trushechkin_open_2022}. In this case, $\Delta E_S^{(1)} = 0$ and $\Delta E_S^{(2)} = 4$, so we again see that the coupling cannot be considered weak also according to this definition: $\frac{\| H_\textrm{int}^{(1)} \|}{\Delta E_S^{(1)}} = \infty$ and $\frac{\| H_\textrm{int}^{(2)} \|}{\Delta E_S^{(2)}} \approx 2.5$. \\

Finally, one might consider the previous definitions too strong and not experimentally realistic \cite{reimann_canonical_2010, dong_quantum_2007, riera_thermalization_2012}. In this case, Ref.~\cite{riera_thermalization_2012} proposes the following more physical perturbation-theoretic conditions (see also Ref.~\cite[pg.63]{gogolin_equilibration_2016}): 
\begin{equation} \label{eqn:def_weak_coupling_riera_gogolin_eisert}
    \| H_\textrm{int} \| \ll \frac{1}{\beta} \ll \Delta E, 
\end{equation}
where $\Delta E$ is the width of the microcanonical window (the standard deviation of the global Hamiltonian) and, if the system thermalizes, $\beta$ is the inverse temperature. Using the $\beta$ obtained via fitting as explained in Sec.~\ref{subsubsec:determining_beta_mu}, in Fig.~\ref{fig:Check_weak_coupling_beta_norm_Hint_params_from_fit_two} we plot $\beta \| H_\textrm{int}^{(2)} \|$ against the initial state $\theta$, for a two-site subsystem and for $N = 12, 14, 16$. One-site subsystems are instead taken to have an effective infinite temperature (i.e., $\beta = 0$) due to $H_S^{(1)} = 0$. We believe this explains why the predictions based on the NATS work well for one-body observables but are less accurate than ours for two-body ones.

Moreover, we also compare $\| H_\textrm{int} \|$ with $\Delta E$, which we can compute exactly from Eq.~\eqref{eqn:energy_std_app}, namely, 
\begin{equation} \label{eqn:weak_coupling_ratio_Hint_dE}
    \frac{\| H_\textrm{int}^{(2)} \|}{\Delta E} \approx \frac{10}{4 \cos^2\theta [ (1 + 3\sin^2\theta)N
      - 5\sin^2\theta - 1 + 2\epsilon (1 + 3\sin^2\theta) + \epsilon^2 (1 + \sin^2\theta)]},
\end{equation}
where we have set $J=1$. 
In Fig.~\ref{fig:ratio_Hint_dE_vs_theta_two-site} we plot Eq.~\eqref{eqn:weak_coupling_ratio_Hint_dE} against the initial state parametrized by the angle $\theta$ for the system sizes considered in this study, i.e., $N=12, 14, 16$, and for $\epsilon=0.3$.
While these ratios are certainly lower than those found with the other methods, they still do not appear to be small enough to be in the weak coupling regime needed for the derivation of the local NATS. 
We conclude that indeed for the model and parameters considered the interaction is not sufficiently weak to justify the use of the local NATS to predict the stationary behavior of two-body observables.

\begin{figure}[ht]
    \centering
    \subfloat[Plot of $\beta \| H_\textrm{int}^{(2)} \|$ against the initial state parametrized by the angle $\theta$ for two-site subsystems. The inverse temperature $\beta$ is obtained by fitting to the data the analytical form of the probabilities based on the NATS, as explained in Sec.~\ref{subsubsec:determining_beta_mu}. \label{fig:Check_weak_coupling_beta_norm_Hint_params_from_fit_two}]{\includegraphics[width=0.45\textwidth]{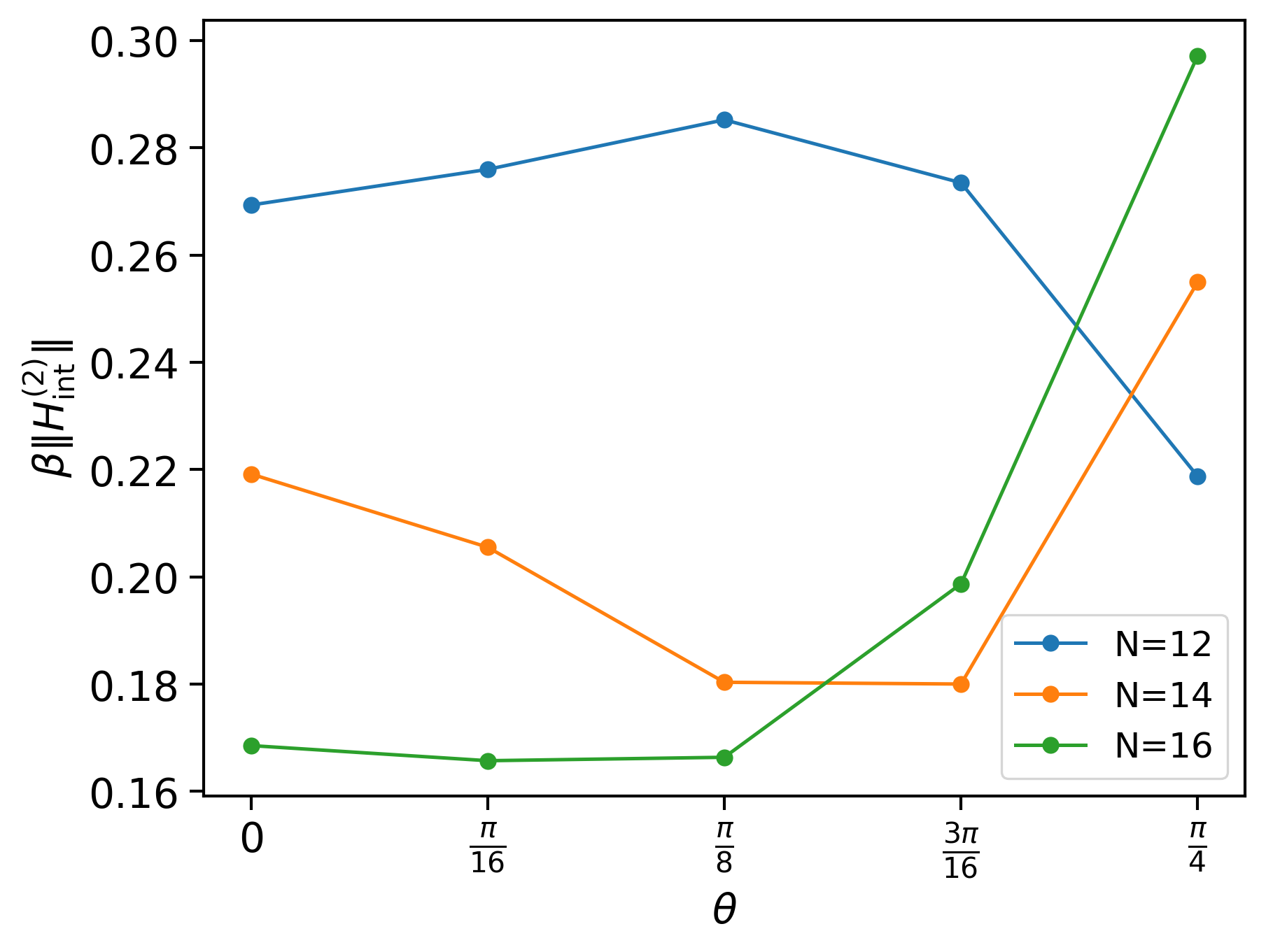}} \hspace{1cm}
    \subfloat[Plot of the ratio $\frac{\| H_\textrm{int}^{(2)} \|}{\Delta E}$ against the initial state parametrized by $\theta$, for two-site subsystems. $\Delta E$ is the standard deviation of the energy. \label{fig:ratio_Hint_dE_vs_theta_two-site}] {\includegraphics[width=0.45\textwidth]{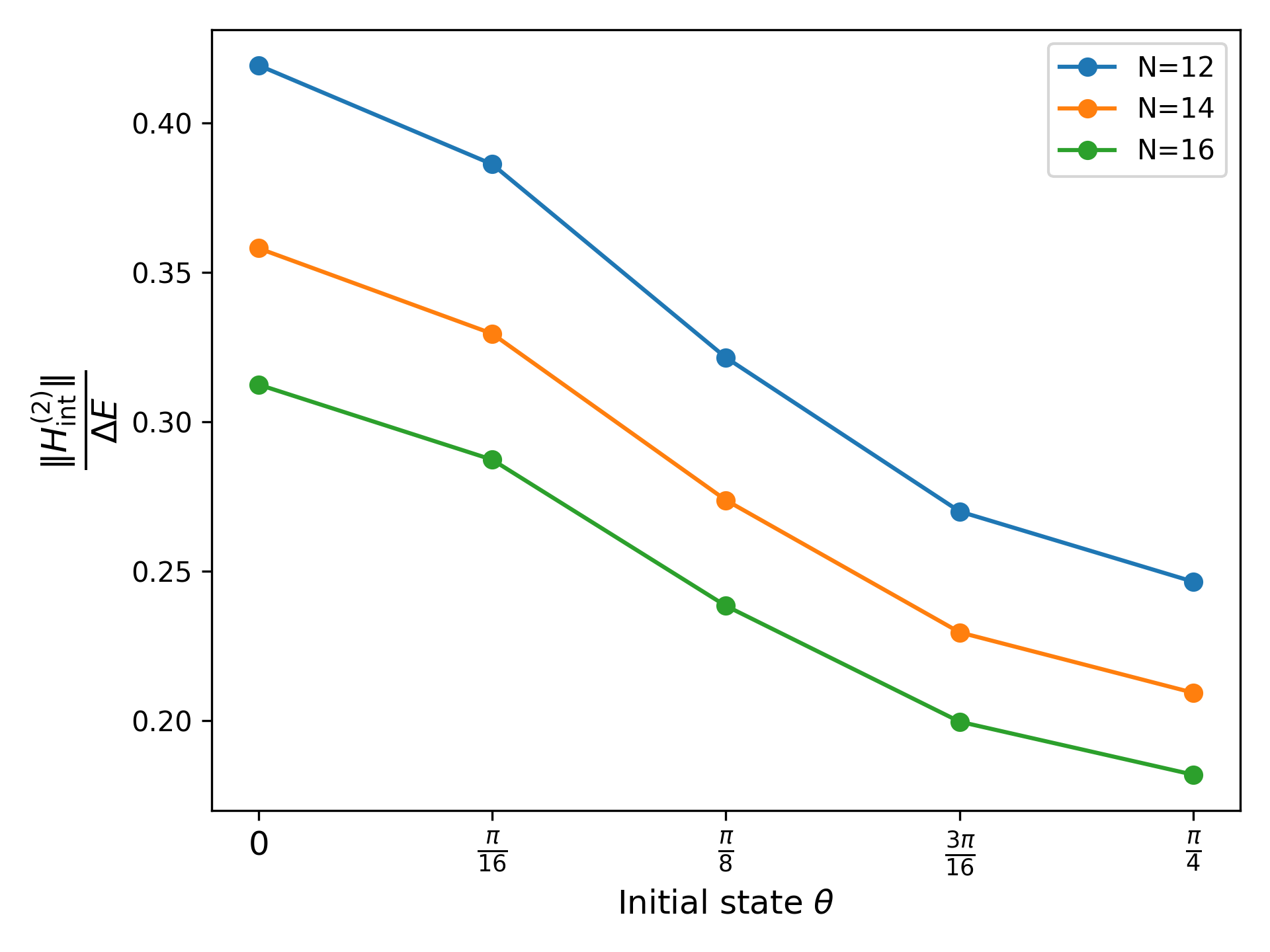}}
    \caption{We show that even according to the third definition of weak coupling [Eq.~\eqref{eqn:def_weak_coupling_riera_gogolin_eisert}] the interaction is not sufficiently weak to justify the applicability of the NATS. The operator norm of the two-site interaction Hamiltonian is $\| H_\textrm{int}^{(2)} \| \approx 9.8051$ (with $J=1$). The three curves in each subplot correspond to the three total system sizes considered, i.e., $N=12, 14, 16$. } \label{fig:coupling_non_weak_def_riera_gogolin_eisert}
\end{figure}

\end{document}